\documentclass[a4paper,UKenglish,cleveref, autoref, thm-restate]{lipics-v2021}

\nolinenumbers
\title{Bayesian Inference in Quantum Programs} 

\author{Christina Gehnen}{RWTH Aachen University, Germany }{christina.gehnen@cs.rwth-aachen.de}{https://orcid.org/0000-0002-6548-3432}{}

\author{Dominique Unruh}{RWTH Aachen University, Germany}{cqwp.uxhird@rwth.unruh.de}{https://orcid.org/0000-0001-8965-1931}{}

\author{Joost-Pieter Katoen}{RWTH Aachen University, Germany}{katoen@cs.rwth-aachen.de}{https://orcid.org/0000-0002-6143-1926}{}

\authorrunning{C. Gehnen, D. Unruh and J.-P. Katoen} 

\Copyright{Christina Gehnen, Dominique Unruh, Joost-Pieter Katoen} 



\funding{This work is supported by the ERC Grant 819317 CerQuS and the Interdisciplinary Doctoral Program in Quantum Systems Integration funded by the BMW group.}

\hideLIPIcs  

\EventEditors{Keren Censor-Hillel, Fabrizio Grandoni, Joel Ouaknine, and Gabriele Puppis}
\EventNoEds{4}
\EventLongTitle{52nd International Colloquium on Automata, Languages, and Programming (ICALP 2025)}
\EventShortTitle{ICALP 2025}
\EventAcronym{ICALP}
\EventYear{2025}
\EventDate{July 8--11, 2025}
\EventLocation{Aarhus, Denmark}
\EventLogo{}
\SeriesVolume{334}
\ArticleNo{146}

\usepackage{graphicx}
\usepackage{braket}
\usepackage{amsmath}
\usepackage{physics}
\usepackage{nicefrac}
\usepackage{amssymb}
\usepackage{hyperref}
\usepackage{stmaryrd}
\usepackage{color}
\usepackage{todonotes}
\usepackage[symbol]{footmisc}

\begin{document}
\newcommand*{\density}{\densitygen{\mathcal{H}} }
\newcommand*{\densityFull}{\mathcal{D}(\mathcal{H})}
\newcommand*{\densitygen}[1]{\mathcal{D}^- (#1)}
\newcommand*{\traceclass}{T(\mathcal{H})}
\newcommand*{\bounded}{B(\mathcal{H})}

\newcommand*{\predicate}{\predicategen{\mathcal{H}}}
\newcommand*{\predicategen}[1]{\mathcal{P}(#1)}

\newcommand*{\hoare}[3]{\{#1\} #2 \{#3\}}

\newcommand*{\semantics}[1]{\llbracket #1 \rrbracket}
\newcommand*{\semanticsOriginal}[1]{\semantics{#1}_{og}}
\newcommand*{\semanticsRho}[1]{\semantics{#1}_{\rho}}
\newcommand*{\semanticsErr}[1]{\semantics{#1}_{\lightning}}
\newcommand*{\semanticsRhoSingle}[1]{\semantics{#1}_{\tilde{\rho}}}

\newcommand*{\R}{\mathbb{R}}
\newcommand*{\C}{\mathbb{C}}

\newcommand*{\skipbf}{\mathbf{skip}}
\newcommand*{\qzero}{q:=0}
\newcommand*{\Uq}{\overline{q}:= U\overline{q}}
\newcommand*{\observe}{\mathbf{observe } \text{ }(\Bar{q},O)}
\newcommand*{\measure}{\mathbf{measure } \text{ }M[\Bar{q}]:\Bar{S}}
\newcommand*{\measurePrime}{\mathbf{measure } \text{ }M[\Bar{q}]:\Bar{S'}}
\newcommand*{\concat}{S_1;S_2}
\newcommand*{\while}{\mathbf{while } \text{ }M[\Bar{q}]=1 \text{ }\mathbf{ do } \text{ }S}
\newcommand*{\ifstatement}[2]{\mathbf{if} \text{ }M[\Bar{q}]=1 \text{ } \mathbf{then} \text{ } #1 \text{ } \mathbf{else} \text{ } #2}

\newcommand*{\qwp}[2]{qwp\llbracket #1 \rrbracket (#2)}
\newcommand*{\qwlp}[2]{qwlp\llbracket #1 \rrbracket (#2)}
\newcommand*{\qcwp}[2]{qcwp\llbracket #1 \rrbracket (#2)}
\newcommand*{\qcwlp}[2]{qcwlp\llbracket #1 \rrbracket (#2)}

\newcommand*{\intType}{\mathit{Int}}
\newcommand*{\boolType}{\mathit{Bool}}

\newcommand*{\emptyProgram}{\downarrow}
\newcommand*{\config}[2]{\langle #1 , #2 \rangle}
\newcommand*{\errConfig}{\langle \lightning \rangle}
\newcommand*{\setConfig}[0]{\mathcal{C}(\mathcal{H})}
\newcommand*{\termConfig}[0]{\langle sink \rangle} 

\newcommand*{\rewardMC}[3]{\mathfrak{R}_{#1}^{#2}\llbracket #3 \rrbracket}
\newcommand*{\operationalMC}[2]{\mathfrak{R}_{#1}\llbracket #2 \rrbracket}
\newcommand*{\prob}[2]{Pr^{#1}(#2)}

\newcommand*{\eventually}[0]{\lozenge}

\newcommand*{\paths}{\mathit{Paths}}
\newcommand*{\pathsFin}{\paths_{\mathit{fin}}}
\newcommand*{\mc}{\mathcal{M}}
\newcommand*{\densityNumberPairs}{\mathcal{D}\mathcal{R}}
\newcommand*{\states}{\Sigma}

\newcommand*{\fracTr}{\hat{tr}}
\newcommand*{\df}{:=}

\newcommand*{\identityOp}{\textbf{I}}
\newcommand*{\identityOpVar}[1]{\identityOp_{#1}}
\newcommand*{\zeroOp}{\mathbf{0}}
\newcommand*{\zeroOpVar}[1]{\zeroOp_{#1}}

\newcommand*{\varSet}{\mathit{Var}}
\newcommand*{\ext}{\mathit{span}}

\newcommand*{\multdot}{\cdot}
\newcommand*{\plusdot}{+}

\maketitle

\begin{abstract}
    Conditioning is a key feature in probabilistic programming to enable modeling the influence of data (also known as observations) to the probability distribution described by such programs. Determining the posterior distribution is also known as Bayesian inference. This paper equips a quantum while-language with conditioning, defines its denotational and operational semantics over infinite-dimensional Hilbert spaces, and shows their equivalence. We provide sufficient conditions for the existence of weakest (liberal) precondition-transformers and derive inductive characterizations of these transformers. It is shown how w(l)p-transformers can be used to assess the effect of Bayesian inference on (possibly diverging) quantum programs.
\end{abstract}
\tableofcontents
\section{Introduction}
\label{sec:introduction}
Quantum verification is a important part of the rapidly evolving field of quantum computing and information. The importance comes from several factors. Firstly, quantum computers operate in a completely different way than classical computers do. Principles of quantum mechanics are important to algorithm designers but in general unintuitive to most people. This leads to a higher risk of introducing logical errors. Secondly, quantum algorithms are often used in safely critical areas such as cryptography and optimization where those mistakes can lead to serious issues.
Classical testing and debugging methods do not directly apply to quantum computing. Testing on quantum computers is challenging due to high execution costs, probabilistic outcomes, and noise from environmental interactions. While simulators help, they have limitations such as scalability. Debugging is also difficult, as measuring quantum variables alters their state, preventing traditional inspection methods.

Testing only verifies specific inputs without guaranteeing overall correctness, whereas formal verification ensures correctness for all inputs.
Weakest preconditions define input states that ensure a given postcondition holds after execution. Inspired by the importance of conditioning and Bayesian inference in probabilistic programs, we extend the calculus from \cite{DHondtWeakestPreconditions,floydHoareLogic} to incorporate ``observations''. Combining weakest preconditions for total correctness and weakest liberal preconditions for partial correctness, we determine whether a predicate holds assuming all observations hold, i.e, compute a conditional probability.

This new statement could be used to aid in debugging to locate logical mistakes. Assume having a theoretical algorithm and a wrong implementation. To figure out which parts are wrong, fixing variable values by observations can help identify errors by comparing implementation samples with the expected distribution. Similarly, given a complex (possibly wrong) algorithm, adding observations can help understanding parts of the algorithm by comparing it to its intuitive understanding. For instance, in a random walk algorithm with a random starting point, analyzing success probability from a ``good'' starting point can help to understand the algorithm.
Unlike traditional assertions, observations can be useful even when they don’t always hold.
Another possible application is error correction, where outputs are often analysed assuming no more than $t$ qubit errors occurred per step to ensure successful error correction.

\paragraph*{Related Work}
The general idea of weakest preconditions was first developed by Dijkstra for classical programs \cite{Dijkstra76,Dijkstra75}, then for probabilistic programs \cite{KOZEN1985162,McIverWpProb} and later for quantum programs \cite{DHondtWeakestPreconditions}. D'Hondt and Panangaden \cite{DHondtWeakestPreconditions} defined predicates as positive operators as we do and focused on total correctness and finite-dimensional Hilbert spaces. \cite{floydHoareLogic} extended this approach to partial correctness and gave an explicit representation of the predicate transformer for the quantum while-language.
An alternative to define predicates is to use projections \cite{ZhouAppliedQHL}. There have been several extensions like adding classical variables \cite{DENG202273, FengQHLClassicalVars} or non-determinism \cite{FengNondeterministicQuantumVerification}.

A runtime assertion scheme using projective predicates for testing and debugging has been introduced in \cite{projectionAssertions}. In contrast, our approach enables debugging, but in addition provides formal guarantees on the correctness based on the satisfaction of assertions and allows infinite-dimensional Hilbert spaces. A survey about studies and approaches of debugging of quantum programs is given in \cite{NeedOfToolsDebuggingQuantumPrograms}.
Another idea to locate bugs is to use incorrectness logic with projective predicates \cite{YanIncorrectnessLogic}. The idea of conditional weakest preconditions has been introduced in \cite{Nori, conditioningProb} for probabilistic programs.

The concept of choosing specific measurement outcomes is also known as \emph{postselection}. \cite{postbqp} shows that the class of problems solvable by quantum programs with postselection in polynomial time, called Postselected Bounded-Error Quantum Polynomial-Time (PostBQP), is the same as the ones in the complexity class Probabilistic Polynomial-Time (PP). This equivalence is shown by solving a representative PP-complete problem, MAJ-SAT, using a quantum program with postselection. We confirm the correctness of this program in Section~\ref{sec:example} by using conditional weakest preconditions.

\paragraph*{Main Contributions}

\begin{itemize}
    \item Conditional weakest-precondition transformers: We define a weakest precondition calculus for reasoning about programs with an ``observe'' statement. The conditional weakest precondition, defined in terms of weakest (liberal) preconditions transformers, reveals the probability of a postcondition given all observations succeed.

    The definition of the transformers is semantic, i.e., formulated in a generic way based on the denotational semantics and not tied to a specific syntax of programs (but we also give explicit rules for our syntax by recursion over the structure of a program).

    \item Semantics: We develop both denotational and operational semantics of a simple quantum while-language with ``observe'' statements and show their equivalence.

    \item Our definition of weakest (liberal) preconditions is a conservative extension of \cite{floydHoareLogic}, supporting ``observe'' statements.
    Further differences include: Our definition is semantic and we support infinite-dimensional quantum systems (e.g., to support quantum integers)\footnote{Notice that \cite{floydHoareLogic} also defines a language with quantum integers. However, they do not explicitly specify the various notions of convergence of operators (e.g., operator topologies, convergence of infinite sums, existence of suprema), making it difficult to verify whether their rules are sound in the infinite-dimensional case.}.
\end{itemize}

\paragraph*{Structure}
We first recall important definitions in Section~\ref{sec:preliminaries}. The main contributions are in Section~\ref{sec:syntaxSemantics} and Section~\ref{sec:wp}: Section~\ref{sec:syntaxSemantics} introduces the ``observe'' statement and its semantics whereas Section~\ref{sec:wp} defines weakest (liberal) preconditions and finally conditional weakest (liberal) preconditions.
Two examples in Section~\ref{sec:example} illustrate our approach, followed by conclusions in Section~\ref{sec:conclusion}.
\section{Preliminaries}
\label{sec:preliminaries}

\subsection{Hilbert Spaces}
Let $\langle \cdot \mid \cdot \rangle$ denote the {inner product} over a vector space $\mathcal{V}$. The \emph{norm (or length) of a vector} $u$, denoted $\Vert u \Vert$, is defined as $\sqrt{\langle u \mid u \rangle}$. The vector $u$ is called a unit vector if $\Vert u \Vert = 1$. Vectors $u,v$ are \emph{orthogonal} ($u \bot v$) if $\langle u \mid v \rangle = 0$. 
The sequence $\{u_i\}_{i \in \mathbb{N}}$ of vectors $u_i \in \mathcal{V}$  is a \emph{Cauchy sequence}, if for any $\epsilon>0$, there exists a positive integer $N$ such that $\Vert u_n - u_m \Vert < \epsilon$ for all $n,m\geq N$. If for any $\epsilon >0$, there exists a positive integer $N$ such that $\Vert u_n - u \Vert < \epsilon$ for all $n \geq N$, then $u$ is the limit of $\{u_i\}_{i \in \mathbb{N}}$, denoted $u = \lim_{i \to \infty} u_i$.

A family $\{u_i\}_{i \in I}$ of vectors in $\mathcal{V}$ is \emph{summable} with the sum $v = \sum_{i \in I}u_i$ if for every $\epsilon>0$ there exists a finite $J \subseteq I$ such that $\Vert v - \sum_{i \in K} u_i \Vert < \epsilon$ for every finite $K \subseteq I$ and $J \subseteq K$.

A \emph{Hilbert space} $\mathcal{H}$ is a complete inner product space, i.e, every Cauchy sequence of vectors in $\mathcal{H}$ has a limit \cite{floydHoareLogic}.
An orthonormal \emph{basis} of a Hilbert space $\mathcal{H}$ is a (possibly infinite) family $\{u_i\}_{i \in I}$ of unit vectors if they are pairwise orthogonal (i.e., $u_i \bot u_j$ for $i\neq j, i,j\in I$) and every $v\in \mathcal{H}$ can be written as $v = \sum_{i \in I} \langle u_i \mid v\rangle \cdot u_i$ (in the sense above). The cardinality of $I$, denoted $\abs{I}$, is the dimension of $\mathcal{H}$. Hilbert spaces and its elements can be combined using the \emph{tensor product} $\otimes$ \cite[Def. IV.1.2]{takesaki1979theory}.

We use Dirac notation $\ket{\phi}$ to denote vectors of a vector space where $\bra{\phi}$ is the dual vector of $\ket{\phi}$ \cite{Dirac}, i.e., $\bra{\phi}=\ket{\phi}^\dagger$.

\begin{example}
    A typical Hilbert space over the set $X$ is
    \begin{align*}
        l^2(X) = \{\sum_{n\in X} \alpha_n \ket{n}\mid \alpha_n \in \mathbb{C} \text{ for all }n\in X\text{ and } \sum_{n\in X} \abs{\alpha_n}^2 < \infty\}
    \end{align*}
    where the inner product is defined as $
        (\sum_{n\in X} \alpha_n \ket{n}, \sum_{n\in X}\alpha'_n \ket{n}) = \sum_{n\in X} \overline{\alpha_n} \alpha'_n$. By $\overline{x+yi}=x-yi$ we denote the complex conjugate of $x+yi \in \mathbb{C}$.
    An orthonormal basis, also called \emph{computational basis}, is $\{\ket{n}\mid n\in X\}$.
    For (countably) infinite sets $X$, the basis is (countably) infinite and thus $l^2(X)$ is a (countably) infinite Hilbert space.
    $l^2(\mathbb{Z})$ can be used for quantum integers and is also denoted by $\mathcal{H}_\infty$.
    For qubits, we use $l^2(\{0,1\})$ and denote it as $\mathcal{H}_2$.
\end{example}

\subsection{Operators}
In the following, all vector spaces will be over $\mathbb{C}$.
For vector spaces $\mathcal{V}, \mathcal{W}$, a function $f:\mathcal{V}\to \mathcal{W}$ is called \emph{linear} if $f(ax+y)=af(x)+f(y)$ for $x,y\in \mathcal{V}$ and $a\in \mathbb{C}$. If $\mathcal{V}, \mathcal{W}$ are normed vector spaces then $f$ is called \emph{bounded linear} if $f$ is linear and $\norm{f(x)} \leq c \cdot \norm{ x}$ for some constant $c \geq 0$ for all $x \in \mathcal{V}$.
If $\mathcal{H}$ is a Hilbert space, we call bounded linear functions on $\mathcal{H} \to \mathcal{H}$ \emph{operators}. Let $\bounded$ denote the space of all operators on $\mathcal{H}$ and $A\ket{\phi}$ the result of applying operator $A$ to $\ket{\phi}\in \mathcal{H}$.
For this work, we additionally generalize the notion of linearity to functions that are defined on subsets of the vector space:
For (normed) vector spaces $S\subseteq \mathcal{V},T\subseteq \mathcal{W}$ with $\ext(S)=\mathcal{V}$ and $\ext(T)=\mathcal{W}$, we call $f: S \to T$ \emph{(bounded) linear} iff there exists a (bounded) linear function $\Bar{f}: \mathcal{V} \to \mathcal{W}$ such that $\Bar{f}(s) = f(s)$ for $s \in S$. $\ext(S)$ includes all finite linear combinations of $S$.

Let $A$ and $B$ be operators on $\mathcal{H}_1$ and $\mathcal{H}_2$ with $\ket{\phi} \in \mathcal{H}_1, \ket{\psi} \in \mathcal{H}_2$. By \cite[Def. IV.1.3]{takesaki1979theory}, the tensor product $A \otimes B$ is the unique operator that satisfies $(A \otimes B)( \ket{\phi }\otimes \ket{\psi}) = A \ket{\phi} \otimes B\ket{\psi}$
For matrices, the tensor product is also called the \emph{Kronecker product}.

For every operator $A$ on $\mathcal{H}$, there exists an operator $A^\dagger$ on $\mathcal{H}$ with $\langle\ket{\phi}, A\ket{\psi}\rangle=\langle A^\dagger\ket{\phi},\ket{\psi}\rangle$ for all $\ket{\phi},\ket{\psi}\in \mathcal{H}$.
An operator $A$ on $\mathcal{H}$ is called \emph{positive} if $\bra{\psi}A\ket{\psi}\geq 0$ for all states $\ket{\psi}\in \mathcal{H}$ \cite{Nielsen_Chuang_2010}.
The \emph{identity operator} $\identityOpVar{\mathcal{H}}$ on $\mathcal{H}$ is defined by $\identityOpVar{\mathcal{H}} \ket{\phi}=\ket{\phi}$. The \emph{zero operator} on $\mathcal{H}$, denoted by $\zeroOpVar{\mathcal{H}}$, maps every vector to the zero vector. We omit $\mathcal{H}$ if it is clear from the context.
An \textit{unitary operator} $U$ is an operator such that its inverse is its adjoint $U^{-1}=U^{\dagger}$, i.e., $U^\dagger U = \identityOp$ and $U U^\dagger = \identityOp$ \cite{verificationSummary}.
An \emph{(ortho)projector} is an operator $P: \mathcal{H} \to \mathcal{H}$ such that $P^2 = P=P^\dagger$. For every closed subspace $S$, there exists a projector $P_S$ with image $S$ \cite[Prop. II.3.2 (b)]{conway1994}.

An operator $A$ is a \emph{trace class} operator if there exists an orthonormal basis $\{\ket{\psi_i}\}_{i \in I}$ such that $\{\bra{\psi_i} \cdot \abs{A} \cdot \ket{\psi_i} \}_{i \in I}$ is summable where $\abs{A}$ is the unique positive operator $B$ with $B^\dagger B = A^\dagger A$. Then the trace of $A$ is defined as $tr(A) = \sum_{i \in I} \bra{\psi_i} \cdot A \cdot  \ket{\psi_i}$
where $\{\ket{\psi_i}\}_{i \in I}$ is an orthonormal basis. For a trace class operator $A$, it can be shown that $tr(A)$ is independent of the chosen base \cite{floydHoareLogic}. The trace is cyclic, i.e., $tr(AB)=tr(BA)$ \cite{heisenbergdualityUnruh}, linear, i.e., $tr(A+B) = tr(A) + tr(B)$, scalar, i.e., $tr(cA) = c\cdot tr(A)$ for a constant $c$ \cite{conway2000courseOperator} and multiplicative, i.e., $tr(A \otimes B)= tr(A)tr(B)$ holds \cite{heisenbergdualityUnruh} for trace class operators $A,B$.
We use $\traceclass$ to denote the space of trace class operators on $\mathcal{H}$.
Positive trace class operators with $tr(\rho)\leq 1$ are called \textit{partial density operators}. The set of partial density operators is denoted $\density$ with $\ext(\density)=\traceclass$.
\textit{Density operators} are partial density operators with $tr(\rho) = 1$. They are denoted as $\densityFull$. The \emph{support} of a partial density operator $\rho$ is the smallest closed subspace $S$ such that $P_S \rho P_S = \rho$.

Let us consider some properties of functions that map operators to operators.
$f: T_1 \to T_2$ with $T_1 \subseteq T(\mathcal{H}_1), T_2 \subseteq T(\mathcal{H}_2)$ is \emph{trace-reducing} if $tr(f(\rho)) \leq tr(\rho)$ for all positive $\rho \in T_1$. $f: B_1 \subseteq B(\mathcal{H}_1) \to B_2 \subseteq B(\mathcal{H}_2) $ is \emph{positive} if $f(a) $ is positive for positive $a \in B_1$ and \emph{subunital} if $f(\identityOpVar{\mathcal{H}_1}) \sqsubseteq \identityOpVar{\mathcal{H}_2}$ and $\identityOpVar{\mathcal{H}_1} \in B_1$, where $\sqsubseteq$ is defined just below.

\subsubsection{The Loewner Partial Order}
To order operators, the \textit{Loewner partial order} is used. For any operators $A,B$, it is defined by $A\sqsubseteq B$ iff $B-A$ is a positive operator. This is equivalent to $tr(A\rho)\leq tr(B \rho)$ for all partial density operators $\rho \in \density$ \cite{floydHoareLogic}.
The Loewner order is compatible w.r.t. addition (also known as monotonic), i.e., $A \sqsubseteq B$ implies $A+C \sqsubseteq B+C$ for any $C$, and w.r.t. multiplication of non-negative scalars, i.e., $A\sqsubseteq B$ implies $c A \sqsubseteq c B$ for $c \geq 0 $ \cite{boyd2004convex}.

Using this order, we can define predicates \cite{DHondtWeakestPreconditions}. A \emph{quantum predicate} on a Hilbert space $\mathcal{H}$ is defined as an operator $P$ on $\mathcal{H}$ with $\zeroOpVar{\mathcal{H}}\sqsubseteq P \sqsubseteq \identityOpVar{\mathcal{H}}$. The set of quantum predicates on $\mathcal{H}$ is denoted by $\predicategen{\mathcal{H}}$ and $\ext(\predicate) = \bounded$.

The Loewner partial order is an $\omega$-complete partial order ($\omega$-cpo) on the set of partial density operators \cite{YingPredicateTranserformerSemantics}. Thus each increasing sequence of partial density operators has a least upper bound.
This also holds for the set of predicates \cite{DHondtWeakestPreconditions}.

An important property that we need is continuity of the trace operator. First of all, we note that the trace-operator is order-continuous on partial density operators with respect to $\sqsubseteq$, i.e., $\bigvee_{i \in \mathbb{N}} tr(\rho_i) = tr( \bigvee_{i \in \mathbb{N}} \rho_i)$ for any increasing sequence of partial density operators $\{\rho_i\}_{i \in \mathbb{N}}$. Without going further into details, this holds because for an increasing sequence of real numbers, the least upper bound and the limit coincide, the same also holds for partial density operators \cite{heisenbergdualityUnruh} and because the trace is linear and bounded, it is also trace-norm continuous.
Continuity w.r.t. predicates means $\bigvee_{i \in \mathbb{N}} tr(P_i \rho) = tr((\bigvee_{i \in \mathbb{N}} P_i) \rho)$ for every $\rho \in \density$ and increasing sequence of predicates $\{P_i\}_{i \in \mathbb{N}}$. Without going into further detail, we can show that a function $f: \mathcal{B}(\mathcal{H}) \to \mathbb{C}$ defined by $f(A)=tr(A\rho)$ for a fixed $\rho \in \density$ is weak$^*$-continuous and convergence of positive bounded operators in the weak$^*$-topology coincides with the supremum \cite{heisenbergdualityUnruh}. Similar, the same property holds for decreasing sequences of predicates $\{P_i\}_{i \in \mathbb{N}}$ and the greatest lower bound $\bigwedge_{i \in \mathbb{N}} P_i$.

\subsection{Quantum-specific Preliminaries}
Due to a postulate of quantum mechanics, the state space of an isolated quantum system can be described as a Hilbert space where states correspond to unit vectors (up to a phase shift) in its state space \cite{floydHoareLogic}. A quantum state is called \textit{pure} if it can be described by a vector in the Hilbert space; otherwise \textit{mixed}, i.e., it is a probabilistic distribution over pure states.
We use partial density operators to describe mixed states, in particular to capture the current state of a program.
If a quantum system is in a pure state $\ket{\psi_i}$ with probability $p_i$ (with $\sum_i p_i \leq 1$), then this is represented by the partial density operator $\rho=\sum_i p_i \ket{\psi_i}\bra{\psi_i}$.

To obtain the current value of e.g. a quantum variable, we cannot simply look at it. In quantum mechanics, each measurement can impact the current state of a qubit.

A \emph{measurement} is a (possible infinite) family of operators $\{M_m\}_{m\in I}$ where $m$ is the measurement outcome and $\sum_{m\in I} M_m^\dagger M_m = \identityOp$ \footnote[1]{As in \cite{heisenbergdualityUnruh}, we mean convergence of sums with respect to SOT (strong operator topology) which is the topology where $\lim_{i\to \infty} a_i = a$ holds iff for all $\phi$: $\lim_{i\to \infty} a_i \phi = a \phi$ \cite[Prop. IX.1.3(c)]{conway1994}.}.
If the quantum system is in state $\rho \in \density$ before the measurement $\{M_m\}$, then the probability for result $m$ is $p(m)=tr( M_m \rho M^\dagger_m)$ and the post-measurement state is $\rho_m = \frac{M_m\rho M^\dagger_m}{p(m)}$.
An important kind of measurement is the \textit{projective measurement}. It is a set of projections $\{P_m\}$ over $\mathcal{H}$ with $\sum_m P_m =\identityOp$.
An important property of projective measurements is that if a state $\rho$ is measured by a projective measurement $\{P,I-P\}$ and $supp(\rho)\subseteq P$ holds, then $\rho$ is not changed.

\subsection{Markov Chains}

A \emph{Markov chain} (MC) is a tuple $\mc = (\states,\textbf{P},s_{init})$ where
    \begin{itemize}
        \item $\states$ is a nonempty (possibly uncountable) set of states,
        \item $\textbf{P}:\states \times \states \rightarrow [0,1]$ with $\sum_{s' \in \states}\textbf{P}(s,s')=1$ is the transition probability function. Let $s \overset{p}{\rightarrow} s'$ denote $\textbf{P}(s,s')=p$.
        \item $s_{init}\in \states$ is the initial state.
    \end{itemize}
Note that in comparison to \cite{mcBible,conditioningProb}, $\states$ can be uncountable. However, in our setting the reachable set of states will be countable as every state $s$ can only have a countable number of successor states $s'$ with $\textbf{P}(s,s')>0$. Therefore, even if $\states$ is uncountable, the set of reachable states is countable and all results from \cite{mcBible,conditioningProb} still apply.

A path of a MC $\mc$ is an infinite sequence $s_0 s_1 s_2 \ldots \in \states^\omega$ with $s_0 = s_{init}$ and $\textbf{P}(s_i ,s_{i+1})>0$ for all $i$. We use $\paths(\mc)$ to denote the set of paths in $\mc$ and $\pathsFin(\mc)$ for the finite path prefixes. If it is clear from the context, we omit $\mc$.
The probability distribution $Pr^\mc$ on $\paths(\mc)$ is defined using cylinder sets as in \cite{mcBible}.
In a slight abuse of notation, we write $\prob{\mc}{\hat{\pi}}$ for $\prob{\mc}{Cyl(\hat{\pi})}$ for $\hat{\pi}\in \pathsFin(\mc)$ where $Cyl(\hat{\pi})$ denotes the cylinder set of $\hat{\pi}$. We write $s_0 \rightarrow_{p}^* s_n$ where $p = \sum_{s_0 \ldots s_n \in \pathsFin(\mc)} \textbf{P}(s_0 \ldots s_n)$ is the probability to reach $s_n$ from $s_0$.
Given a target set of reachable states $T\subseteq \states$, let $\eventually T$ be the (measurable) set of infinite paths that reach the target set $T$.
The probability of reaching $T$ is $\prob{\mc}{\eventually T} = \sum_{\hat{\pi} \in \pathsFin(\mc) \cap (\states\backslash T)^* T} \prob{\mc}{\hat{\pi}}$.
Analogously, let $\neg \eventually T$ be the set of paths that never reach $T$; $\prob{\mc}{\neg \eventually T} = 1 - \prob{\mc}{\eventually T}$.

\section{Quantum Programs with Observations}
\label{sec:syntaxSemantics}

We assume $\varSet$ to be a finite set of quantum variables with two types: Boolean and integer. As in \cite{floydHoareLogic}, the corresponding Hilbert spaces are
\begin{align*}
    &\mathcal{H}_2 = \{\alpha \ket{0}+\beta \ket{1} \mid \alpha,\beta \in \mathbb{C}\}, \\
&\mathcal{H}_\infty = \Big \{\sum_{n\in \mathbb{Z}} \alpha_n \ket{n}\mid \alpha_n \in \mathbb{C} \text{ for all }n\in \mathbb{Z}\text{ and } \sum_{n\in \mathbb{Z}} \abs{\alpha_n}^2 < \infty\Big \}.
\end{align*}
Each variable $q\in \varSet$ has a type $type(q)\in \{\boolType, \intType\}$. Its state space $\mathcal{H}_q$ is $\mathcal{H}_2$ if $type(q)=\boolType$ and $\mathcal{H}_{\infty}$ otherwise.
The state space of a quantum register $\Bar{q} = q_1, ..., q_n$ is defined by the tensor product $\mathcal{H}_{\Bar{q}} =\bigotimes_{i=1}^n \mathcal{H}_{q_i}$ of state spaces of $q_1$ through $q_n$.

\subsection{Syntax}
A quantum while-program has the following syntax:
{\small
\begin{equation*}
    S::=\skipbf \mid \qzero \mid \Uq \mid \observe \mid \concat \mid \measure \mid \while
\end{equation*}}
where
\begin{itemize}
    \item $q$ is a quantum variable,
    \item $\Bar{q}$ is a quantum register,
    \item $U$ from statement $\Uq$ is a unitary operator on $\mathcal{H}_{\Bar{q}}$ and $\Bar{q}$ is the same on both sides,
    \item $O$ in $\observe$ is a projection on $\mathcal{H}_{\Bar{q}}$
    \item the measurement $M=\{M_m\}_{m\in I}$ in $\measure$ is on $\mathcal{H}_{\Bar{q}}$ and $\Bar{S}=\{S_m\}_{m \in I}$ is a family of quantum programs where each $S_m$ corresponds to an outcome $m\in I$,
    \item the measurement in $\while$ on $\mathcal{H}_{\Bar{q}}$ has the form $M=\{M_0, M_1\}$.
\end{itemize}
Our programs extend \cite{floydHoareLogic} with the new statement $\observe$. We only allow projective predicates $O$ for observations. It is conceivable that it can also be based on more general predicates $O\in \predicate$ but it is not clear what the intuitive operational meaning of such $O$ would be, so we choose to pursue the simpler case.

We use $\ifstatement{S_1}{S_0}$ as syntactic sugar for a measurement statement with $M=\{M_0,M_1\}$ and $\Bar{S}=\{S_0,S_1\}$.

By $\equiv$ we denote syntactic equality of quantum programs.
We use $\mathit{var}(S)$ to denote the set of variables occurring in program $S$. The Hilbert space of $\mathit{var}(S)$ is denoted by $\mathcal{H}_{\mathit{all}}$. If the set of variables is clear from the context, we just write $\mathcal{H}$.

For $\Bar{q}=q_1,...,q_n$ and operator $A$ on $\mathcal{H}_{\Bar{q}}$, we define its cylinder extension by $A \otimes I_{Var\backslash \{\Bar{q}\}}$ on $\mathcal{H}_{all}$ and abbreviate it by $A$ if it is clear from the context.
Let $\ket{\phi} \bra{\psi}_q$ denote the value of quantum variable $q$ in the state $\ket{\phi} \bra{\psi}$. We sometimes refer to it meaning its cylinder extension on $\mathcal{H}_{\mathit{all}}$ \cite{floydHoareLogic}. This notation is equivalent to $q(\ket{\phi}\bra{\psi})$ in \cite{heisenbergdualityUnruh}.

\subsection{Semantics}
In this section, we define an operational and denotational semantics for quantum while-programs with observations and show their equivalence.
\begin{figure}
    \begin{tabular}{c}
    \begin{tabular}{@{}c@{}}
    \hline
    $\config{\skipbf}{\sigma} \overset{1}{\rightarrow} \config{\emptyProgram}{\sigma}$
    \end{tabular}  \hspace{1cm}
    \begin{tabular}{@{}c@{}}
    $type(q)=\intType\land \sigma' = \sum_{n\in \mathbb{Z}} \ket{0}\bra{n}_q\sigma \ket{n}\bra{0}_q$\\
    \hline
    $\config{\qzero}{\sigma}\overset{1}{\rightarrow} \config{\emptyProgram}{\sigma'}$
    \end{tabular} \vspace{0.3cm} \\

    \begin{tabular}{@{}c@{}}
        $type(q)=\boolType\land \sigma' =\ket{0} \bra{0}_q\sigma \ket{0} \bra{0}_q + \ket{0}\bra{1}_q\sigma \ket{1} \bra{0}_q$\\
        \hline
        $\config{\qzero}{\sigma}\overset{1}{\rightarrow} \config{\emptyProgram}{\sigma'}$
        \end{tabular}\hspace{1cm}
        \begin{tabular}{@{}c@{}}
            \\
            \hline
            $\config{\Uq}{\sigma} \overset{1}{\rightarrow} \config{\emptyProgram}{U\sigma U^\dagger}$
            \end{tabular}\vspace{0.3cm} \\

    \begin{tabular}{@{}c@{}}
    $tr(O \sigma O^\dagger)>0$ \\
    \hline
    $ \config{\observe}{\sigma} \overset{tr(O \sigma O^\dagger)}{\rightarrow} \config{\emptyProgram}{\frac{O\sigma O ^\dagger}{tr(O \sigma O^\dagger)}}$
    \end{tabular} \hspace{1cm}
    \begin{tabular}{@{}c@{}}
    $tr(O \sigma O^\dagger)<1$ \\
    \hline
    $\config{\observe}{\sigma}\overset{1-tr(O \sigma O^\dagger)}{\rightarrow} \errConfig$
    \end{tabular} \vspace{0.3cm} \\

    \begin{tabular}{@{}c@{}}
            $M=\{M_m\}_{m \in I} \land m\in I \land tr(M_m \sigma M_m^\dagger)>0$\\
            \hline
            $\config{\measure}{\sigma} \overset{ tr(M_m \sigma M_m^\dagger)}{\rightarrow} \config{S_m}{\frac{M_m \sigma M_m^\dagger}{tr(M_m \sigma M_m^\dagger)}}$
    \end{tabular} \hspace{1cm}
    \begin{tabular}{@{}c@{}}
    $\config{S_1}{\sigma} \overset{p}{\rightarrow} \errConfig$ \\
    \hline
    $ \config{\concat}{\sigma} \overset{p}{\rightarrow} \errConfig$ \\
    \end{tabular} \vspace{0.3cm} \\

    \begin{tabular}{@{}c@{}}
        $\config{S_1}{\sigma} \overset{p}{\rightarrow} \config{S_1'}{\sigma'}$ \\
        \hline
        $ \config{\concat}{\sigma} \overset{p}{\rightarrow} \config{S_1';S_2}{\sigma'}$ \\
    \end{tabular} \hspace{1cm}
    \begin{tabular}{@{}c@{}}
    $tr(M_0 \sigma M_0^\dagger)>0$ \\
    \hline
    $\config{\while}{\sigma} \overset{tr(M_0 \sigma M_0^\dagger)}{\rightarrow}\config{\emptyProgram}{\frac{M_0 \sigma M_0^\dagger}{tr(M_0 \sigma M_0^\dagger)}}$
    \end{tabular}\vspace{0.3cm} \\

    \begin{tabular}{@{}c@{}}
        $tr(M_1 \sigma M_1^\dagger)>0$\\
        \hline
         $\config{\while}{\sigma} \overset{tr(M_1 \sigma M_1^\dagger)}{\rightarrow} \config{S; \while}{\frac{M_1 \sigma M_1^\dagger}{tr(M_1 \sigma M_1^\dagger)}}$
    \end{tabular}\vspace{0.3cm} \\

    \begin{tabular}{@{}c@{}}
        \hline
        $\errConfig \overset{1}{\rightarrow} \termConfig$
    \end{tabular}\hspace{1cm}
    \begin{tabular}{@{}c@{}}
        \hline
        $\config{\emptyProgram}{\sigma} \overset{1}{\rightarrow} \termConfig$
    \end{tabular}\hspace{1cm}
    \begin{tabular}{@{}c@{}}
        \hline
        $\termConfig \overset{1}{\rightarrow} \termConfig$
    \end{tabular}
\end{tabular}
    \caption{Transition probability function of MC $\operationalMC{\rho}{S}$ for all $\sigma \in \densityFull$ where $\emptyProgram;S_2 \equiv S_2$}
    \label{fig:transitions}
    \end{figure}
\subsubsection{Operational Semantics}
We start by defining the operational semantics of a program $S$ as a Markov chain inspired by \cite{conditioningProb} instead of non-deterministic relations in comparison to \cite{floydHoareLogic}. A quantum \emph{configuration} is a tuple $\config{S}{\rho}$ with density operator $\rho \in \densityFull$. Note that we consider normalized density operators $\densityFull$ instead of partial density operators $\density$.
Intuitively, $S$ is the program that is left to evaluate and $\rho$ is the current state. We use $\emptyProgram$ to denote that there is no program left to evaluate. The set of all configurations over $\mathcal{H}$ is denoted as $\setConfig$. The quantum configuration for violated observations is $\errConfig$ and for termination is $\termConfig$.

\begin{definition}
    The \emph{operational semantics} of a program $S$ with initial state $\rho \in \densityFull$ is defined as the Markov chain $\operationalMC{\rho}{S} =(\states,\textbf{P},s_{init})$ where:
    \begin{itemize}
        \item $\states=\setConfig \cup \{\errConfig, \termConfig\}$,
        \item $s_{init}=\config{S}{\rho}$,
        \item $\textbf{P}$ is the smallest function satisfying the inference rules in Figure~\ref{fig:transitions} where $c \overset{p}{\rightarrow} c'$ means $\textbf{P}(c,c')=p>0$. For all other pairs of states the transition probability is $0$.
    \end{itemize}
\end{definition}

The meaning of a transition $\config{S}{\sigma} \overset{p}{\rightarrow} \config{S'}{\sigma'}$ is that after evaluating program $S$ on state $\sigma$, with probability $p$ the new state is $\sigma'$ and the program left to execute is $S'$. For the observe statement, there are two successors, $\config{\observe}{\sigma} \overset{tr(O \sigma O^\dagger)}{\rightarrow} \config{\emptyProgram}{\frac{O\sigma O ^\dagger}{tr(O \sigma O^\dagger)}}$ and $\config{\observe}{\sigma}\overset{1-tr(O \sigma O^\dagger)}{\rightarrow} \errConfig$. The observation $O$ is satisfied by state $\sigma$ with probability $tr(O \sigma O^\dagger)$ and then it terminates successfully. If the observation is violated (with probability $1-tr(O \sigma O^\dagger)$), the successor state is $\errConfig$, the state that captures paths with violated observations.
For details of the other rules we refer to \cite{floydHoareLogic}.

\subsubsection{Denotational Semantics}
We now provide a denotational semantics for quantum while-programs. To handle observations and distinguish between non-terminating runs and those that violate observations, we introduce denotational semantics in a slightly different way than \cite{floydHoareLogic}. To do so, we start with defining some basics:

For tuples $(\rho,p), (\sigma,q) \in \density \times \R_{\geq 0}$, we define multiplication with a constant $a \in \R_{\geq 0}$ and addition entrywise: $
    a(\rho,p) := (a \rho, a p)$ and $
    (\rho,p)+(\sigma,q) := (\rho+\sigma, p+q)$.

The least upper bound (lub) of a set of tuples is defined as the entrywise lub provided it exists, i.e., $\bigvee_{n=0}^\infty (\rho_n,p_n):= (\bigvee_{n=0}^\infty \rho_n, \bigvee_{n=0}^\infty p_n)$ where $\bigvee_{n=0}^\infty \rho_n$ is the lub w.r.t. the Loewner partial order $\sqsubseteq$ and $\bigvee_{n=0}^\infty p_n$ is the lub w.r.t. to the classical ordering $\leq$ on $\mathbb{R}_{\geq 0}$.

As the probability of violating observations depends on the density operator, we introduce $\densityNumberPairs = \{(\rho,p) \in \density \times \R_{\geq 0} \mid tr(\rho) +p \leq 1\} \subseteq \traceclass \times \mathbb{C}$. $\traceclass \times \mathbb{C}$ is isomorphic to the set of operators of the form $\begin{pmatrix}
    \rho & \\
     & p
\end{pmatrix} \in T(\mathcal{H}\otimes \mathbb{C})$. Thus the trace and the norm from $T(\mathcal{H}\otimes \mathbb{C})$ apply. Specifically, $\tilde{tr}(\rho,p):= tr(\rho) +p$ and $\norm*{(\rho,p)} := \norm*{\rho} + \abs*{p}$ for $(\rho,p)\in \traceclass \times \mathbb{C}$.
\begin{definition}
     The \emph{denotational semantics} of a quantum program $S$ is defined as a mapping $\semantics{S}:\densityNumberPairs \to \densityNumberPairs$. For $(\rho,p)\in \densityNumberPairs$, $\rho$ is used for density-transformer semantics as defined in \cite{floydHoareLogic} and $p$ for the probability of an observation violation.

The denotational semantics for $(\rho,p) \in \densityNumberPairs$ is given by
\begin{itemize}
    \item $\semantics{\skipbf} (\rho ,p)= (\rho,p)$.

    \item $\semantics{\qzero } (\rho,p )= \begin{cases}
        (\ket{0} \bra{0}_q \rho \ket{0}\bra{0}_q+\ket{0} \bra{1}_q\rho \ket{1} \bra{0}_q,p) &\text{, if } type(q)=\boolType\\
        (\sum_{n\in \mathbb{Z}} \ket{0} \bra{n}_q \rho \ket{n} \bra{0}_q,p)  &\text{, if } type(q)=\intType.
        \end{cases}$

    \item $\semantics{\Uq} (\rho,p )=(U\rho U^\dagger ,p)$.

    \item $\semantics{\observe} (\rho,p)
    = (O \rho O^\dagger,p+ tr(\rho)-tr(O\rho O^\dagger) )$.

    \item $\semantics{\concat} (\rho,p )= \semantics{S_2} (\semantics{S_1} (\rho,p ) ) $.

    \item $\semantics{\measure }(\rho ,p)= \sum_m \semantics{S_m } (M_m \rho M_m ^\dagger ,0) + (\zeroOp,p)$ with $M=\{M_m\}_{m \in I}$ and $\Bar{S}=\{S_m\}_{m\in I}$.

    \item $\semantics{\while} (\rho,p) = \bigvee_{n=0}^\infty \semantics{(\while)^n} (\rho,p)$
    with $M=\{M_0, M_1\}$ where loop unfoldings are defined inductively
    \begin{align*}
        (\while)^0 &\equiv \Omega\\
        (\while)^{n+1} &\equiv \ifstatement{S; (\while)^n}{\skipbf}
    \end{align*}
    where $\Omega$ is a syntactic quantum program with $\semantics{\Omega}(\rho,p) = (\zeroOp,p)$ as in \cite{floydHoareLogic}.
\end{itemize}
\end{definition}
We write $\semanticsRho{S} (\rho,p)$ and $\semanticsErr{S} (\rho,p)$ to denote the first/second component of $\semantics{S} (\rho,p)$.
It follows directly that our definition is a conservative extension of \cite{floydHoareLogic}:
\begin{proposition}
    \label{prop:unconditionedSemantics}
    For an observe-free program $S$, input state $\rho \in \density$ and $p\in \R_{\geq 0}$, is $\semantics{S} (\rho,p) = (\semanticsOriginal{S} (\rho),p)$
    where $\semanticsOriginal{S} (\rho)$ is the denotational semantics as defined in \cite{floydHoareLogic}.
\end{proposition}

Some intuition behind those tuples: If $\semantics{S}(\rho,0) = (\rho',p')$ for a program $S$ with initial pair $(\rho,0)$, then the probability of violating an observation while executing $S$ on $\rho \in \densityFull$ is $p'$. The probability of terminating normally (without violating an observation) is given by $tr(\rho')$ and the probability for non-termination is $1-tr(\rho')-p'$. As in the observe-free case, $\rho'$ is the (non-normalized) state after $S$ has been executed (and terminated) on $\rho$. It is easy to see that only the observation statement can change the value of the second entry.

\begin{proposition}
    \label{pro:allClaims}
    For $(\rho,p), (\rho,q )\in \densityNumberPairs$ and program $S$:
    \begin{enumerate}
        \item $\semanticsRho{S}(\rho,p) = \semanticsRho{S}(\rho,q)$
        \item $p \leq\semanticsErr{S}(\rho,p)$
        \item if $(\rho,q+p)\in \densityNumberPairs$ then $\semanticsErr{S}(\rho,q+p) = \semanticsErr{S}(\rho,q) + p$
        \item $\tilde{tr}(\semantics{S}(\rho,p))\leq \tilde{tr}(\rho,p)$
        \item $\semantics{S}$ is well defined, i.e., $\semantics{S}(\rho,p) \in \densityNumberPairs$ and the least upper bound exists.
        \item $\semantics{S}$ is linear
    \end{enumerate}
\end{proposition}
\begin{proof}
    All claims can be shown by doing an induction over $S$, see Appendix~\ref{app:semantics}.
\end{proof}
As $\semanticsRho{S}(\rho,p) = \semanticsRho{S}(\rho,q)$, we use $\semanticsRho{S}(\rho)$ instead.
Three consequences of Proposition~\ref{pro:allClaims}:
\begin{lemma}
    \label{lem:bounded}
 For $(\rho,p), (\sigma,q) \in \densityNumberPairs$ with $(\rho+\sigma,p+q) \in \densityNumberPairs$ and programs $S, S_1, S_2$:
 \begin{enumerate}
    \item $tr(\semanticsRho{S}(\rho,p))\leq tr(\rho)$, i.e., $\semanticsRho{S}$ is trace-reducing
    \item $\semanticsErr{\concat} (\rho,q+p) =
        \semanticsErr{S_2}(\semanticsRho{S_1}(\rho,0),q) + \semanticsErr{S_1}(\rho,p)$
    \item $\semanticsRho{S}$ is bounded linear
 \end{enumerate}
\end{lemma}
The proof can be found in the Appendix~\ref{app:semantics}.

\subsubsection{Equivalence of Semantics}
The following lemma asserts the equivalence of our operational and denotational semantics. Intuitively, the denotational semantics gives a distribution over final states and its second component captures the probability to reach $\errConfig$, the state for violated observations. As the operational semantics is only defined for $tr(\rho) = 1$, we only consider this case:
\begin{lemma}
    \label{lem:operationalVsDenotational}
    For any program $S$ and initial state $\rho \in \densityFull$
    \begin{itemize}
        \item $\semantics{S}(\rho,0) = (\sum_{\rho'} \prob{\operationalMC{\rho}{S}}{\eventually \config{\emptyProgram}{\rho'}} \cdot \rho',\prob{\operationalMC{\rho}{S}}{\eventually \errConfig})$
        \item $\prob{\operationalMC{\rho}{S}}{\eventually \termConfig} = tr(\semanticsRho{S}(\rho,0))+\semanticsErr{S}(\rho,0)$
    \end{itemize}
\end{lemma}
\begin{proof}
    The first item can be shown by induction over $S$. For the second item we use that every path that eventually reaches $\termConfig$ passes through either a $\config{\emptyProgram}{\rho'}$ or a $\errConfig$ state, see Appendix~\ref{app:semantics}.
\end{proof}
\section{Weakest Preconditions}
\label{sec:wp}

In this section, we consider how we can extend the weakest precondition calculus to capture observations and thus compute conditional probabilities of quantum programs using deductive verification. Recall that a predicate $P$ satisfies $\zeroOp \sqsubseteq P \sqsubseteq \identityOp$. Let $tr(P\rho)$ by the probability that $\rho$ satisfies $P$. Note that if $P$ is a projector, then $tr(P\rho)$ equals the probability that $\rho$ gives answer ``yes'' in a measurement defined by $P$. Even if $P$ is not a projection, $tr(P\rho)$ is the average value of measuring $\rho$ with the measurement described by the observable $P$. 
If not given directly, all proofs can be found in the Appendix~\ref{app:wp}.

\subsection{Total and Partial Correctness}
Defining the semantics in a different way also changes the definition of Hoare logic with total and partial correctness \cite{floydHoareLogic}:
\begin{definition}
    Let $P,Q\in \predicate$, $S$ a program, $\rho\in \density$ and $\{P\} S \{Q\}$ a correctness formula. Then
    \begin{enumerate}
        \item (total correctness) $\models_{tot} \{P\} S \{Q\}$ iff $tr(P\rho)\leq tr(Q\semanticsRho{S } (\rho,0))$
        \item (partial correctness) $\models_{par} \{P\} S \{Q\}$ iff $tr(P\rho)\leq tr(Q\semanticsRho{S } (\rho,0))+ tr(\rho)- tr(\semanticsRho{S } (\rho,0)) - \semanticsErr{S } (\rho,0)$
    \end{enumerate}
\end{definition}
Let us explain this definition. Assume $tr(\rho)=1$, otherwise all probabilities mentioned in the following are non-normalized.
Recall that $tr(P\rho)$ is the probability that state $\rho$ satisfies predicate $P$ and $tr(Q\semanticsRho{S } (\rho,0))$ is the probability that the state after execution of $S$ starting with $\rho$ satisfies predicate $Q$. Total correctness entails that the probability of a state satisfying precondition $P$ is at most the probability that it satisfies postcondition $Q$ after execution of $S$. This only involves terminating runs.
In the formula of partial correctness, the summand $\semanticsErr{S } (\rho,0)$ captures the probability that an observation is violated during executing program $S$ on state $\rho$. As before, $tr(\rho)- tr(\semanticsRho{S } (\rho,0))$ captures the probability that $S$ on state $\rho$ does not terminate.

Similar to \cite{floydHoareLogic}, we have some nice but different properties:
\begin{proposition}
    \begin{enumerate}
        \item $\models_{tot} \hoare{P}{S}{Q}$ implies $\models_{par} \hoare{P}{S}{Q}$
        \item $\models_{tot} \hoare{\zeroOp}{S}{Q}$. However, $\models_{par} \hoare{P}{S}{\identityOp}$ does not hold in general.
        \item For $P_1,P_2,Q_1,Q_2 \in \predicate$ and $\lambda_1, \lambda_2 \in \mathbb{R}_{\geq 0}$ with $\lambda_1 P_1 + \lambda_2 P_2, \lambda_1 Q_1 + \lambda_2 Q_2 \in \predicate $:
        $\models_{tot} \hoare{P_1}{S}{Q_1} \land \models_{tot} \hoare{P_2}{S}{Q_2}$ implies $\models_{tot} \hoare{\lambda_1 P_1 + \lambda_2 P_2}{S}{\lambda_1 Q_1 + \lambda_2 Q_2} $
    \end{enumerate}
\end{proposition}
\begin{proof}
    \begin{enumerate}
        \item Follows from definition and $tr(\semanticsRho{S} (\rho,0))+ \semanticsErr{S} (\rho,0)\leq tr(\rho)+0$ (Proposition~\ref{pro:allClaims})
        \item $\models_{tot} \hoare{\zeroOp}{S}{Q}$ follows from definition, $ tr(\zeroOp)=0$ and $tr(Q\sigma)\geq 0$ for all $\sigma\in \density$. For disproving $\models_{par} \hoare{P}{S}{\identityOp}$, consider a program (with only one variable $q$) $S \equiv \mathbf{observe }(\ket{1}\bra{1}, q)$. Then the statement does not holds with $P=\identityOp$ and $\rho = \ket{0}\bra{0}$ because $\semanticsErr{S}(\rho,0) > 0$.
        \item Follows from the linearity of the trace and the definition of $\models_{tot}$.
    \end{enumerate}
\end{proof}

\subsection{Weakest (Liberal) Preconditions}
Given a postcondition and a program, we are interested in the best (weakest) precondition w.r.t. total and partial correctness:
\begin{definition}
    Let program $S$ and predicate $P \in \predicate$.
    \begin{enumerate}
        \item The \emph{weakest precondition} is defined as $\qwp{S}{P}= \sup \{Q \mid \text{ } \models_{tot} \hoare{Q}{S}{P}\}$. Thus $\models_{tot}\hoare{\qwp{S}{P}}{S}{P}$ and $\models_{tot}\hoare{Q}{S}{P}$ implies $Q\sqsubseteq \qwp{S}{P}$ for all $Q\in \predicate$.
        \item The \emph{weakest liberal precondition} is defined as $\qwlp{S}{P}= \sup \{Q \mid \text{ }\models_{par} \hoare{Q}{S}{P}\}$. Thus $\models_{par}\hoare{\qwlp{S}{P}}{S}{P}$ and $\models_{par}\hoare{Q}{S}{P}$ implies $Q\sqsubseteq \qwlp{S}{P}$ for all $Q\in \predicate$.
    \end{enumerate}
\end{definition}
The following lemmas show that these suprema indeed exist. Both proofs are based on the Schrödinger-Heisenberg duality \cite{heisenbergdualityUnruh}.
\begin{lemma}
    \label{lem:schroedinger}
    For a function $\semantics{S}: \densityNumberPairs\to \densityNumberPairs$ with properties as in Proposition~\ref{pro:allClaims}, the weakest precondition $qwp\llbracket S \rrbracket: \predicate \rightarrow \predicate$ exists and is bounded linear and subunital. It satisfies $tr(\qwp{S}{P}\rho) = tr(P\semanticsRho{S}(\rho,0))$ for all $\rho \in \density, P \in \predicate$ and it is the only function of this type with this property.
\end{lemma}
This lemma (and the following one) does not require $\semantics{S}$ to be a denotational semantics of some program $S$. In contrast to \cite{floydHoareLogic}, this result thus still holds if the language is extended as long as the conditions still holds.
\begin{lemma} \label{lem:wlpexistence}
    For a function $\semantics{S}:\densityNumberPairs \rightarrow \densityNumberPairs$ with properties as in Proposition~\ref{pro:allClaims}, the weakest liberal precondition $qwlp\llbracket S \rrbracket: \predicate \rightarrow \predicate$ exists and is subunital. It satisfies
    \begin{equation*}
        tr(\qwlp{S}{P}\rho) = tr(P\semanticsRho{S}(\rho,0)) + tr(\rho)-tr(\semanticsRho{S}(\rho,0))- \semanticsErr{S}(\rho,0)
    \end{equation*}
    for each $\rho \in \density, P\in \predicate$ and it is the only function of this type with this property.
\end{lemma}
This general theorem about the existence of weakest liberal preconditions also applies for programs without observations (because $\semanticsErr{S}(\rho,0)=0$ and $\semanticsRho{S}(\rho,0)=\semanticsOriginal{S}(\rho)$ for each $\rho$ for observation-free program $S$, Proposition~\ref{prop:unconditionedSemantics}).
Lemma~\ref{lem:schroedinger} and~\ref{lem:wlpexistence} extend \cite{DHondtWeakestPreconditions} to the infinite-dimensional case and to partial correctness, i.e., the existence of weakest liberal preconditions.
Now we consider some healthiness properties about weakest (liberal) preconditions: 
\begin{proposition}
    \label{prop:healthWP}
    For every program $S$, the function $qwp\llbracket S \rrbracket:\predicate \to \predicate$ satisfies:
    \begin{itemize}
        \item Bounded linearity
        \item Subunitality: $\qwp{S}{\identityOp}\sqsubseteq \identityOp$
        \item Monotonicity: $P\sqsubseteq Q$ implies $\qwp{S}{P}\sqsubseteq \qwp{S}{Q}$
        \item Order-continuity: $\qwp{S}{\bigvee_{i=0}^\infty P_i} = \bigvee_{i=0}^\infty \qwp{S}{P_i}$ if $\bigvee_{i=0}^\infty P_i$ exists
    \end{itemize}
\end{proposition}
\begin{proposition} \label{prop:healthWLP}
    For every program $S$, the function $qwlp\llbracket S \rrbracket:\predicate \to \predicate$ satisfies:
    \begin{itemize}
        \item Affinity: The function $f:\predicate \to \predicate$ with $f(P)=\qwlp{S}{P} - \qwlp{S}{\zeroOp}$ is linear.
        Note that this implies convex-linearity and sublinearity.
       \item Subunitality: $\qwlp{S}{\identityOp}\sqsubseteq \identityOp$
        \item Monotonicity: $P\sqsubseteq Q$ implies $\qwlp{S}{P}\sqsubseteq \qwlp{S}{Q}$
        \item Order-continuity: $\qwlp{S}{\bigvee_{i=0}^\infty P_i} = \bigvee_{i=0}^\infty \qwlp{S}{P_i}$ if $\bigvee_{i=0}^\infty P_i$ exists
    \end{itemize}
\end{proposition}
For our denotational semantics $\semantics{S}$, we can also give an explicit representation of $qwp\llbracket S \rrbracket$:
\begin{proposition}
    \label{prop:wpDef}
    Let $P \in \predicate$:
    \begin{itemize}
        \item $\qwp{\skipbf}{P}=P$
        \item $\qwp{\qzero}{P}=
        \begin{cases}
            \ket{0}_q\bra{0}P\ket{0} \bra{0}_q + \ket{1}\bra{0}_q P \ket{0} \bra{1}_q &
         ,\text{if }type(q)=\boolType \\
            \sum_{n\in \mathbb{Z}} \ket{n} \bra{0}_q P \ket{0}\bra{n}_q&
        ,\text{if }type(q)=\intType
        \end{cases}$
        \item $\qwp{\Uq}{P}= U^\dagger P U$
        \item $\qwp{\observe}{P} = O^\dagger P O$
        \item $\qwp{\concat}{P}=\qwp{S_1}{\qwp{S_2}{P}}$
        \item $\qwp{\measure}{P}\rho  = \sum_m M_m^\dagger (\qwp{S_m}{P}) M_m$
        \item $\qwp{\while'}{P} = \bigvee_{n=0}^\infty P_n$ with
        \begin{align*}
            P_0  = \zeroOp, &&
            P_{n+1} = [M_0^\dagger P  M_0] + [M_1^\dagger(\qwp{S'}{P_n}) M_1]
        \end{align*}
        and $\bigvee_{n=0}^\infty$ denoting the least upper bound w.r.t. $\sqsubseteq$.
    \end{itemize}
\end{proposition}
\begin{proof}
    We prove $tr(\qwp{S}{P}\rho) = tr(P\semanticsRho{S}(\rho,0))$, which then, together with Lemma~\ref{lem:schroedinger} implies that it is indeed the weakest precondition.
\end{proof}
In this and the following proposition we mean convergence of sums with respect to the SOT, more details can be found in \cite{heisenbergdualityUnruh}.
For weakest liberal preconditions the explicit representation looks quite similar:
\begin{proposition}
    \label{prop:wlpDef}
    Let $P \in \predicate$. For most cases, $\qwlp{S}{P}$ is defined analogous to $\qwp{S}{P}$ (replacing every occurrence of $qwp$ by $qwlp$). The only significant difference is the while-loop:
    $\qwlp{\while'}{P} = \bigwedge_{n=0}^\infty P_n$ with
        \begin{align*}
            P_0  =\identityOp, &&
            P_{n+1} = [M_0^\dagger P  M_0] + [M_1^\dagger(\qwlp{S'}{P_n}) M_1]
        \end{align*}
        and $\bigwedge_{n=0}^\infty$ denoting the greatest lower bound w.r.t. $\sqsubseteq$.
\end{proposition}
\begin{proof} This proof is similar to Proposition~\ref{prop:wpDef} except that we show that $tr(\qwlp{S}{P}\rho) = tr(P\semanticsRho{S}(\rho,0)) + tr(\rho)-tr(\semanticsRho{S}(\rho,0))- \semanticsErr{S}(\rho,0)$ holds which then implies together with Lemma~\ref{lem:wlpexistence} that it is indeed the weakest liberal precondition.
\end{proof}
Both explicit representations above are conservative extensions of the weakest (liberal) precondition calculus in \cite{floydHoareLogic}.

For the following explanations, assume $tr(\rho)=1$, otherwise the probabilities are not normalized. To understand those definitions, consider $tr(\qwp{S}{P} \rho)$.  Due to the duality from Lemma~\ref{lem:schroedinger}, $tr(\qwp{S}{P}\rho) = tr(P\semanticsRho{S}(\rho,0))$, so it is the probability that the result of running program $S$ (without violating any observations) on the initial state $\rho$ satisfies predicate $P$. Similarly, $tr(\qwlp{S}{P}\rho)$ adds the probability of non-termination too. This is equivalent to the standard interpretation of weakest (liberal) preconditions in \cite{floydHoareLogic}.

For programs with observations $tr(\qwlp{S}{P}\rho)$ does not include runs that violate an observation. Thus, $tr(\qwlp{S}{\identityOp} \rho)$ gives the probability that no observation is violated during the run of $S$ on input state $\rho$ (while for programs without observations and in \cite{floydHoareLogic}, this will always be $tr(\rho)=1$). The probability that a program state $\rho$ will satisfy the postcondition $P$ after executing program $S$ while not violating any observation is then a conditional probability. To handle this case, we introduce conditional weakest preconditions inspired by \cite{conditioningProb} in the next section.

\subsection{Conditional Weakest Preconditions}
In the following, we consider pairs of predicates. Addition and multiplication are interpreted pointwise, i.e., $(P,Q) + (P',Q') = (P+P',Q+Q')$ and $M \multdot (P,Q) = (MP,MQ)$ resp. $(P,Q) \multdot M = (PM,QM)$ where $M$ can be a constant or an operator. Multiplication binds stronger than addition.

We define a natural ordering on pairs of predicates that is used for example to express healthiness conditions:
\begin{definition}
We define $\unlhd$ on $\predicate^2$ as follows: $
    (P,Q) \unlhd (P',Q') \Leftrightarrow P \sqsubseteq P' \land Q' \sqsubseteq Q$
where $\sqsubseteq$ is the Loewner partial order. The least element is $(\zeroOp,\identityOp)$ and the greatest element $(\identityOp,\zeroOp)$.
The least upper bound of an increasing chain $\{(P_i,Q_i)\}_{i\in \mathbb{N}}$ for $(P_i,Q_i)\in \predicate^2$ is given pointwise by
$\bigvee_{i=0}^\infty (P_i,Q_i) = (\bigvee_{i=0}^\infty P_i ,\bigwedge_{i=0}^\infty Q_i)$.
\end{definition}
\begin{lemma}
    $\unlhd$ is an $\omega$-cpo on $\predicate^2$.
\end{lemma}
\begin{proof}
    We have to show that $\bigvee_{i=0}^\infty (P_i,Q_i)$ exists for any increasing chain $\{(P_i,Q_i)\}_{i\in \mathbb{N}}$.
    If $\{(P_i,Q_i)\}_{i\in \mathbb{N}}$ is increasing with respect to $\unlhd$, then $\{P_i\}_{i\in \mathbb{N}}$ is increasing and $\{Q_i\}_{i\in \mathbb{N}}$ decreasing with respect to $\sqsubseteq$. The existence of $\bigvee_{i=0}^\infty P_i $ follows directly from $\sqsubseteq$ being an $\omega$-cpo on $\predicate$. For $\bigwedge_{i=0}^\infty Q_i$, we use the same trick as in the proof of Proposition~\ref{prop:wlpDef}: $\{\identityOp - Q_i\}_{i\in \mathbb{N}}$ is an increasing chain of predicates and thus $ \bigwedge_{i=0}^\infty Q_i = \identityOp - \bigvee_{i=0}^\infty (\identityOp - Q_i)$ exists too. That means both $\bigvee_{i=0}^\infty P_i $ and $\bigwedge_{i=0}^\infty Q_i $ exist and thus also $\bigvee_{i=0}^\infty (P_i,Q_i)$.
\end{proof}
Combining weakest preconditions and liberal weakest preconditions, we can define \emph{conditional weakest preconditions} similar to the probabilistic case \cite{conditioningProb}:
\begin{definition}
    \label{def:qcwp}
    The \emph{conditional weakest precondition} transformer is a mapping $qcwp\llbracket S \rrbracket: \predicate^2\to \predicate^2$ defined as $
    \qcwp{S}{P,Q} := (\qwp{S}{P}, \qwlp{S}{Q})$.
\end{definition}
Similar to the weakest precondition calculus, we can also give an explicit representation which can be found in the Appendix~\ref{app:wp}, Lemma~\ref{lem:explicitCWP}.

Again, assume $tr(\rho)=1$ in the following, otherwise the probabilities are not normalized.
Note that $tr(\qwlp{S}{\identityOp}\rho)$ is the probability that no observation is violated and $tr(\qwp{S}{P}\rho)$ the probability that $P$ is satisfied after $S$ has been executed on $\rho$ (see above). We are interested in the conditional probability of establishing the postcondition given that no observation is violated, namely $\frac{tr(\qwp{S}{P}\rho)}{tr(\qwlp{S}{\identityOp}\rho)}$. Notice that for $\qcwp{S}{(P,\identityOp)}=(A,B)$, this is simply $\frac{tr(A\rho)}{tr(B\rho)}$. That means we can immediately read of this conditional probability from $qcwp\llbracket S \rrbracket$. Formally, we use \begin{align*}
    \fracTr(A\rho,B\rho) := \begin{cases}\frac{tr(A\rho)}{tr(B\rho)}, &\text{ if } tr(B\rho)\neq 0\\
        \text{undefined} ,&\text{ otherwise}.
    \end{cases}
\end{align*}

We now establish some properties of conditional weakest preconditions:
\begin{proposition} \label{prop:healthcwp}
    For every program $S$, the function $qcwp\llbracket S \rrbracket:\predicate^2 \to \predicate^2$ satisfies:
    \begin{itemize}
        \item Has a linear interpretation: for all $\rho \in \density, a,b \in \mathbb{R}_{\geq 0}$ and $P,P'\in \predicate$ with $aP+bP' \in \predicate$
        \begin{align*}
            \fracTr(\qcwp{S}{aP+bP',\identityOp} \multdot \rho) = a \cdot \fracTr(\qcwp{S}{P,\identityOp}\multdot \rho) + b\cdot\fracTr(\qcwp{S}{P',\identityOp} \multdot\rho)
        \end{align*}
        \item Affinity: The function $\qcwp{S}{P,Q} - \qcwp{S}{\zeroOp,\zeroOp}$ is linear. Note that this implies convex-linearity and sublinearity.
        \item Monotonicity: $(P,P')\unlhd (Q,Q')$ implies $\qcwp{S}{P,P'}\unlhd \qcwp{S}{Q,Q'}$
        \item Continuity: $\qcwp{S}{\bigvee_{i=0}^\infty (P_i,Q_i)} = \bigvee_{i=0}^\infty \qcwp{S}{P_i,Q_i}$ if $\bigvee_{i=0}^\infty (P_i,Q_i)$ exists
    \end{itemize}
\end{proposition}

\subsection{Conditional Weakest Liberal Preconditions}
Similar to the conditional weakest precondition, we can also define the same with weakest liberal preconditions for partial correctness:
\begin{definition}
    The \emph{conditional weakest liberal precondition} $qcwlp:\predicate^2\to \predicate^2$ is defined as $\qcwlp{S}{P,Q}:= (\qwlp{S}{P}, \qwlp{S}{Q})$ for each program $S$ and predicates $P,Q \in \predicate$.
\end{definition}
\begin{definition} We define $\dot \unlhd$ on $\predicate^2$ as follows $
        (P,Q) \phantom{.}\dot \unlhd \phantom{.} (P',Q') \Leftrightarrow P \sqsubseteq P' \land Q \sqsubseteq Q'$
    where $\sqsubseteq$ is the Loewner partial order. The least element is $(\zeroOp,\zeroOp)$ and the greatest element $(\identityOp,\identityOp)$.
    The least upper bound of an increasing chain $\{(P_i,Q_i)\}_{i\in \mathbb{N}}$ for $(P_i,Q_i)\in \predicate^2$ is given pointwise by $
        \bigvee_{i=0}^\infty (P_i,Q_i) = (\bigvee_{i=0}^\infty P_i ,\bigvee_{i=0}^\infty Q_i)$.
\end{definition}
Note that in contrast to $\unlhd$, both components are ordered in the same direction. Here it follows directly that $\dot \unlhd$ is an $\omega$-cpo on $\predicate^2$.

Similar as for $qcwp$, we can now read off the conditional satisfaction of $P$ when we want non-termination to count as satisfaction:
$\fracTr(\qcwlp{S}{P,\identityOp}\multdot \rho) = \frac{tr(\qwlp{S}{P} \rho)}{tr(\qwlp{S}{\identityOp}\rho)}$ which is equal to dividing the probability to satisfy $P$ after execution (including non-termination) by the probability to not violate an observation\footnote[1]{So far, we considered conditional weakest preconditions for total and partial correctness, i.e., $\qcwp{S}{P,\identityOp} = (\qwp{S}{P}, \qwlp{S}{\identityOp})$ and $\qcwlp{S}{P,\identityOp} = (\qwlp{S}{P}, \qwlp{S}{\identityOp})$.
In \cite[Sect. 8.3]{benniDiss} it is argued why other combinations such as $(\qwp{S}{P},\qwp{S}{\identityOp})$ and $(\qwlp{S}{P},\qwp{S}{\identityOp})$ only make sense if a program is almost-surely terminating, i.e., without non-termination. Their arguments apply without change in our setting, so we do not consider these combinations either.}.

We can also conclude some properties about conditional weakest liberal preconditions:
\begin{proposition} \label{prop:healthcwlp}
    For every program $S$, the function $qcwlp\llbracket S \rrbracket:\predicate^2 \to \predicate^2$ satisfies:
    \begin{itemize}
        \item Affinity: The function $\qcwlp{S}{P,Q} - \qcwlp{S}{\zeroOp,\zeroOp}$ is linear. Note that this implies convex-linearity and sublinearity.
        \item Monotonicity: $(P,Q) \phantom{.}\dot \unlhd \phantom{.} (P',Q')$ implies $\qcwlp{S}{P,Q}\phantom{.}\dot \unlhd \phantom{.}\qcwlp{S}{P',Q'}$
        \item Continuity: $\qcwlp{S}{\bigvee_{i=0}^\infty (P_i,Q_i)} = \bigvee_{i=0}^\infty \qcwlp{S}{P_i,Q_i}$ if $\bigvee_{i=0}^\infty (P_i,Q_i)$
    \end{itemize}
\end{proposition}

\subsection{Observation-Free Programs}
For observation-free programs, our interpretations coincides with the satisfaction of weakest (liberal) preconditions from \cite{floydHoareLogic}:
\begin{lemma}
    For an observation-free program $S$, predicate $P\in \predicate$ and state $\rho \in \densityFull$:
    \begin{align*}
        \fracTr(qcw(l)p\llbracket S \rrbracket (P,\identityOp) \multdot \rho)= tr(qw(l)p\llbracket S \rrbracket (P)\rho)
    \end{align*}
\end{lemma}
\begin{proof}
    For every observation-free program $S$ is $\qwlp{S}{\identityOp} = \identityOp$ \cite{floydHoareLogic} and which means for $\rho \in \densityFull$ (and not for all $\rho \in \density$) $\fracTr(qcw(l)p\llbracket S \rrbracket (P,\identityOp) \multdot \rho)$ is equal to
    \begin{align*}
        \frac{tr(qw(l)p\llbracket S \rrbracket (P)\rho)}{tr(\qwlp{S}{\identityOp}\rho)} = \frac{tr(qw(l)p\llbracket S \rrbracket (P)\rho)}{tr(\rho)} = tr(qw(l)p\llbracket S \rrbracket (P)\rho).
    \end{align*}
\end{proof}

\subsection{Correspondence to Operational MC Semantics}
The aim of this section is to establish a correspondence between $\qcwp{S}{P,\identityOp}$ and the operational semantics of $S$. In order to reason about $P$ in terminal states of the Markov chain, we use rewards. First, we equip the Markov chain used for the operational semantics with a reward function with regard to a postcondition $P$:
\begin{definition}
    \label{def:opMC}
    For program $S$ and postcondition $P$, the Markov reward chain $\rewardMC{\rho}{P}{S}$ is the MC $\operationalMC{\rho}{S}$ extended with a function $r:\states \to \mathbb{R}_{\geq 0}$ such that $r(\config{\emptyProgram}{\rho'}) = tr(P\rho')$ and $r(s)=0$ for all other states $s \in \states$.
\end{definition}
The (liberal) reward of a path $\pi$ of $\rewardMC{\rho}{P}{S}$ is defined as
 $r(\pi) = \begin{cases}
    tr(P\rho') &,\text{ if } \config{\emptyProgram}{\rho'}\in \pi\\
    0 &, \text{ else}
\end{cases}$ and $lr(\pi)=r(\pi)$ expect if $\termConfig \not \in \pi$, then $lr(\pi)=1$.

The expected reward of $\eventually \termConfig$ is the expected value of $r(\pi)$ for all $\pi\in \eventually \termConfig$, i.e., $ER^{\rewardMC{\rho}{P}{S}} (\eventually \termConfig) = \sum_{\rho'} \prob{\rewardMC{\rho}{P}{S}}{\eventually \config{\emptyProgram}{\rho'}}\cdot tr(P\rho')$. The liberal version adds rewards of non-terminating paths, i.e., $LER^{\rewardMC{\rho}{P}{S}} (\eventually \termConfig) = ER^{\rewardMC{\rho}{P}{S}} (\eventually \termConfig) + \prob{\rewardMC{\rho}{P}{S}}{\neg \eventually \termConfig}$.

Now we start by showing some auxiliary results, similar to \cite[Lemma 5.5, 5.6]{conditioningProb}:
\begin{lemma}
    \label{lem:ErWpEqual}
    For a program $S$, state $\rho\in \densityFull$, predicate $P\in\predicate$ we have
    \begin{align*}
        \prob{\rewardMC{\rho}{P}{S}}{\neg \eventually \errConfig} = tr(\qwlp{S}{\identityOp} \rho), &&
        (L)ER^{\rewardMC{\rho}{P}{S}} (\eventually \termConfig)= tr(qw(l)p\llbracket S \rrbracket (P) \rho)
    \end{align*}
\end{lemma}
\begin{proof}
    Follows from Lemma~\ref{lem:operationalVsDenotational}, Lemma~\ref{lem:schroedinger} and Lemma~\ref{lem:wlpexistence}.
\end{proof}
We are interested in the conditional (liberal) expected reward of reaching $\termConfig$ from the initial state $\config{S}{\rho}$, conditioned on not visiting $\errConfig$:
\begin{align*}
    C(L)ER^{\rewardMC{\rho}{P}{S}} (\eventually \termConfig \mid \neg \eventually \errConfig):=& \frac{(L)ER^{\rewardMC{\rho}{P}{S}} (\eventually \termConfig)}{\prob{\rewardMC{\rho}{\identityOp}{S}}{\neg \eventually \errConfig}}
\end{align*}
This reward is equivalent to our interpretation of $qcw(l)p$, analogous to \cite[Theorem 5.7]{conditioningProb}:
\begin{theorem} \label{th:correspondenceMcCqwp}
    For a program $S$, state $\rho\in \densityFull$, predicates $P,Q\in\predicate$ we have
    \begin{align*}
        C(L)ER^{\rewardMC{\rho}{P}{S}} (\eventually \termConfig \mid \neg \eventually \errConfig) &= \fracTr( qcw(l)p\llbracket S \rrbracket (P,\identityOp) \multdot \rho)
    \end{align*}
\end{theorem}
\begin{proof}
    Assuming $tr(\qwlp{S}{\identityOp} \rho) >0$, then $C(L)ER^{\rewardMC{\rho}{P}{S}} (\eventually \termConfig \mid \neg \eventually \errConfig)$ is equal to
    \begin{align*}
        \frac{(L)ER^{\rewardMC{\rho}{P}{S}} (\eventually \termConfig)}{\prob{\rewardMC{\rho}{\identityOp}{S}}{\neg \eventually \errConfig}} = \frac{tr(qw(l)p \llbracket S \rrbracket (P) \rho)}{tr(\qwlp{S}{\identityOp} \rho)} = \fracTr(qcw(l)p\llbracket S \rrbracket (P,\identityOp) \multdot \rho)
    \end{align*}
    by Lemma~\ref{lem:ErWpEqual}.
    If $tr(\qwlp{S}{\identityOp} \rho)=0$, then $C(L)ER^{\rewardMC{\rho}{P}{S}} (\eventually \termConfig \mid \neg \eventually \errConfig)$ is undefined which means both sides of the statement are undefined and thus equal.
\end{proof}

\section{Examples}
\label{sec:example}
In this section we provide two examples on how conditional quantum weakest preconditions can be applied.
\subsection{The Quantum Fast-Dice-Roller}
In probabilistic programs, generating a uniform distribution using fair coins is a challenge. The fast dice roller efficiently simulates the throw of a fair dice, generating a uniform distribution about $N$ possible outcomes \cite{lumbroso}. We solve this problem for $N=6$ with quantum gates by creating qubits $q,p,r$ with Hadamard gates and using the observe statement to reject the $qp=11$ case, leaving $6$ possible outcomes ($qpr=000, \ldots , 101$), see Figure~\ref{fig:fdr}.

\begin{figure}
  \[\begin{array}{lcl}
      q   & \df & H q;\\
      p & \df & H p;\\
      \mathbf{observe}(q \otimes p, \identityOp_4 - \ketbra{11}{11});\hskip-2in {}\\
      r   & \df & H r
    \end{array}
  \]
  \caption{Quantum Fast-Dice-Roller. For the identity operator on $\mathcal{H}_2 \otimes \mathcal{H}_2$ we use $\identityOp_4$.}
  \label{fig:fdr}
\end{figure}

Before verifying its correctness, we consider the operational semantics:
We have three binary variables, so $\mathcal{H}_{all} = \mathcal{H}_2 ^{\otimes 3}$, denoted as $\mathcal{H}$. The first variable is $q$, the second one $p$ and the last one $r$ and $\rho_0 \in \densityFull$ an initial state. The operational semantics is represented by the Markov chain in Figure~\ref{fig:opSem}, Appendix~\ref{app:ex}.
To prove correctness, we focus on the probability of termination and reaching the desired state without violating the observation. This probability cannot be directly read from the operational semantics, even for this simple program. To specify this property formally, we use the reward MC as defined in Definition~\ref{def:opMC}. The desired probability can be computed using conditional weakest preconditions, see Theorem~\ref{th:correspondenceMcCqwp}.

To terminate in a state where the probability of all six outcomes is equal and forms a distribution, we verify that we reach the uniform superposition $\ket{\phi} = \sqrt{\frac{1}{6}}\big(\ket{000}+\ket{001}+\ket{010}+\ket{011}+\ket{100}+\ket{101}\big)$ over $6$ states. Measuring in the computational basis yields a uniform distribution. After computing the conditional weakest precondition, we can determine the likelihood of each input state reaching the fixed uniform superposition and producing a uniform distribution, assuming the observation is not violated.
We use the decoupling of $\qcwp{S}{(P,\identityOp)}$ and compute $\qwp{S}{P}$ and $\qwlp{S}{\identityOp}$ separately where $P = \ketbra{\phi}{\phi}$ and $S$ stands for our fast-dice roller program (Figure~\ref{fig:fdr}). The results of applying rules of Proposition~\ref{prop:wpDef} and~\ref{prop:wlpDef} can be found in Appendix~\ref{app:ex}.
The probability that an input state $\rho$ will reach the desired uniform superposition is $\fracTr(\qcwp{S}{(P,\identityOp)} \multdot \rho)$, that is
\begin{align*}
  \frac{tr(\qwp{S}{P}\rho)}{tr(\qwlp{S}{\identityOp}\rho)}=\begin{cases}
    1 &,\text{if }\rho=\ketbra{000}{000}\\
    0 &, \text{if }\rho=\ketbra{x}{x} \text{ with } x\in\{001,011,101,111\}\\
    0.1111 &,\text{if }\rho=\ketbra{x}{x} \text{ with } x\in\{010,100,110\}.
  \end{cases}
\end{align*}
$\rho=\ketbra{000}{000}$ will reach the desired superposition with probability $1$ assuming no observation is violated. We can also see that $tr(\qwp{S}{P}\ketbra{000}{000})\neq 1$ so even with the ``best'' input, our conditional weakest precondition calculus gives more information than $\qwp{S}{P}$.

\subsection{MAJ-SAT}
To demonstrate our approach, we will verify the correctness of a program that is used to solve MAJ-SAT. Unlike SAT, which asks whether there exists at least one satisfying assignment of a Boolean formula, MAJ-SAT asks whether a Boolean formula is satisfied by at least half of all possible variable assignments. MAJ-SAT is known to be PP-complete and \cite{postbqp} uses it to prove the equivalence of the complexity classes PostBQP and PP.

\begin{figure}
  \[\begin{array}{lcl}
      \overline{q} & \df & 0^{\otimes n};\\
      y & \df & 0;\\
      \overline{q} & \df & H^{\otimes n}\overline{q};\\
      \overline{q}y & \df & U_f \overline{q}y; \\
      \overline{q} & \df & H^{\otimes n}\overline{q};\\
      \mathbf{observe}(\overline{q}, \ket{0}\bra{0}^{\otimes n});\hskip-2in {} \\
      z & \df & 0;\\
      z & \df & R_k z;\\
      zy & \df & CH;\\
      \mathbf{observe}(y, \ket{1}\bra{1})\hskip-2in {}
    \end{array}
  \]
  \caption{Inner loop body $S_k$ of the quantum algorithm solving MAJ-SAT as presented in~\cite{postbqp}. $y,z$ are qubits, $\overline{q}$ is an $n$-qubit sized register (formally $n$ qubits $q_1,\dots,q_n$). We use $\overline{q} \df 0^{\otimes n}$ to denote that all $n$ qubits of $\overline{q}$ are set of $0$. $R_k = \frac{1}{\sqrt{1+4^k}}\begin{pmatrix} 1 & -2^k \\ 2^k & 1 \end{pmatrix}$ is a rotation matrix depending on the parameter $k$ and $CH$ is a controlled Hadamard.}
  \label{fig:program}
\end{figure}

Formally, we are faced with the following problem: A formula with $n$ variables can be represented by a function $f:\{0,1\}^n \to \{0,1\}$ with $s = \abs{\{f(x) = 1\}}$. The goal is to determine whether $s < 2^{n-1}$ holds or not.
Aaronson \cite{postbqp} presents a PostBQP algorithm for this problem. A PostBQP algorithm is one that runs in polynomial time, is allowed to perform measurements to check whether certain conditions are satisfied (analogous to our observe statement) and is required to produce the correct result with high probability conditioned on those measurements succeeding.
The algorithm from \cite{postbqp} is as follows:
\[\begin{array}{l}
      \text{for } k = -n , ... , n:\\
      \hspace{0.5cm} \text{repeat }n \text{ times:}\\
      \hspace{1cm} S_k \\
      \hspace{0.5cm}\text{if }S_k \text{ succeeded more than } \frac34 n \text{ times}:\\
      \hspace{1cm} \text{return true} \\
      \text{return false}
    \end{array}
  \]
  where $S_k$ is given in Figure~\ref{fig:program} and succeeding means that measuring $z$ in the $\ket{+}, \ket{-}$ basis returns $\ket{+}$.
The core idea is to show that $S_k$ succeeds with probability $\leq\frac12$ for all $k$ if $s\geq 2^{n-1}$ and with probability $\geq \Bigl(\frac{1+\sqrt2}{\sqrt 6}\Bigr)^2\geq 0.971$ for at least one $k$ otherwise.
Hence the overall algorithm solves MAJ-SAT. To keep this example manageable, we focus on the analysis of $S_k$ alone.

We use conditional weakest preconditions and determine $\qcwp{S_k}{P,\identityOp^{\otimes n+2}}$ (which depends on the parameters $n,s,k$).
Here the postcondition $P$ corresponds to $z$ being in state $\ket+$, formally $P=\ket{+}\bra{+}_z\otimes\identityOp$.
Then the probability that $S_k$ succeeds is $\Pr_{nsk}:=\fracTr(\qcwp{S_k}{P,\identityOp^{\otimes n+2}} \odot \rho)$ for initial state $\rho$.

We computed the cwp symbolically using a computer algebra system, but the resulting formulas were quite unreadable. So for the sake of this example, we present numerical results of computing cwp instead.
Since $S_k$ does not contain any loops, the cwp can be computed by mechanic application of the rules for observation, assignment, and application of unitaries.
Performing these calculations for selected values of $n$ and $s$ and all $k=-n,\dots,n$, we find that in each case the cwp is of the form $(c \identityOp, c'\identityOp)$ for some $c,c'\in\R$.
This is to be expected since all variables $\bar q,z,y$ are initialized at the beginning of the program,
so the cwp should not depend on the initial state, i.e., all matrices should be multiples of the identity.
In that case, $\Pr_{nsk}=c/c'$.
In \autoref{fig:table}, we show $\max_k \Pr_{nsk}$ for selected $s,n$.
(The claim from \cite{postbqp} is that the success probability of $S_k$ is $\geq 0.971$ for some $k$ if $s<2^{n-1}$ and $\leq 1/2$ for all $k$ otherwise, so we only care about the maximum over all $k$.)
We see that $\max_k \Pr_{nsk}$ is indeed $\geq 0.971$ and $\leq 1/2$ in those two cases.
This confirms the calculation from \cite{postbqp}, using our logic.
(At least for the values of $s,n$ we computed.)

\begin{figure}
  \begin{tabular}{ c | c c c c c c}
     &  $s=2$ & $s=3$ & $s=4$ & $s=7$ & $s=8$ & \dots \\
    \hline
    $n=2$ & $0.5$ & $0.3838$ & $0.3286$ &  \\
    $n=3$ & $\underline{0.9714}$ & $\underline{0.9991}$ & $0.5$ & $0.4247$ & $0.4123$ & \dots \\
    $n=4$ & $\underline{0.9991}$ & $\underline{0.9933}$ & $\underline{0.9714}$ & $\underline{0.9889}$ & $0.5$ & \dots \\
    $n=5$ & $\underline{0.9889}$ & $\underline{0.9828}$ & $\underline{0.9991}$ & $\underline{0.9977}$ & $\underline{0.9714}$ & \dots \\
    \dots & \dots& \dots & \dots & \dots & \ldots
   \end{tabular}
  \caption{Maximum of $\Pr_{nsk}={\fracTr(\qcwp{S_k}{P,\identityOp^{\otimes n+2}} \odot \rho)}$ for $k\in [-n,n]$. The cases where $s<2^{n-1}$ are underlined.}
  \label{fig:table}
\end{figure}

\section{Conclusion}
\label{sec:conclusion}

We introduced the observe statement in the quantum setting for infinite-dimensional cases, supported by operational, denotational and weakest precondition semantics. We defined conditional weakest preconditions, proved their equivalence to the operational semantics and applied them to an example using Bayesian inference. Future work includes the interpretation of predicates and exploration of alternatives to observe statements such as rejection sampling or hoisting in the probabilistic case. Additionally, the challenge of combining non-determinism with conditioning in probabilistic systems \cite{conditioningProb} may extend to quantum programs.



\bibliography{literature}

\begin{thebibliography}{10}

\bibitem{postbqp}
Scott Aaronson.
\newblock Quantum computing, postselection, and probabilistic polynomial-time.
\newblock {\em Proceedings of the Royal Society A: Mathematical, Physical and Engineering Sciences}, 461, 09 2005.
\newblock \href {https://doi.org/10.1098/rspa.2005.1546} {\path{doi:10.1098/rspa.2005.1546}}.

\bibitem{mcBible}
Christel Baier and Joost{-}Pieter Katoen.
\newblock {\em Principles of model checking}.
\newblock {MIT} Press, 2008.

\bibitem{boyd2004convex}
Stephen Boyd and Lieven Vandenberghe.
\newblock {\em Convex Optimization}.
\newblock Cambridge University Press, 2004.

\bibitem{conway2000courseOperator}
John~B. Conway.
\newblock {\em A Course in Operator Theory}.
\newblock Graduate Studies in Mathematics. American Mathematical Society, 2000.

\bibitem{conway1994}
John~B. Conway.
\newblock {\em A Course in Functional Analysis}.
\newblock Graduate Texts in Mathematics. Springer New York, 2007.
\newblock \href {https://doi.org/10.1007/978-1-4757-4383-8} {\path{doi:10.1007/978-1-4757-4383-8}}.

\bibitem{DENG202273}
Yuxin Deng and Yuan Feng.
\newblock Formal semantics of a classical-quantum language.
\newblock {\em Theoretical Computer Science}, 913:73--93, 2022.
\newblock \href {https://doi.org/10.1016/j.tcs.2022.02.017} {\path{doi:10.1016/j.tcs.2022.02.017}}.

\bibitem{DHondtWeakestPreconditions}
Ellie D'Hondt and Prakash Panangaden.
\newblock Quantum weakest preconditions.
\newblock {\em Math. Struct. Comput. Sci.}, 16(3):429--451, 2006.

\bibitem{NeedOfToolsDebuggingQuantumPrograms}
Olivia Di~Matteo.
\newblock On the need for effective tools for debugging quantum programs.
\newblock In {\em Proceedings of the 5th ACM/IEEE International Workshop on Quantum Software Engineering}, Q-SE 2024, page 17–20, New York, NY, USA, 2024. Association for Computing Machinery.
\newblock \href {https://doi.org/10.1145/3643667.3648226} {\path{doi:10.1145/3643667.3648226}}.

\bibitem{Dijkstra75}
Edsger~W. Dijkstra.
\newblock Guarded commands, nondeterminacy and formal derivation of programs.
\newblock {\em Commun. ACM}, 18(8):453–457, 1975.
\newblock \href {https://doi.org/10.1145/360933.360975} {\path{doi:10.1145/360933.360975}}.

\bibitem{Dijkstra76}
Edsger~W. Dijkstra.
\newblock {\em A Discipline of Programming}.
\newblock Prentice-Hall, 1976.

\bibitem{Dirac}
Paul A.~M. Dirac.
\newblock {\em The Principles of Quantum Mechanics}.
\newblock Clarendon Press, Oxford, 1930.

\bibitem{FengNondeterministicQuantumVerification}
Yuan Feng and Yingte Xu.
\newblock Verification of nondeterministic quantum programs.
\newblock In {\em Proceedings of the 28th ACM International Conference on Architectural Support for Programming Languages and Operating Systems, Volume 3}, ASPLOS 2023, page 789–805, New York, NY, USA, 2023. Association for Computing Machinery.
\newblock \href {https://doi.org/10.1145/3582016.3582039} {\path{doi:10.1145/3582016.3582039}}.

\bibitem{FengQHLClassicalVars}
Yuan Feng and Mingsheng Ying.
\newblock Quantum {H}oare logic with classical variables.
\newblock {\em ACM Transactions on Quantum Computing}, 2(4), 2021.
\newblock \href {https://doi.org/10.1145/3456877} {\path{doi:10.1145/3456877}}.

\bibitem{benniDiss}
Benjamin~L. Kaminski.
\newblock {\em Advanced Weakest Precondition Calculi for Probabilistic Programs}.
\newblock PhD thesis, 02 2019.
\newblock \href {https://doi.org/10.18154/RWTH-2019-01829} {\path{doi:10.18154/RWTH-2019-01829}}.

\bibitem{KOZEN1985162}
Dexter Kozen.
\newblock A probabilistic {PDL}.
\newblock {\em Journal of Computer and System Sciences}, 30(2):162--178, 1985.
\newblock \href {https://doi.org/10.1016/0022-0000(85)90012-1} {\path{doi:10.1016/0022-0000(85)90012-1}}.

\bibitem{verificationSummary}
Marco Lewis, Sadegh Soudjani, and Paolo Zuliani.
\newblock Formal verification of quantum programs: Theory, tools, and challenges.
\newblock {\em ACM Transactions on Quantum Computing}, 2023.
\newblock \href {https://doi.org/10.1145/3624483} {\path{doi:10.1145/3624483}}.

\bibitem{projectionAssertions}
Gushu Li, Li~Zhou, Nengkun Yu, Yufei Ding, Mingsheng Ying, and Yuan Xie.
\newblock Projection-based runtime assertions for testing and debugging quantum programs.
\newblock {\em Proc. ACM Program. Lang.}, 4(OOPSLA), 2020.
\newblock \href {https://doi.org/10.1145/3428218} {\path{doi:10.1145/3428218}}.

\bibitem{lumbroso}
J{\'{e}}r{\'{e}}mie~O. Lumbroso.
\newblock Optimal discrete uniform generation from coin flips, and applications.
\newblock {\em CoRR}, abs/1304.1916, 2013.

\bibitem{McIverWpProb}
Annabelle McIver and Carroll Morgan.
\newblock {\em Abstraction, Refinement and Proof for Probabilistic Systems}.
\newblock Monographs in Computer Science. Springer, 2005.
\newblock \href {https://doi.org/10.1007/b138392} {\path{doi:10.1007/b138392}}.

\bibitem{Nielsen_Chuang_2010}
Michael~A. Nielsen and Isaac~L. Chuang.
\newblock {\em Quantum Computation and Quantum Information: 10th Anniversary Edition}.
\newblock Cambridge University Press, 2010.
\newblock \href {https://doi.org/10.1017/CBO9780511976667} {\path{doi:10.1017/CBO9780511976667}}.

\bibitem{Nori}
Aditya~V. Nori, Chung-Kil Hur, Sriram~K. Rajamani, and Selva Samuel.
\newblock R2: an efficient mcmc sampler for probabilistic programs.
\newblock In {\em Proceedings of the Twenty-Eighth AAAI Conference on Artificial Intelligence}, AAAI'14, page 2476–2482. AAAI Press, 2014.

\bibitem{conditioningProb}
Federico Olmedo, Friedrich Gretz, Nils Jansen, Benjamin~L. Kaminski, Joost-Pieter Katoen, and Annabelle Mciver.
\newblock Conditioning in probabilistic programming.
\newblock {\em ACM Trans. Program. Lang. Syst.}, 40(1), 2018.
\newblock \href {https://doi.org/10.1145/3156018} {\path{doi:10.1145/3156018}}.

\bibitem{takesaki1979theory}
Masamichi Takesaki.
\newblock {\em Theory of Operator Algebras I}.
\newblock Number Bd. 1 in Encyclopaedia of Mathematical Sciences. Springer New York, 1979.
\newblock \href {https://doi.org/10.1007/978-1-4612-6188-9} {\path{doi:10.1007/978-1-4612-6188-9}}.

\bibitem{isabelleproof}
Dominique Unruh.
\newblock The tensor product on {H}ilbert spaces.
\newblock {\em Arch. Formal Proofs}, 2024.
\newblock URL: \url{https://www.isa-afp.org/entries/Hilbert\_Space\_Tensor\_Product.html}.

\bibitem{heisenbergdualityUnruh}
Dominique Unruh.
\newblock Quantum references, 2024.
\newblock \href {https://arxiv.org/abs/2105.10914v3} {\path{arXiv:2105.10914v3}}.

\bibitem{YanIncorrectnessLogic}
Peng Yan, Hanru Jiang, and Nengkun Yu.
\newblock On incorrectness logic for quantum programs.
\newblock {\em Proc. ACM Program. Lang.}, 6(OOPSLA1), 2022.
\newblock \href {https://doi.org/10.1145/3527316} {\path{doi:10.1145/3527316}}.

\bibitem{floydHoareLogic}
Mingsheng Ying.
\newblock Floyd-{H}oare logic for quantum programs.
\newblock {\em ACM Trans. Program. Lang. Syst.}, 2012.
\newblock \href {https://doi.org/10.1145/2049706.2049708} {\path{doi:10.1145/2049706.2049708}}.

\bibitem{YingPredicateTranserformerSemantics}
Mingsheng Ying, Runyao Duan, Yuan Feng, and Zhengfeng Ji.
\newblock Predicate transformer semantics of quantum programs.
\newblock {\em Semantic Techniques in Quantum Computation}, 2010.
\newblock \href {https://doi.org/10.1017/CBO9781139193313.009} {\path{doi:10.1017/CBO9781139193313.009}}.

\bibitem{ZhouAppliedQHL}
Li~Zhou, Nengkun Yu, and Mingsheng Ying.
\newblock An applied quantum {H}oare logic.
\newblock In {\em Proceedings of the 40th ACM SIGPLAN Conference on Programming Language Design and Implementation}, PLDI 2019, page 1149–1162, New York, NY, USA, 2019. Association for Computing Machinery.
\newblock \href {https://doi.org/10.1145/3314221.3314584} {\path{doi:10.1145/3314221.3314584}}.

\end{thebibliography}

\appendix
\section{Appendix}
\label{sec:app}
\subsection{Proofs Concerning the Semantics}
\label{app:semantics}

Before proving the properties of the denotational semantics, we need to show an auxiliary lemma:
\begin{lemma}
    \label{lem:lupislinear}
    \begin{enumerate}
        \item If for every $n\in \mathbb{N}$, $f_n: \density \to \density$ is bounded linear and pointwise increasing, that means for every fixed $\rho \in \density$ $m>n$ implies $f_n(\rho) \sqsubseteq f_m(\rho)$, then $f_\infty (\rho) := \bigvee_{n=0}^\infty f_n(\rho)$ exists and $f_\infty$ is linear.
        \item If for every $n \in \mathbb{N}$, $e_n:\densityNumberPairs \to \mathbb{R}_{\geq 0}$ is bounded linear and pointwise increasing, then $e_\infty (\rho,p) := \bigvee_{n=0}^\infty e_n(\rho,p)$ exists for every $(\rho,p) \in \densityNumberPairs$ and $e_\infty$ is linear.
    \end{enumerate}

\end{lemma}
\begin{proof}
    \begin{enumerate}
        \item First of all, $f_\infty (\rho)$ exists because $(\density, \sqsubseteq)$ is an $\omega$-cpo and $\{f_n(\rho)\}_{n\in \mathbb{N}}$ increasing.
        Then we can write $f_\infty (\rho)$ as a sum of $g_n(\rho) := \begin{cases} f_0(\rho) &,\text{if } n=0\\
            f_n(\rho)-f_{n-1}(\rho) &, \text{if }n>0\end{cases}$ because
        \begin{equation*}
            f_\infty(\rho) = \bigvee_{n=0}^\infty f_n(\rho) \overset{\text{\tiny\cite[Lem. 30]{heisenbergdualityUnruh}}}{=}\lim_{n \to \infty} f_n(\rho) = \lim_{n\to \infty} \sum_{j=0}^n g_j(\rho) = \sum_{n \in \mathbb{N}}g_n(\rho).
        \end{equation*}
        Each $g_n$ is a linear function an thus has an extension $\overline{g_n}= \overline{f_n} - \overline{f_{n-1}}$ from $\ext(\density) \to \ext(\density)$. Each $\rho\in \ext(\density)$ can be written as a finite linear combination of $\rho'_i \in \density$, i.e., $\rho = \sum_i \lambda_i \rho'_i$ by definition of $\ext$. Then
        \begin{align*}
            \sum_i \lambda_i \sum_n g_n(\rho'_i)=  \sum_i \lambda_i \sum_n \overline{g_n}(\rho'_i) \overset{(*)}{=} \sum_n \sum_i \lambda_i \overline{g_n}(\rho'_i)= \sum_n \overline{g_n}(\rho) =: \overline{f_\infty}(\rho)
        \end{align*}
        where $(*)$ uses that $\sum_i$ is finite.
        The left sum exists (see above), thus the right hand side exists too.
        To show that $f_\infty$ is linear, we show that $\overline{f_\infty}$ is a linear extension of $f_\infty$. It is $\overline{f_\infty}(\rho)=\sum_n \overline{g_n}(\rho) = \sum_n g_n(\rho) = f_\infty(\rho)$ for $\rho \in \density$. Also \begin{align*}
            &\phantom{=.}\overline{f_\infty}(a \rho + \sigma)
            = \sum_n \overline{g_n}(a \rho + \sigma)
            = \sum_n a \overline{g_n}(\rho) + \overline{g_n}(\sigma)\\
            &= a \sum_n \overline{g_n}(\rho) + \sum_n \overline{g_n}(\sigma)
            = a \overline{f_\infty}(\rho) + \overline{f_\infty}(\sigma)
        \end{align*}
        where all sums exist as shown already. That concludes that $f_\infty$ is linear.

        \item The existence of $e_\infty (\rho,p)$ is clear as it is the least upper bound of a bounded set of real numbers. We can write $e_\infty (\rho,p)$ as a sum of $h_n(\rho,p) := \begin{cases} e_0(\rho,p) &, \text{if }n=0\\
            e_n(\rho,p)-e_{n-1}(\rho,p) &,\text{if } n>0 \end{cases}$ because
        \begin{equation*}
            e_\infty(\rho,p) = \bigvee_{n=0}^\infty e_n(\rho,p) = \lim_{n \to \infty} e_n(\rho,p) = \lim_{n\to \infty} \sum_{j=0}^n h_j(\rho,p) = \sum_{n \in \mathbb{N}}h_n(\rho,p).
        \end{equation*}
        Each $h_n$ is a linear function and thus has an extension $\overline{h_n}= \overline{e_n} - \overline{e_{n-1}}$ from $\ext(\densityNumberPairs) \to \ext(\mathbb{R}_{\geq 0})$. Each $(\rho,p) \in \ext(\densityNumberPairs)$ can be written as a finite linear combination of $(\rho'_i,p_i) \in \densityNumberPairs$, i.e., $(\rho,p) = \sum_i \lambda_i (\rho'_i,p_i)$ by definition of $\ext$. Then
        \begin{align*}
            &\phantom{=} \sum_i \lambda_i \sum_n h_n(\rho'_i,p_i)=  \sum_i \lambda_i \sum_n \overline{h_n}(\rho'_i,p_i) \overset{\text{fin. sum}}{=} \sum_n \sum_i \lambda_i \overline{h_n}(\rho'_i,p_i)\\
            &= \sum_n \overline{h_n}(\rho,p) =: \overline{e_\infty}(\rho,p).
        \end{align*}
        The left sum exists (see above), thus the right hand side exists too. To show that $e_\infty$ is linear, we show that $\overline{e_\infty}$ is a linear extension of $e_\infty$
        It is $\overline{e_\infty}(\rho,p)=\sum \overline{h_n}(\rho,p) = \sum_n h_n(\rho,p) = e_\infty(\rho,p)$ for $(\rho,p) \in \densityNumberPairs$. Also \begin{align*}
            &\phantom{=.}\overline{e_\infty}(a (\rho,p) + (\sigma,q))
            = \sum_n \overline{h_n}(a (\rho,p) + (\sigma,q))
            = \sum_n a \overline{h_n}(\rho,p) + \overline{h_n}(\sigma,q)\\
            &= a \sum_n \overline{h_n}(\rho,p) + \sum_n \overline{h_n}(\sigma,q)
            = a \overline{e_\infty}(\rho,p) + \overline{e_\infty}(\sigma,q)
        \end{align*}
        where all sums exist as shown already. That concludes that $e_\infty$ is linear.
    \end{enumerate}
\end{proof}

We now prove Proposition~\ref{pro:allClaims}:
\begin{proof}
    We show all statements together by doing an induction over the structure of $S$.
    \begin{itemize}
        \item For $S \equiv \skipbf$:
        \begin{enumerate}
            \item $\semanticsRho{\skipbf}(\rho,p) =\rho = \semanticsRho{\skipbf}(\rho,q)$

            \item $p \leq p = \semanticsErr{\skipbf}(\rho,p)$

            \item $\semanticsErr{\skipbf}(\rho,q+p) = q+p = \semanticsErr{\skipbf}(\rho,q) + p$

            \item $\tilde{tr}(\semantics{\skipbf}(\rho,p)) = tr(\semanticsRho{\skipbf}(\rho,p)) + \semanticsErr{\skipbf}(\rho,p) = tr(\rho) +p = \tilde{tr}(\rho,p)$

            \item $(\rho,p) \in \densityNumberPairs$ implies $\semantics{\skipbf}(\rho,p) = (\rho,p) \in \densityNumberPairs$

            \item Linearity:
            We define $\overline{\semantics{\skipbf}}: \ext(\densityNumberPairs) \to \ext(\densityNumberPairs)$ as $\overline{\semantics{\skipbf}}(\rho,p) = (\rho,p)$ for $(\rho,p) \in \ext(\densityNumberPairs)$ which is linear and equal to $\semantics{\skipbf}(\rho,p)$ for $(\rho,p) \in \densityNumberPairs$, thus $\semantics{\skipbf}$ is linear by definition.
        \end{enumerate}
        \item For $S\equiv \qzero$: We only show the case of $type(q)=\intType$, for $\boolType$ it is similar. First of all, we show convergence of the sum.
        We know by \cite[Lem. 30]{heisenbergdualityUnruh} that the supremum and limit coincides in the SOT and if the trace of each element is upper bounded, then the limit exists. That means, we have to show that $tr(\sum_{n\in F} \ket{0} \bra{n}_q \rho \ket{n} \bra{0}_q)$ is bounded for every finite set $F\subseteq \mathbb{Z}$:
        \begin{align*}
            &\phantom{=.}tr\left(\sum_{n\in F} \ket{0} \bra{n}_q \rho \ket{n} \bra{0}_q\right) = \sum_{n\in F} tr\left(\ket{0} \bra{n}_q \rho \ket{n} \bra{0}_q\right) = \sum_{n\in F} tr\left(\rho \ket{n} \bra{0}_q \ket{0} \bra{n}_q  \right) \\
            &= tr\left(\rho \sum_{n\in F} \ket{n} \bra{0}_q \ket{0} \bra{n}_q \right)
            \leq \norm{\rho}_{tr}  \norm{\sum_{n\in F} \ket{n} \bra{0}_q \ket{0} \bra{n}_q}_{op} \leq \norm{\rho}_{tr} \norm{\identityOp}_{op} \leq \norm{\rho}_{tr}
        \end{align*}
        Then \begin{align*}
            \bigvee_{\text{finite }F \subseteq \mathbb{Z}} \sum_{n \in F} \ket{0} \bra{n}_q \rho \ket{n} \bra{0}_q = \sum_{n\in \mathbb{Z}} \ket{0} \bra{n}_q \rho \ket{n} \bra{0}_q.
        \end{align*}
        \begin{enumerate}
            \item $\semanticsRho{\qzero}(\rho,p) = \semanticsRho{\qzero}(\rho,q)$ follows directly from definition

            \item $p \leq p = \semanticsErr{\qzero}(\rho,p)$

            \item $\semanticsErr{\qzero}(\rho,q+p) = q+p =\semanticsErr{\qzero}(\rho,q) + p$

            \item
            \begin{align*}
                &\phantom{=.}\tilde{tr}\left(\semantics{\qzero}(\rho,p)\right) = tr\left(\semanticsRho{ \qzero } (\rho ,p)\right) +\semanticsErr{ \qzero } (\rho ,p)\\
                &=tr\left(\sum_{n\in \mathbb{Z}} \ket{0} \bra{n}_q \rho \ket{n} \bra{0}_q\right) + p= p+ \sum_{n\in \mathbb{Z}} tr\left(\ket{0} \bra{n} _q \rho \ket{n} \bra{0}_q\right)\\
                &=p+ \sum_{n\in \mathbb{Z}} tr\left(\ket{n} \bra{0}_q \ket{0} \bra{n} _q \rho \right) = p+ \sum_{n\in \mathbb{Z}} tr\left(\ket{n}\bra{n}_q \rho \right)\\
                &= p+ tr\left([\sum_{n\in \mathbb{Z}} \ket{n}\bra{n}_q ]\rho \right)= p+tr(\rho) = \tilde{tr}(\rho,p)
            \end{align*}

            \item
             $(\rho,p) \in \densityNumberPairs$ implies $\tilde{tr}(\semantics{\qzero}(\rho,p))\leq \tilde{tr}(\rho,p)\leq 1$ and $ \semanticsErr{\qzero}(\rho,p) \geq p \geq 0$.

            If $\rho$ is positive, then $\ket{0} \bra{n}_q \rho \ket{n} \bra{0}_q$ is positive. The infinite sum of positive operators $A,B$ is positive again,
            so $\semanticsRho{\qzero}(\rho,p)$ is positive and $\semantics{\qzero}(\rho,p) \in \densityNumberPairs$.

            \item Linearity:
            We define $\overline{\semantics{\qzero}}: \ext(\densityNumberPairs) \to \ext(\densityNumberPairs)$ as
            \begin{align*}
                 \overline{\semantics{\qzero}}(\rho,p) =\left(\sum_{n\in \mathbb{Z}} \ket{0} \bra{n}_q \rho \ket{n} \bra{0}_q, p\right)
            \end{align*} for $(\rho,p) \in \ext(\densityNumberPairs)$ which is linear and equal to $\semantics{\qzero}(\rho,p)$ for $(\rho,p) \in \densityNumberPairs$, thus $\semantics{\qzero}$ is linear by definition.
        \end{enumerate}
        \item For $S \equiv \Uq$:
        \begin{enumerate}
            \item $\semanticsRho{\Uq}(\rho,p) =U \rho U^\dagger = \semanticsRho{\Uq}(\rho,q)$

            \item $p \leq p = \semanticsErr{\Uq}(\rho,p)$

            \item $\semanticsErr{\Uq}(\rho,q+p) =q+p = \semanticsErr{\Uq}(\rho,q) + p$

            \item $\tilde{tr}(\semantics{\Uq}(\rho,p)) = tr\left(\semanticsRho{\Uq} (\rho,p)\right)+ \semanticsErr{ \Uq} (\rho,p) = tr(U\rho U^\dagger) +p = tr(U^\dagger U\rho ) +p  =tr(\rho)+p = \tilde{tr}(\rho,p)$

            \item $(\rho,p) \in \densityNumberPairs$ implies $\tilde{tr}(\semantics{\Uq}(\rho,p))\leq \tilde{tr}(\rho,p)\leq 1$ and $ \semanticsErr{\Uq}(\rho,p) \geq p \geq 0$.

            It is $tr(\semanticsRho{\Uq}(\rho,p)) = tr(U^\dagger U\rho ) = tr(\rho)$ and thus $tr(\semanticsRho{\Uq}(\rho,p)) \leq 1$. If $\rho$ is positive, then $U\rho U^\dagger = \semanticsRho{S}(\rho,p)$ is positive. In total, $\semanticsRho{\Uq}(\rho,p) \in \density$ and $\semantics{\Uq}(\rho,p) \in \densityNumberPairs$.

            \item Linearity:
            We define $\overline{\semantics{\Uq}}: \ext(\densityNumberPairs) \to \ext(\densityNumberPairs)$ as
            \begin{align*}\overline{\semantics{\Uq}}(\rho,p) =(U \rho U^\dagger, p)
            \end{align*} for $(\rho,p) \in \ext(\densityNumberPairs)$ which is linear and equal to $\semantics{\Uq}(\rho,p)$ for $(\rho,p) \in \densityNumberPairs$, thus $\semantics{\Uq}$ is linear by definition.

        \end{enumerate}
        \item For $S \equiv \observe$:
        \begin{enumerate}
            \item $\semanticsRho{\observe}(\rho,p) = O \rho O^\dagger = \semanticsRho{\observe}(\rho,q)$

            \item $p \leq p+ tr(\rho)-tr(O \rho O^\dagger) = \semanticsErr{\observe}(\rho,p)$ because $tr(\rho) = tr(\identityOp\rho)\geq tr(O^\dagger O \rho) = tr(O \rho O^\dagger)$

            \item $\semanticsErr{\observe}(\rho,q+p) = q+p+tr(\rho)- tr(O\rho O^\dagger)  = \semanticsErr{\observe}(\rho,q) + p$

            \item \begin{align*}
                &\phantom{=.}\tilde{tr}(\semantics{\observe}(\rho,p)) \\
                &= tr(\semanticsRho{\observe} (\rho,p))+ \semanticsErr{ \observe} (\rho,p) \\
                &= tr(O \rho O^\dagger) + p+ tr(\rho)-tr(O\rho O^\dagger) = p+ tr(\rho) = \tilde{tr}(\rho,p)
            \end{align*}

            \item It is $\tilde{tr}(\semantics{\observe}(\rho,p))\leq \tilde{tr}(\rho,p)\leq 1$ and $ \semanticsErr{\observe}(\rho,p)\geq p \geq 0$ for $(\rho,p) \in \densityNumberPairs$.

            Also $tr(\semanticsRho{\observe}(\rho,p)) = tr(O^\dagger O\rho ) \leq tr(\rho)\leq 1$. If $\rho$ is positive, then $O \rho O^\dagger = \semanticsRho{S}(\rho,p)$ is positive and $\semantics{\observe}(\rho,p) \in \densityNumberPairs$.

            \item Linearity:
            We define $\overline{\semantics{\observe}}: \ext(\densityNumberPairs) \to \ext(\densityNumberPairs)$ by
            \begin{equation*}
                \overline{\semantics{\observe}}(\rho,p)=\left(O \rho O^\dagger, p + tr(\rho)-tr(O \rho O^\dagger)\right)
            \end{equation*}
            which is linear and equal to $\semantics{\observe}(\rho,p)$ for $(\rho,p) \in \densityNumberPairs$, so $\semantics{\observe}$ is linear by definition.
        \end{enumerate}
        \item For $S \equiv \concat$:
        \begin{enumerate}
            \item $\semanticsRho{\concat}(\rho,p) =\semanticsRho{S_2} (\semantics{S_1} (\rho,p ) )
            = \semanticsRho{S_2}(\rho',p')
            = \semanticsRho{S_2}(\rho',q')
            = \semanticsRho{S_2}(\semantics{S_1}(\rho,q)) =\semanticsRho{\concat}(\rho,q)$
            where $\semantics{S_1}(\rho,p)=(\rho',p')$ and $\semantics{S_1}(\rho,q) = (\rho',q')$.

            \item $\semanticsErr{\concat}(\rho,p)
                = \semanticsErr{S_2}(\semantics{S_1}(\rho,p))
                = \semanticsErr{S_2}(\semanticsRho{S_1}(\rho,p), \semanticsErr{S_1}(\rho,p))
                \geq  \semanticsErr{S_1}(\rho,p)\geq p$

            \item
            \begin{align*}
                &\phantom{=.}\semanticsErr{\concat}(\rho,q+p) \\
                &= \semanticsErr{S_2}(\semantics{S_1}(\rho,q+p)) \\
                &= \semanticsErr{S_2}((\semanticsRho{S_1}(\rho,q+p), \semanticsErr{S_1}(\rho,q+p)))\\
                &= \semanticsErr{S_2}((\semanticsRho{S_1}(\rho,q), \semanticsErr{S_1}(\rho,q)+p)) \\
                &= \semanticsErr{S_2}((\semanticsRho{S_1}(\rho,q), \semanticsErr{S_1}(\rho,q))) +p \\
                &= \semanticsErr{S_2}(\semantics{S_1}(\rho,q))+p\\
                &= \semanticsErr{\concat}(\rho,q) + p
            \end{align*}

            \item $\tilde{tr}(\semantics{\concat}(\rho,p)) \leq  \tilde{tr}(\semantics{S_1}(\rho,p)) \leq \tilde{tr}(\rho,p)$

            \item $(\rho,p) \in \densityNumberPairs$ implies $\tilde{tr}(\semantics{\concat}(\rho,p))\leq \tilde{tr}(\rho,p)\leq 1$ and $ \semanticsErr{\concat}(\rho,p)\geq p \geq 0$.
            Also
            \begin{align*}
                \semanticsRho{\concat}(\rho,p) = \semanticsRho{S_2} (\semantics{S_1} (\rho,p )) = \semanticsRho{S_2} (\semanticsRho{S_1} (\rho,p ),0)
            \end{align*}
            As $\semanticsRho{S_1} (\rho,p ) \in \density$ and $\semanticsRho{S_2}(\rho,p) \in \density$ by induction hypothesis for all $\rho \in \density$, is  $\semanticsRho{S_2} (\semanticsRho{S_1} (\rho,p ),0) \in \density$ as well and thus $\semantics{\concat}(\rho,p) \in \densityNumberPairs$.

            \item Linearity:
            We know that $\semantics{S_1}$ and $\semantics{S_2}$ are linear, thus there exist linear functions $\overline{\semantics{S_1}}$ and $\overline{\semantics{S_2}}$ with $\overline{\semantics{S_1}}(\rho,p) = \semantics{S_1}(\rho,p)$ and $\overline{\semantics{S_2}}(\rho,p) = \semantics{S_2}(\rho,p)$ for $(\rho,p)\in \densityNumberPairs$. We define $\overline{\semantics{\concat}}: \ext(\densityNumberPairs) \to \ext(\densityNumberPairs)$ as $\overline{\semantics{\concat}}(\rho,p) = (\overline{\semantics{S_2}}(\overline{\semantics{S_1}}(\rho,p)))$ for $(\rho,p) \in \ext(\densityNumberPairs)$ which is linear. Also
            \begin{align*}
                \overline{\semantics{S_2}}\left(\overline{\semantics{S_1}}(\rho,p)\right) = \overline{\semantics{S_2}}\left(\semantics{S_1}(\rho,p)\right) = \semantics{S_2}(\semantics{S_1}(\rho,p)) = \semantics{\concat}(\rho,p)
            \end{align*}
            for $(\rho,p) \in \densityNumberPairs$, thus $\semantics{\concat}$ is linear by definition.
        \end{enumerate}
        \item For $S \equiv \measurePrime$:
        First of all, we show convergence of the sum.
        We know by \cite[Lem. 30]{heisenbergdualityUnruh} that the supremum and limit coincides in the SOT and if the trace of each element is upper bounded, then the limit exists. That means, we have to show that $tr(\sum_{m \in M'} \semanticsRho{S'_m}(M_m \rho M_m^\dagger))$ is bounded for every finite set $M'\subseteq M$:
        \begin{align*}
            &\phantom{=.}tr\left(\sum_{m \in M'} \semanticsRho{S'_m}(M_m \rho M_m^\dagger,0)\right) = \sum_{m \in M'}tr\left( \semanticsRho{S'_m}(M_m \rho M_m^\dagger,0)\right) \leq \sum_{m \in M'}tr( M_m \rho M_m^\dagger) \\
            &= tr\left(\rho \sum_{m \in M'} M_m^\dagger M_m\right) \leq \norm{\rho}_{tr} \norm{\sum_{m \in M'} M_m^\dagger M_m}_{op} \leq \norm{\rho}_{tr} \norm{\identityOp}_{op} = \norm{\rho}_{tr}
        \end{align*}
        because we can apply all previous shown properties for $S'_m$ by induction hypothesis.
        Then \begin{align*}
            \bigvee_{\text{finite }M' \subseteq M} \sum_{m \in M'} \semanticsRho{S'_m}(M_m \rho M_m^\dagger) = \sum_{m \in M} \semanticsRho{S'_m}(M_m \rho M_m^\dagger).
        \end{align*}
        \begin{enumerate}
            \item $\semanticsRho{\measurePrime}(\rho,p) = \sum_m \semanticsRho{S'_m}(M_m \rho M_m ^\dagger ,0) = \semanticsRho{\measurePrime}(\rho,q)$

            \item $p \leq \sum_m 0 + p \leq \sum_m \semanticsErr{S_m } (M_m \rho M_m ^\dagger ,0) +p =\semanticsErr{\measurePrime}(\rho,p)$

            \item $\semanticsErr{\measurePrime}(\rho,q+p) =\sum_m \semanticsErr{S_m } (M_m \rho M_m ^\dagger ,0) +q+p $\\
            $= \semanticsErr{\measurePrime}(\rho,q) + p$

            \item
            \begin{align*}
                &\phantom{=.}\tilde{tr}\left(\semantics{\measurePrime}(\rho,p)\right) = \tilde{tr}\left(\sum_m \semantics{S'_m}(M_m \rho M_m ^\dagger,0) +(\zeroOp,p)\right) \\
                &= \sum_m  \tilde{tr}\left(\semantics{S'_m}(M_m \rho M_m ^\dagger,0)\right) + \tilde{tr}(\zeroOp,p)
                \leq \sum_m \tilde{tr}(M_m \rho M_m ^\dagger,0) + \tilde{tr}(\zeroOp,p) \\
                &= \sum_m tr(M_m ^\dagger M_m \rho )+0 + tr(\zeroOp)+p =  tr\left(\sum_m M_m ^\dagger M_m \rho \right)+p = \tilde{tr}(\rho,p)
            \end{align*}

            \item $\tilde{tr}(\semantics{\measurePrime}(\rho,p))\leq \tilde{tr}(\rho,p)\leq 1$ and $ \semanticsErr{\measurePrime}(\rho,p)\geq p \geq 0$ for $(\rho,p) \in \densityNumberPairs$.
            It is $\semanticsRho{\measurePrime}(\rho,p) = \sum_m \semanticsRho{S'_m}(M_m \rho M_m ^\dagger ,0)$ and $\semanticsRho{S'_m}(M_m \rho M_m ^\dagger ,0) \in \density$ for all $m$ by induction hypothesis. It follows that $\semanticsRho{\measurePrime}(\rho,p)$ is positive because the infinite sum of positive operators is positive.
            Together, $\semantics{\measurePrime}(\rho,p) \in \densityNumberPairs$.

            \item Linearity:
            We know that $\semantics{S'_m}$ is linear for each $m\in I$, thus there exist linear functions $\overline{\semantics{S'_m}}$ with $\overline{\semantics{S'_m}}(\rho,p) = \semantics{S'_m}(\rho,p)$ for $(\rho,p)\in \densityNumberPairs$. We define $\overline{\semantics{\measurePrime}}: \ext(\densityNumberPairs) \to \ext(\densityNumberPairs)$ as $\overline{\semantics{\measurePrime}}(\rho,p) = \sum_m \overline{\semantics{S'_m}}(M_m \rho M_m ^\dagger,p)$ for $(\rho,p) \in \ext(\densityNumberPairs)$ which is linear. Every $\rho\in \ext(\density)$ can be written as a finite linear combination of $\rho'\in \density$, thus the existence of the sum follows.
             Also
            \begin{align*}
                &\phantom{=.}\overline{\semantics{\measurePrime}}(\rho,p) =  \sum_m \overline{\semantics{S'_m}}(M_m \rho M_m ^\dagger,p) \\
                &= \sum_m \semantics{S'_m}(M_m \rho M_m^\dagger,p) = \semantics{\measurePrime}(\rho,p)
            \end{align*}
            for $(\rho,p) \in \densityNumberPairs$, thus $\semantics{\measurePrime}$ is linear by definition.
        \end{enumerate}
        \item For $S \equiv \while'$:
        By induction hypothesis, we have that propositions 1-6 hold for $S'$. Since $(\while')^n$ can be expressed as a term involving only $S',\Omega$, sequential composition and measurements (and since $\Omega$ trivially satisfies $1$-$6$), we derive propositions $1$-$6$ for $(\while')^n$. We refer to this by using $(*)$.

        Before continuing to prove propositions 1-6, we show the existence of the least upper bounds by first showing monotonicity of $\semanticsRho{\while}$:
        \begin{claim}
        \label{cl:monoton}
        For each measurement $M=\{M_0,M_1\}$, quantum register $\Bar{q}$ and program $S'$, $\semanticsRho{\while'}$ is monotonic, i.e., for all $n \geq 0$ and $ (\rho,p) \in \densityNumberPairs$ it is
        \begin{align*}
            \semanticsRho{(\while')^{n}}(\rho,p) &\sqsubseteq \semanticsRho{(\while')^{n+1}}(\rho,p)
        \end{align*}
        \end{claim}
        \begin{claimproof}
            First of all, we show
            \begin{align*}
                \semanticsRho{(\while')^n}(\rho,p) = \sum_{k=0}^{n-1} M_0 \left(f^k(\rho)\right) M_0 ^\dagger
            \end{align*}
            by induction over $n\geq 1$ where $f(\rho):=\semanticsRho{S'}(M_1 \rho M_1^\dagger,0)$.
            \begin{itemize}
                \item $n=1$:
            \begin{align*}
                \semanticsRho{(\while')^1}(\rho,p) =\semanticsRho{\ifstatement{S'; \Omega}{\skipbf}}(\rho,p)
            \end{align*}
            implies
            \begin{align*}
                &\phantom{=.}\semanticsRho{\ifstatement{S'; \Omega}{\skipbf}}(\rho,p) \\
                &= \semanticsRho{\skipbf } (M_0 \rho M_0 ^\dagger ,0) + \semanticsRho{S'; \Omega} (M_1 \rho M_1 ^\dagger,0) \\
                &= M_0 \rho M_0 ^\dagger + \semanticsRho{\Omega}(f(\rho))\\
                &= M_0 \rho M_0^\dagger + \zeroOp\\
                &= M_0 \rho M_0^\dagger \\
                &= M_0 (f^0(\rho)) M_0 ^\dagger \\
                &= \sum_{k=0}^{0} M_0 (f^k(\rho)) M_0 ^\dagger .
            \end{align*}
            \item $n=n+1$:
            \begin{align*}
                &\phantom{=.}\semanticsRho{(\while')^{n+1}}(\rho,p) \\
                &= \semanticsRho{\ifstatement{S'; (\while')^n}{\skipbf}}(\rho,p)
            \end{align*} implies
            \begin{align*}
                & \phantom{=.}\semanticsRho{(\while')^{n+1}}(\rho,p) \\
                &= \semanticsRho{\ifstatement{S'; (\while')^n}{\skipbf}}(\rho,p) \\
                &= \semanticsRho{\skipbf } (M_0 \rho M_0^\dagger,0) + \semanticsRho{S'; (\while')^n} (M_1 \rho M_1 ^\dagger,0) \\
                &= M_0 \rho M_0^\dagger + \semanticsRho{(\while')^n}(f(\rho),0)\\
                &= M_0 \rho M_0^\dagger + \sum_{k=0}^{n-1} M_0 f^k(f(\rho)) M_0 ^\dagger\\
                &= M_0 \rho M_0^\dagger + \sum_{k=1}^{n} M_0 f^k(\rho) M_0 ^\dagger=  \sum_{k=0}^{n} M_0 f^k(\rho) M_0 ^\dagger.
            \end{align*}
            \end{itemize}

            Based on that,
            \begin{align*}
                &\phantom{=.}\semanticsRho{(\while')^{n+1}}(\rho,p) - \semanticsRho{(\while')^{n}}(\rho,p) \\
                &= \sum_{k=0}^{n} M_0 (f^k(\rho)) M_0 ^\dagger - \sum_{k=0}^{n-1} M_0 (f^k(\rho)) M_0 ^\dagger = M_0 (f^n(\rho)) M_0^\dagger
            \end{align*}
            is a positive operator which concludes the proof.
        \end{claimproof}

        This now implies the existence of the least upper bounds:
        \begin{itemize}
            \item Due to the monotonicity, $\{\semanticsRho{(\while')^n} (\rho,p)\}_{n=0}^\infty$ is an increasing sequence. Each $\semanticsRho{(\while')^n} (\rho,p)$ can be written as a sequential composition of measurements, that means $\semanticsRho{(\while')^n} (\rho,p)\in \density$ for all $n$. Then the least upper bound exists because for the set of partial density operators the Loewner order is an $\omega$-cpo.
            \item For all $n$ is $\semanticsErr{(\while')^n} (\rho,p) \leq tr(\rho)+p$ and $\semanticsErr{(\while')^n} (\rho,p) \in \mathbb{R}$. Thus, $\{ \semanticsErr{(\while')^n} (\rho,p) \}_{n=0}^\infty$ is a non-empty set of reals with an upper bound. By the least-upper-bound-property of reals, it follows that a least upper bound exists.
        \end{itemize}

        \begin{enumerate}
            \item $\semanticsRho{\while'}(\rho,p) = \bigvee_{n=0}^\infty \semanticsRho{(\while')^n} (\rho,p)$\\
            $= \bigvee_{n=0}^\infty \semanticsRho{(\while')^n} (\rho,q) = \semanticsRho{\while'}(\rho,q)$

            \item $p = \bigvee_{n=0}^\infty p \leq \bigvee_{n=0}^\infty \semanticsErr{(\while')^n} (\rho,p) $ \\
            $= \semanticsErr{\while'}(\rho,p)$

            \item
            \begin{align*}
                &\phantom{=.}\semanticsErr{\while'}(\rho,q+p) \\
                &= \bigvee_{n=0}^\infty \semanticsErr{(\while')^n} (\rho,q+p) \\
                &=\bigvee_{n=0}^\infty \left[\semanticsErr{(\while')^n} (\rho,q)+p\right]\\
                &= \bigvee_{n=0}^\infty [\semanticsErr{(\while')^n} (\rho,q)]+p\\
                &= \semanticsErr{\while'}(\rho,q) + p
            \end{align*}

            \item $\tilde{tr}(\semantics{\while'}(\rho,p))\leq \tilde{tr}(\rho,p)$ We show this by proving
            \begin{align*}
                \tilde{tr}(\semantics{(\while')^n}(\rho,p)) \leq \tilde{tr}(\rho,p)
            \end{align*} for all $n\geq 0$ by induction:
            \begin{itemize}
                \item For $n=0$:
                \begin{align*}
                    \tilde{tr}(\semantics{(\while')^0}(\rho,p)) = tr(\zeroOp) + p =0+ p\leq tr(\rho)+p = \tilde{tr}(\rho,p)
                \end{align*}
                \item For $n+1$:
                \begin{align*}
                    &\phantom{=.}\tilde{tr}\left(\semantics{(\while')^{n+1}}(\rho,p)\right)  \\
                    &= tr\left(\semanticsRho{\ifstatement{S'; (\while')^n}{\skipbf}}(\rho,p)\right) \\
                    &\phantom{=}+  \semanticsErr{\ifstatement{S'; (\while')^n}{\skipbf}}(\rho,p) \\
                    &= tr\left(\semanticsRho{\skipbf} (M_0 \rho M_0 ^\dagger ,0) + \semanticsRho{S';(\while')^n} (M_1 \rho M_1 ^\dagger ,0) \right) \\
                    &\phantom{=}+ \semanticsErr{\skipbf} (M_0 \rho M_0 ^\dagger ,0) + \semanticsErr{S';(\while')^n} (M_1 \rho M_1 ^\dagger ,0) + p \\
                    &= tr\left(M_0 \rho M_0 ^\dagger + \semanticsRho{(\while')^n} (\semantics{S} (M_1 \rho M_1 ^\dagger ,0))\right)\\
                    &\phantom{=} + \semanticsErr{(\while')^n} (\semantics{S'} (M_1 \rho M_1 ^\dagger ,0)) + p\\
                    &= tr(M_0 \rho M_0 ^\dagger) +  \tilde{tr}\left(\semantics{(\while')^n} (\semantics{S'} (M_1 \rho M_1 ^\dagger ,0))\right) + p\\
                    &\overset{(*)}{\leq} tr(M_0 \rho M_0 ^\dagger) + \tilde{tr }\left(\semantics{S'} (M_1 \rho M_1 ^\dagger ,0)\right) + 0 +  p\\
                    &\overset{(*)}{\leq}  tr(M_0 \rho M_0 ^\dagger) + tr(M_1 \rho M_1 ^\dagger)+ 0+  0 + p\\
                    &= tr([M_0 ^\dagger M_0+ M_1 ^\dagger M_1]  \rho )+ p\\
                    &= tr(\rho)+p = \tilde{tr}(\rho,p)
                \end{align*}
            \end{itemize}

            \item We have already shown that the least upper bounds exists. For all $(\rho,p)\in \densityNumberPairs$ is
            $\tilde{tr}(\semantics{\while'}(\rho,p))\leq \tilde{tr}(\rho,p)\leq 1$ and $\semanticsErr{\while'}(\rho,p) \geq p \geq 0$.

            As $\semanticsRho{\while'}(\rho,p) = \bigvee_{n=0}^\infty \semanticsRho{(\while')^n} (\rho,p)$ and each $\semanticsRho{(\while')^n} (\rho,p) \in \density$, the least upper bound is positive as well. Also $tr(\semanticsRho{\while'}(\rho,p)) \leq tr(\rho)$, thus $\semanticsRho{\while'}(\rho,p) \in \density$ and $\semantics{\while'}(\rho,p) \in \densityNumberPairs$.

            \item Linearity:
            We know that $\semantics{(\while')^n}$ is linear for each $n$ and increasing, thus $\bigvee_{n=0}^\infty \overline{\semantics{(\while')^n}}$ is linear as well by Lemma~\ref{lem:lupislinear}. Note that we use that $\semanticsRho{S'}(\rho,p)$ is independent of $p$ (property $1$).
        \end{enumerate}
    \end{itemize}
\end{proof}

In the following proof we show the equivalence of the operational and denotational semantics, i.e., the proof of Lemma~\ref{lem:operationalVsDenotational}.
\begin{proof}
    \begin{enumerate}
        \item The first claim consist of two parts:
        \begin{enumerate}
            \item $\semanticsRho{S}(\rho,0) = \sum_{\rho'} \prob{\operationalMC{\rho}{S}}{\eventually \config{\emptyProgram}{\rho'}} \cdot \rho'$
    \item $\semanticsErr{S}(\rho,0) = \prob{\operationalMC{\rho}{S}}{\eventually \errConfig} $
        \end{enumerate}

    We show both at the same time doing an induction over the structure of $S$:
    \begin{itemize}
        \item $S\equiv \skipbf$: $\config{\skipbf}{\rho} \overset{1}{\rightarrow} \config{\emptyProgram}{\rho} \overset{1}{\rightarrow} \termConfig \overset{1}{\rightarrow} \termConfig \overset{1}{\rightarrow}  \dots$ is the only possible run of $\operationalMC{\rho}{\skipbf}$.
        \begin{enumerate}
            \item $\semanticsRho{\skipbf}(\rho,0) =  \rho = \sum_{\rho'} \prob{\operationalMC{\rho}{\skipbf}}{\eventually \config{\emptyProgram}{\rho'}} \cdot \rho'$
        \item $\semanticsErr{\skipbf}(\rho,0) = 0 =  \prob{\operationalMC{\rho}{\skipbf}}{\eventually \errConfig}$
        \end{enumerate}

        \item $S \equiv \qzero$: We only show the case of $type(q) = \intType$, for $\boolType$ it is similar.\\
        $\config{\qzero}{\rho} \overset{1}{\rightarrow} \config{\emptyProgram}{\sum_{n\in \mathbb{Z}} \ket{0} \bra{n}_q \rho \ket{n} \bra{0}_q} \overset{1}{\rightarrow} \termConfig \overset{1}{\rightarrow} \termConfig \overset{1}{\rightarrow}  \dots$ is the only possible run of $\operationalMC{\rho}{\qzero}$.
        \begin{enumerate}
            \item $\semanticsRho{\qzero}(\rho,0) =  \sum_{n\in \mathbb{Z}} \ket{0} \bra{n}_q \rho \ket{n} \bra{0}_q = \sum_{\rho'} \prob{\operationalMC{\rho}{\qzero}}{\eventually \config{\emptyProgram}{\rho'}} \cdot \rho' $
            \item $\semanticsErr{\qzero}(\rho,0) = 0 =  \prob{\operationalMC{\rho}{\qzero}}{\eventually \errConfig}$
        \end{enumerate}

        \item $S \equiv \Uq$: $\config{\Uq}{\rho} \overset{1}{\rightarrow} \config{\emptyProgram}{U\rho U^\dagger} \overset{1}{\rightarrow} \termConfig \overset{1}{\rightarrow} \termConfig \overset{1}{\rightarrow}  \dots$ is the only possible run of $\operationalMC{\rho}{\Uq}$.

        \begin{enumerate}
            \item $\semanticsRho{\Uq}(\rho,0) =  U\rho U^\dagger = \sum_{\rho'} \prob{\operationalMC{\rho}{\Uq}}{\eventually \config{\emptyProgram}{\rho'}} \cdot \rho'$
        \item $\semanticsErr{\Uq}(\rho,0) = 0= \prob{\operationalMC{\rho}{\Uq}}{\eventually \errConfig}$
        \end{enumerate}

        \item $S \equiv \observe$: There are exactly two possible runs of $\operationalMC{\rho}{\observe}$ (assuming $0<tr(O\rho O^\dagger)<1$; otherwise only one of them): The successful run $\config{\observe}{\rho} \overset{tr(O \rho O ^\dagger)}{\rightarrow} \config{\emptyProgram}{\frac{O\rho O^\dagger}{tr(O\rho O^\dagger)}} \overset{1}{\rightarrow} \termConfig \overset{1}{\rightarrow} \termConfig \overset{1}{\rightarrow}  \dots$ and the non-successful run $\config{\observe}{\rho} \overset{1- tr(O \rho O ^\dagger)}{\rightarrow} \errConfig \overset{1}{\rightarrow} \termConfig \overset{1}{\rightarrow} \termConfig \overset{1}{\rightarrow}  \dots$ where the only one that satisfies $\eventually \errConfig$ is the second one with probability $1- tr(O \rho O ^\dagger)$.
        \begin{enumerate}
            \item $\semanticsRho{\observe}(\rho,0) =  O\rho O^\dagger = tr(O\rho O^\dagger) \cdot \frac{O\rho O^\dagger}{tr(O\rho O^\dagger)}$\\
            $= \sum_{\rho'} \prob{\operationalMC{\rho}{\observe}}{\eventually \config{\emptyProgram}{\rho'}} \cdot \rho'$ (if $tr(O\rho O^\dagger)>0$, else $\frac{O\rho O^\dagger}{tr(O\rho O^\dagger)}:= \zeroOp$)
        \item$\semanticsErr{\observe}(\rho,0) = tr(\rho)-tr(O\rho O^\dagger) = 1- tr(O \rho O ^\dagger)$\\
            $=  \prob{\operationalMC{\rho}{\observe}}{\eventually \errConfig}$
    \end{enumerate}

        \item $S \equiv \concat$: There are two disjoint possibilities for paths reaching $\errConfig$ in $\operationalMC{\rho}{\concat}$. Either during execution of $S_1$ (happens with probability $\prob{\operationalMC{\rho}{S_1}}{\eventually \errConfig}$), or during execution of $S_2$ after successfully executing $S_1$ without reaching $\errConfig$. Each path that does not reach $\errConfig$ during execution of $S_1$, but reaches $\errConfig$ at some point must go through a $\config{\emptyProgram;S_2}{\rho'}$ state.
        \begin{enumerate}
            \item It is \begin{align*}
            \semanticsRho{S_1}(\rho,0) &= \sum_{\rho'} \prob{\operationalMC{\rho}{S_1}}{\eventually \config{\emptyProgram}{\rho'}} \cdot \rho' = \sum_{\rho'} \prob{\operationalMC{\rho}{S_1;S_2}}{\eventually \config{\emptyProgram;S_2}{\rho'}} \cdot \rho',\\
            \semanticsRho{S_2}(\rho',0) &= \sum_{\rho''} \prob{\operationalMC{\rho'}{S_2}}{\eventually \config{\emptyProgram}{\rho''}} \cdot \rho'' = \sum_{\rho''} \prob{\operationalMC{\rho'}{\emptyProgram; S_2}}{\eventually \config{\emptyProgram}{\rho''}} \cdot \rho'' .
        \end{align*} Then
        \begin{align*}
            &\phantom{=.}\semanticsRho{\concat}(\rho,0)\\
            &=  \semanticsRho{S_2}\left(\sum_{\rho'} \prob{\operationalMC{\rho}{S_1}}{\eventually \config{\emptyProgram}{\rho'}} \cdot \rho' , 0\right) \\
            &= \sum_{\rho'} \prob{\operationalMC{\rho}{S_1}}{\eventually \config{\emptyProgram}{\rho'}} \cdot \semanticsRho{S_2}(\rho',0) \\
            &= \sum_{\rho'} \prob{\operationalMC{\rho}{S_1}}{\eventually \config{\emptyProgram}{\rho'}} \cdot \sum_{\rho''} \prob{\operationalMC{\rho'}{S_2}}{\eventually \config{\emptyProgram}{\rho''}} \cdot \rho''\\
            &= \sum_{\rho',\rho''} \left(\prob{\operationalMC{\rho}{S_1;S_2}}{\eventually \config{\emptyProgram;S_2}{\rho'}} \cdot \prob{\operationalMC{\rho'}{\emptyProgram; S_2}}{\eventually \config{\emptyProgram}{\rho''}} \cdot \rho''\right)\\
            &= \sum_{\rho''} \prob{\operationalMC{\rho}{S_1;S_2}}{\eventually \config{\emptyProgram}{\rho''}} \cdot \rho''
        \end{align*}
        \item \begin{align*}
            &\phantom{=.}\semanticsErr{\concat}(\rho,0) = \semanticsErr{S_2}(\semantics{S_1} (\rho,0 )) \\
            &= \semanticsErr{S_2}\left(\semanticsRho{S_1} (\rho,0 ),\semanticsErr{S_1}(\rho,0)\right)= \semanticsErr{S_2}(\semanticsRho{S_1} (\rho,p ),0)+ \semanticsErr{S_1}(\rho,0)\\
            &= \semanticsErr{S_2}\left(\sum_{\rho'} \prob{\operationalMC{\rho}{S_1}}{\eventually \config{\emptyProgram}{\rho'}} \cdot \rho' ,0\right)+ \prob{\operationalMC{\rho}{S_1}}{\eventually \errConfig}\\
            &= \sum_{\rho'} \left(\prob{\operationalMC{\rho}{S_1}}{\eventually \config{\emptyProgram}{\rho'}} \cdot \semanticsErr{S_2}(\rho',0) \right)+ \prob{\operationalMC{\rho}{S_1}}{\eventually \errConfig}\\
            &= \sum_{\rho'} \left(\prob{\operationalMC{\rho}{S_1;S_2}}{\eventually \config{\emptyProgram;S_2}{\rho'}} \cdot \prob{\operationalMC{\rho'}{\emptyProgram;S_2}}{\eventually \errConfig} \right)+ \prob{\operationalMC{\rho}{S_1}}{\eventually \errConfig}\\
            &=  \prob{\operationalMC{\rho}{\concat}}{\eventually \errConfig}.
        \end{align*}
        \end{enumerate}

        \item $S \equiv \measurePrime$:
        \begin{enumerate}
            \item \begin{align*}
                &\phantom{=.}\semanticsRho{\measurePrime}(\rho,0) = \sum_m \semanticsRho{S_m } (M_m \rho M_m ^\dagger ,0)\\
            &= \sum_m tr(M_m \rho M_m ^\dagger) \cdot \semanticsRho{S_m } \left(\frac{M_m \rho M_m ^\dagger}{tr(M_m \rho M_m ^\dagger)} ,\frac{0}{tr(M_m \rho M_m ^\dagger)}\right)\\
            &= \sum_m tr(M_m \rho M_m ^\dagger) \cdot \sum_{\rho'} \prob{\operationalMC{\nicefrac{M_m \rho M_m ^\dagger}{tr\left(M_m \rho M_m ^\dagger\right)}}{S_m}}{\eventually \config{\emptyProgram}{\rho'}} \cdot \rho'\\
            &= \sum_{\rho'} \sum_m Pr^{\operationalMC{\rho}{\measurePrime}}\left({\eventually \config{S_m}{\frac{M_m \rho M_m ^\dagger}{tr(M_m \rho M_m ^\dagger)}}}\right) \\
            &\phantom{=} \cdot  \prob{\operationalMC{\nicefrac{M_m \rho M_m ^\dagger}{tr(M_m \rho M_m ^\dagger)}}{S_m}}{\eventually \config{\emptyProgram}{\rho'}} \cdot \rho'\\
            &= \sum_{\rho'} \prob{\operationalMC{\rho}{\measurePrime}}{\eventually \config{\emptyProgram}{\rho'}} \cdot \rho'
        \end{align*}
        \item  \begin{align*}
            &\phantom{=.}\semanticsErr{\measurePrime}(\rho,0) \\
            &= \sum_m \semanticsErr{S_m}(M_m \rho M_m^\dagger,0)\\
            &= \sum_m tr(M_m \rho M_m^\dagger) \cdot \semanticsErr{S_m}\left(\frac{M_m \rho M_m^\dagger}{tr(M_m \rho M_m^\dagger)},\frac{0}{tr(M_m \rho M_m^\dagger)}\right)\\
            &= \sum_m tr(M_m \rho M_m^\dagger) \cdot \prob{\operationalMC{\nicefrac{M_m \rho M_m^\dagger}{tr(M_m \rho M_m^\dagger)}}{S_m}}{\eventually \errConfig}  \\
            &= \sum_m Pr^{\operationalMC{\rho}{\measurePrime}}\left({\eventually \config{S_m}{\frac{M_m \rho M_m ^\dagger}{tr(M_m \rho M_m ^\dagger)}}}\right) \\
            &\phantom{=}\cdot \prob{\operationalMC{\nicefrac{M_m \rho M_m^\dagger}{tr(M_m \rho M_m^\dagger)}}{S_m}}{\eventually \errConfig}  \\
            &=\prob{\operationalMC{\rho}{\measurePrime}}{\eventually \errConfig}.
        \end{align*}
        \end{enumerate}

        \item $S \equiv \while'$:
       \begin{enumerate}
            \item \begin{align*}
                &\phantom{=.}\semanticsRho{\while'}(\rho,0) \\
            &= \bigvee_{n=0}^\infty \semanticsRho{(\while)^n} (\rho,0) \\
            &= \bigvee_{n=0}^\infty \sum_{\rho'} \prob{\operationalMC{\rho}{(\while')^n}}{\eventually \config{\emptyProgram}{\rho'}} \cdot \rho'\\
            &= \sum_{\rho'} \prob{\operationalMC{\rho}{\while'}}{\eventually \config{\emptyProgram}{\rho'}} \cdot \rho'
        \end{align*}

        To show that the last step holds, we have to show that (1)
        \begin{align*}
            &\sum_{\rho'} \prob{\operationalMC{\rho}{\while'}}{\eventually \config{\emptyProgram}{\rho'}} \cdot \rho' \\
            \sqsubseteq &\bigvee_{n=0}^\infty \sum_{\rho'} \prob{\operationalMC{\rho}{(\while')^n}}{\eventually \config{\emptyProgram}{\rho'}} \cdot \rho'
        \end{align*}
        and (2)
        \begin{align*}
            &\sum_{\rho'} \prob{\operationalMC{\rho}{\while'}}{\eventually \config{\emptyProgram}{\rho'}} \cdot \rho'\\ \sqsupseteq &\bigvee_{n=0}^\infty \sum_{\rho'} \prob{\operationalMC{\rho}{(\while')^n}}{\eventually \config{\emptyProgram}{\rho'}} \cdot \rho'.
        \end{align*}

        We start by showing (1):

        Fix a $\rho'$. We show this direction by showing that every path of $\eventually \config{\emptyProgram}{\rho'}$ starting in $\config{\while'}{\rho}$ is also included in set of paths $\eventually \config{\emptyProgram}{\rho'}$ from $\config{(\while')^m}{\rho}$ for $m$ being the number of loop iterations until $\config{\emptyProgram}{\rho'}$ is reached which must exists because each of those paths is finite. Then it is also included in the set of paths with $n\geq m$. This means, if we take the maximum $m'$ of all $m$, every path from the left-hand side is included in this set. The probability for every path is also the same in both sides, i.e., $\prob{\operationalMC{\rho}{\while'}}{\eventually \config{\emptyProgram}{\rho'}}=  \prob{\operationalMC{\rho}{(\while')^{m'}}}{\eventually \config{\emptyProgram}{\rho'}} $.
        As each probability is positive and each $\rho'$ is positive as well,
        \begin{align*}
            &\sum_{\rho'} \prob{\operationalMC{\rho}{\while'}}{\eventually \config{\emptyProgram}{\rho'}} \cdot \rho' \\
            \sqsubseteq& \sum_{\rho'} \prob{\operationalMC{\rho}{(\while')^{m'}}}{\eventually \config{\emptyProgram}{\rho'}} \cdot \rho'\\
            \sqsubseteq &  \bigvee_{n=0}^\infty \sum_{\rho'} \prob{\operationalMC{\rho}{(\while')^n}}{\eventually \config{\emptyProgram}{\rho'}} \cdot \rho'
        \end{align*} holds.

        Now we show that (2) holds as well:
        Let $n\in \mathbb{N}$ be arbitrary but fixed. Every path starting in $\config{(\while')^n}{\rho}$ has $3$ options:
        \begin{enumerate}
            \item it terminates successfully (going through a $\config{\emptyProgram}{\sigma}$ state),
            \item it violates a observation (going through $\errConfig$) or
            \item it will be aborted because of reaching the max of $n$ loop iterations. That formally means that the path will have an infinite loop at a state $\config{(\while')^0}{\sigma}= \config{\textbf{while } M_{trivial}[\Bar{q}]=1 \textbf{ do } \skipbf}{\sigma}$ for some $\sigma \in \densityFull$ where $M_{trivial}=\{M_0=0, M_1=I\}$ by definition.
        \end{enumerate}
         In the second and last case, the path does not reach $\config{\emptyProgram}{\rho'}$ and we don't have to consider them for the sum. Similar to the first case, each path of the form $\eventually \config{\emptyProgram}{\rho'}$ from $\config{(\while')^n}{\rho}$ is included in the set of paths of the form $\eventually \config{\emptyProgram}{\rho'}$ from $\config{\while'}{\rho}$. That means
        \begin{align*}
            &\sum_{\rho'} \prob{\operationalMC{\rho}{\while'}}{\eventually \config{\emptyProgram}{\rho'}} \cdot \rho'
            \\\sqsupseteq & \sum_{\rho'} \prob{\operationalMC{\rho}{(\while')^{m'}}}{\eventually \config{\emptyProgram}{\rho'}} \cdot \rho'
        \end{align*}
        for every $n$, so it is a upper bound w.r.t. the Loewner order $\sqsubseteq$ and because $\bigvee_{n=0}^\infty$ is the least upper bound
        \begin{align*}
            &\sum_{\rho'} \prob{\operationalMC{\rho}{\while'}}{\eventually \config{\emptyProgram}{\rho'}} \cdot \rho'
            \\ \sqsupseteq &\bigvee_{n=0}^\infty \sum_{\rho'} \prob{\operationalMC{\rho}{(\while')^{m'}}}{\eventually \config{\emptyProgram}{\rho'}} \cdot \rho'
        \end{align*}
        holds.

        \item \begin{align*}
            &\phantom{=.}\semanticsErr{\while'}(\rho,0) \\
            &= \bigvee_{n=0}^\infty \semanticsErr{(\while')^n} (\rho,0)\\
            &= \bigvee_{n=0}^\infty \prob{\operationalMC{\rho}{(\while')^n}}{\eventually \errConfig}\\
            &= \prob{\operationalMC{\rho}{\while'}}{\eventually \errConfig}.
        \end{align*}
        This holds because
         \begin{align*}
            \bigvee_{n=0}^\infty \prob{\operationalMC{\rho}{(\while')^n}}{\eventually \errConfig}= \prob{\operationalMC{\rho}{\while'}}{\eventually \errConfig}.
         \end{align*}
         To show this, we first show that the left side is upper bounded by the right side (i.e., $ \leq$):
         All paths starting in $\config{(\while')^n}{\rho}$ will either go through $\errConfig$ or $\config{\emptyProgram}{\sigma}$ and then end in $\termConfig$ or will be aborted because $n$ iterations of the loop are reached but the loop guard is still true. Aborting means that the path will be stuck in an infinite loop. Paths of this form will never reach $\errConfig$ so we don't need to consider their probability. As paths of the right side starting in $\config{\while'}{\rho}$ will never be aborted and  all paths of the left side are included (for each $n$), the right side is an upper bound for each $n$ and thus also of the least upper bound.

         Now we have to show the other direction, i.e., the left side is an upper bound for the right side ($\geq$):
         Every path starting in $\config{\while'}{\rho}$ that reaches $\errConfig$ at some point has to reach $\errConfig$ after a finite number $m$ of loop iterations. That means, for each path there exists a $m$ such that the path is also included in the set of paths starting in $\config{(\while')^m}{\rho}$. Taking the maximum $m'$ of all these $m$, gives a Markov chain with initial state $\config{(\while')^{m'}}{\rho}$ that includes all paths of $\eventually \errConfig$ starting in $\config{\while'}{\rho}$ with the same probability.
         This implies that also the least upper bound of all values $n$ (which includes $m'$) is greater or equal to the right side, which concludes the proof that they are equal.
        \end{enumerate}
    \end{itemize}
    \item Every path that eventually reaches $\termConfig$ passes through either a $\config{\emptyProgram}{\rho'}$ or a $\errConfig$ state.
    Then
    \begin{align*}
        \prob{\operationalMC{\rho}{S}}{\eventually \termConfig} &= \sum_{\rho'} \prob{\operationalMC{\rho}{S}}{\eventually \config{\emptyProgram}{\rho'}} + \prob{\operationalMC{\rho}{S}}{\eventually \errConfig} \\
        &= \sum_{\rho'} \prob{\operationalMC{\rho}{S}}{\eventually \config{\emptyProgram}{\rho'}} \cdot tr(\rho')+ \prob{\operationalMC{\rho}{S}}{\eventually \errConfig} \\
        &= tr(\sum_{\rho'} \prob{\operationalMC{\rho}{S}}{\eventually \config{\emptyProgram}{\rho'}} \cdot \rho') + \prob{\operationalMC{\rho}{S}}{\eventually \errConfig} \\
        &= tr(\semanticsRho{S}(\rho,0))+\semanticsErr{S}(\rho,0)
    \end{align*}
    because $tr(\rho')=1$.
\end{enumerate}

\end{proof}

Here we provide the proof of Lemma~\ref{lem:bounded}:
\begin{proof}
    \begin{enumerate}
        \item Direct consequence of $tr(\semanticsRho{S}(\rho,p))+ \semanticsErr{S} (\rho ,p)\leq tr(\rho)+p$ and $p \leq\semanticsErr{S}(\rho,p)$.
        \item Direct consequence of $\semanticsErr{S}(\rho,q+p) = \semanticsErr{S}(\rho,q) + p$.
        \item Let $\overline{\semanticsRho{S}}: \traceclass \to \traceclass$ be the linear extension of $\semanticsRho{S}$ (existence follows from $\semanticsRho{S}$ being linear).
        Every $\rho\in \traceclass$ can be written as a linear combination of $4$ positive operators, i.e., $\rho= \sum_{i=1}^4 \lambda_i \rho_i$ with $\rho_i \sqsubseteq 0$ and $\norm{\rho_i}\leq \norm{\rho}$ for all $i$ \cite[\text{Lemma Trace\_Class.trace\_class\_decomp\_4pos}]{isabelleproof}. Then
        \begin{align*}
            &\phantom{=}\norm{\overline{\semanticsRho{S}}(\rho,0)} = \norm{\overline{\semanticsRho{S}}(\sum_{i=1}^4 \lambda_i \rho_i)} \leq \sum_{i=1}^4 \abs{\lambda_i}\cdot \norm{ \overline{\semanticsRho{S}}(\rho_i)} =\sum_{i=1}^4 \abs{\lambda_i}\cdot \norm{ \semanticsRho{S}(\rho_i)} \\
            &= \sum_{i=1}^4 \abs{\lambda_i}\cdot tr( \semanticsRho{S}(\rho_i)) \leq \sum_{i=1}^4 \abs{\lambda_i}\cdot tr( \rho_i) = \sum_{i=1}^4 \abs{\lambda_i}\cdot \norm{ \rho_i}\leq \sum_{i=1}^4 \abs{\lambda_i}\cdot \norm{ \rho} \leq 4 \cdot \norm{ \rho}
        \end{align*} which shows that $\overline{\semanticsRho{S}}$ is bounded.
    \end{enumerate}
\end{proof}

\subsection{Proofs Concerning Weakest Preconditions}
\label{app:wp}

Now we show the existence of the weakest precondition, Lemma~\ref{lem:schroedinger}:
\begin{proof}
    We use the Schrödinger-Heisenberg dual \cite[Lemma 35 (ii),(iii),(vi)]{heisenbergdualityUnruh}:
    For every bounded linear, positive, trace-reducing map $f: T(\mathcal{H}) \rightarrow T( \mathcal{H})$ exists a bounded linear, positive, subunital map $g: B(\mathcal{H}) \rightarrow B(\mathcal{H})$ with $tr(f(t)a) = tr(g(a) t)$ for all $t\in T(\mathcal{H}), a \in B(\mathcal{H})$.

    We define $f:\traceclass \to \traceclass$ as $f(\rho) = \overline{\semanticsRho{S}}(\rho,0)$ for all $\rho\in \traceclass$. By Proposition~\ref{pro:allClaims} (and Lemma~\ref{lem:bounded}) and because $f$ is defined as the extension of $\semanticsRho{S}$, $f$ bounded linear, positive and trace-reducing.
    Then there exists a bounded linear, positive, subunital function $\Bar{g}:\bounded \to \bounded$ with $tr(f(t)a) = tr(g(a) t)$ for all $t\in T(\mathcal{H}), a \in B(\mathcal{H})$ by the Schrödinger-Heisenberg duality.
    We define $g:\predicate \to \predicate$ as the restriction of $\Bar{g}$ to $\predicate$. $g$ will later turn out to be the $qwp\llbracket S \rrbracket$. $g$ is well-defined, in the sense that for each predicate $P \in \predicate$, $g(P) \in \predicate$ holds. As $g$ is positive and $P$ is positive, $g(P)$ is also positive ($\zeroOp \sqsubseteq g(P)$). Also $P\sqsubseteq \identityOp$ implies $\identityOp - P$ to be positive. Then $g(\identityOp-P)$ is positive as well, by linearity also $g(\identityOp)- g(P)$, that means $g(P) \sqsubseteq g(\identityOp)$. $g$ is subunital, which means $g(\identityOp) \sqsubseteq \identityOp$. Thus $ g(P) \sqsubseteq \identityOp$. In total, we have showed $\zeroOp \sqsubseteq g(P) \sqsubseteq \identityOp$ which means $g(P)\in \predicate$.

    To show that $qwp\llbracket S \rrbracket$ exist and $qwp\llbracket S \rrbracket = g$, we have to show that for every $P\in \predicate$
    \begin{enumerate}
        \item $\models_{tot}\hoare{g(P)}{S}{P}$ because forall $\rho\in \density$ is $tr(g(P)\rho) = tr(P f(\rho)) = tr(P \semanticsRho{S}(\rho,0))$
        \item For any $Q\in \predicate$ with $\models_{tot}\hoare{Q}{S}{P}$ is $Q\sqsubseteq g(P)$. This is because for all $\rho \in \density$ \begin{equation*}
            tr(Q\rho) \leq tr(P\semanticsRho{S}(\rho,0)) = tr(P f(\rho)) = tr(g(P)\rho)
        \end{equation*}
        which means $Q\sqsubseteq g(P)$ \cite{floydHoareLogic}.
    \end{enumerate}
    Now we have shown that $qwp\llbracket S \rrbracket$ exists, is bounded linear, positive and subunital and satisfies $tr(\qwp{S}{P}\rho) = tr(P \semanticsRho{S}(\rho,0))$ for all $\rho \in \density, P \in \predicate$.

    For showing the second part of the statement, assume there is a function $h:\predicate \to \predicate$ that satisfies $tr(h(P)\rho) = tr(P\semanticsRho{S}(\rho,0))$ for each $\rho \in \density, P\in \predicate$, which would imply $tr(h(P)\rho)=tr(\qwp{S}{P}\rho)$. We can rewrite $\rho=\psi \psi^\dagger$ and get $tr(h(P)\psi \psi^\dagger)=tr(\psi^\dagger h(P)\psi) = \psi^\dagger h(P)\psi$ and analogous $tr(\qwp{S}{P}\psi \psi^\dagger) = \psi^\dagger \qwp{S}{P}\psi$ which then means that $h(P)= \qwp{S}{P}$ by \cite[Chapter II, Proposition 2.15]{conway1994} for each $P$, thus they are equal.

\end{proof}

Now we also prove the existence of weakest liberal precondition, Lemma~\ref{lem:wlpexistence}:
\begin{proof}
    To show the existence of $qwlp$, we point back to the existence of $qwp$ and write $qwlp$ in terms of $qwp$.

    Again, we use the Schrödinger-Heisenberg dual \cite[Lemma 35 (ii), (iii), (vi)]{heisenbergdualityUnruh} that says: For every bounded linear, positive, trace-reducing map $f: \traceclass \to T(\mathbb{C})$ there exists a bounded linear, positive, subunital map $\Bar{g}:B(\mathbb{C}) \to \bounded$ with $tr(f(\rho) c) = tr(\Bar{g}(c)\rho)$ for all $\rho \in T(\mathcal{H}),c\in B(\mathbb{C})$.
    Note that $B(\mathbb{C})$, $T(\mathbb{C})$ and $\mathbb{C}$ are basically the same set, so we will use $\mathbb{C}$ in the following.

    We define $f:\traceclass \to \mathbb{C}$ as $f(\rho)=\overline{\semanticsErr{S}}(\rho,0)$. From Proposition~\ref{pro:allClaims} and the fact that $f$ is defined as the extension of $\semanticsErr{S}$ follows that $f$ is positive, bounded linear and trace-reducing. So there exists a bounded linear, positive, subunital map $\Bar{g}:\mathbb{C} \to \bounded$. We restrict $\Bar{g}$ to $g:\mathbb{R}_{[0,1]} \to \predicate$. We now show that $g$ is well-defined, i.e., $g(c)\in \predicate$ for $c \in \mathbb{R}_{[0,1]} $.
    First, $\zeroOp \sqsubseteq g(c)$ for all $c \in \mathbb{R}_{[0,1]} $ because $\Bar{g}$ and $c$ are positive, which means that $\Bar{g}(c)$ must be positive, so $\zeroOp \sqsubseteq g(c)$.
    Second, $\Bar{g}$ is subunital, so $\Bar{g}(1) \sqsubseteq \identityOp$ and thus also $g(1)\sqsubseteq \identityOp$. Also $\Bar{g}(1-c)$ is positive for $c\in \mathbb{R}_{[0,1]}$, then $g(1-c)$ is positive and due to the linearity also $g(1)-g(c)$. That means $g(c)\sqsubseteq g(1)$ and together with $g(1)\sqsubseteq \identityOp$ is $g(c) \sqsubseteq \identityOp$ and thus a predicate.


    Now we state the following claim $(*)$:
    \begin{align*}
        \qwlp{S}{P} = \qwp{S}{P} - \qwp{S}{\identityOp} + \identityOp - g(1)
    \end{align*}
    The existence of $\qwp{S}{P}$ and $\qwp{S}{\identityOp}$ follows from Lemma~\ref{lem:schroedinger}. The existence of $\identityOp$ is clear and we have just proven the existence of $g$.

    We have to show that $\qwlp{S}{P} \in \predicate$ for $P\in \predicate$. To do so, we first show $\qwlp{S}{P}\sqsubseteq \identityOp$:
    \begin{align*}
         &\phantom{\overset{qwp\llbracket S \rrbracket \text{ is pos.}}{\Leftrightarrow}.}P \sqsubseteq  \identityOp \\
        &\overset{\phantom{qwp\llbracket S \rrbracket \text{ is pos.}}}{\Leftrightarrow} \zeroOp \sqsubseteq  \identityOp-P \\
        &\overset{qwp\llbracket S \rrbracket \text{ is pos.}}{\Leftrightarrow} \zeroOp \sqsubseteq  \qwp{S}{\identityOp-P} \\
        &\overset{qwp\llbracket S \rrbracket \text{ is lin.}\phantom{.}}{\Leftrightarrow} \zeroOp \sqsubseteq  \qwp{S}{\identityOp} -\qwp{S}{P} \\
        &\overset{\phantom{qwp\llbracket S \rrbracket \text{ is pos.}}}{\Leftrightarrow} \qwp{S}{P} \sqsubseteq  \qwp{S}{\identityOp} \\
        &\overset{\phantom{xs.}g \text{ is pos.}\phantom{xs.}}{\Rightarrow} \qwp{S}{P} \sqsubseteq  \qwp{S}{\identityOp} + g(1)\\
        &\overset{\phantom{qwp\llbracket S \rrbracket \text{ is pos.}}}{\Leftrightarrow} \qwp{S}{P} - \qwp{S}{\identityOp} - g(1) \sqsubseteq \zeroOp\\
        &\overset{\phantom{qwp\llbracket S \rrbracket \text{ is pos.}}}{\Leftrightarrow} \qwp{S}{P} - \qwp{S}{\identityOp} + \identityOp - g(1)\sqsubseteq \identityOp\\
        &\overset{\phantom{qwp\llbracket S \rrbracket \text{ is pos.}}}{\Leftrightarrow}\qwlp{S}{P} \sqsubseteq \identityOp
    \end{align*}
    Then $\zeroOp\sqsubseteq \qwlp{S}{P}$ because
    \begin{align*}
        &\phantom{\overset{qwp\llbracket S \rrbracket\text{ is pos.}}{\Rightarrow} .} \forall \rho \in \density: tr(\semanticsRho{S}(\rho,0)) + \semanticsErr{S}(\rho,0) \leq tr(\rho)\\
        &\overset{\phantom{qwp\llbracket S \rrbracket\text{ is pos.}}}{\Leftrightarrow} \forall \rho \in \density: tr(\identityOp\semanticsRho{S}(\rho,0)) + tr(1 \cdot \semanticsErr{S}(\rho,0)) \leq tr(\identityOp \rho)\\
        &\overset{\phantom{xs.s}\text{duality}\phantom{xss}}{\Leftrightarrow} \forall \rho \in \density: tr(\qwp{S}{\identityOp}\rho) + tr(g(1)\rho) \leq tr(\identityOp \rho)\\
        &\overset{\phantom{qwp\llbracket S \rrbracket\text{ is pos.}}}{\Leftrightarrow}\qwp{S}{\identityOp} + g(1) \sqsubseteq \identityOp \\
        &\overset{\phantom{qwp\llbracket S \rrbracket\text{ is pos.}}}{\Leftrightarrow} \zeroOp\sqsubseteq \identityOp - \qwp{S}{\identityOp}  - g(1) \\
        &\overset{qwp\llbracket S \rrbracket\text{ is pos.}}{\Rightarrow} \zeroOp\sqsubseteq \qwp{S}{P} - \qwp{S}{\identityOp} + \identityOp - g(1) \\
        &\overset{\phantom{qwp\llbracket S \rrbracket\text{ is pos.}}}{\Leftrightarrow} \zeroOp\sqsubseteq \qwlp{S}{P}
    \end{align*}
    which together with $\qwlp{S}{P}\sqsubseteq \identityOp$ means $\qwlp{S}{P}\in \predicate$.

    We now show that this is actually weakest liberal precondition by showing for every $P\in \predicate$:
    \begin{enumerate}
        \item $\models_{par}\hoare{\qwlp{S}{P}}{S}{P}$ holds because \begin{align*}
            &\phantom{\overset{\Bar{g}(1) = g(1)}{=}.}tr(\qwlp{S}{P}\rho) \\
            &\overset{\phantom{..a.}(*)\phantom{.a..}}{=} tr ((\qwp{S}{P} - \qwp{S}{\identityOp} + \identityOp - g(1)) \rho) \\
            &\overset{\Bar{g}(1) = g(1)}{=} tr(\qwp{S}{P}\rho) - tr(\qwp{S}{\identityOp}\rho) + tr(\identityOp\rho) - tr(\Bar{g}(1)\rho)\\
            &\overset{\phantom{aa.}\text{dual.}\phantom{a.}}{=} tr(\qwp{S}{P}\rho) - tr(\qwp{S}{\identityOp}\rho) + tr(\identityOp\rho) - tr(1 \cdot f(\rho))\\
            &\overset{\phantom{..a.}\text{def }f\phantom{.a}}{=} tr(P\semantics{S}(\rho)) - tr(\identityOp \semantics{S}(\rho)) + tr(\rho) - tr(1\cdot \semanticsErr{S}(\rho,0))\\
            &\overset{\phantom{\Bar{g}(1) = g(1)}}{=} tr(P \semantics{S}(\rho)) - tr(\semantics{S}(\rho)) + tr(\rho) - \semanticsErr{S}(\rho,0).
        \end{align*}
        \item For any $Q\in \predicate$ with $\models_{par}\hoare{Q}{S}{P}$ is $Q\sqsubseteq \qwlp{S}{P}$. This is because for all $\rho \in \density$ \begin{align*}
            tr(Q\rho) \leq tr(P \semantics{S}(\rho)) - tr(\semantics{S}(\rho)) + tr(\rho) - \semanticsErr{S}(\rho,0) = tr(\qwlp{S}{P}\rho)
        \end{align*}
        so $Q \sqsubseteq \qwlp{S}{P}$.
    \end{enumerate}

    It remains to show that $\qwlp{S}{P}$ is bounded linear and subunital. Bounded linear follows because $qwp\llbracket S \rrbracket, \identityOp$ and $g$ are bounded linear. Subunital follows because
    \begin{align*}
        \qwlp{S}{\identityOp} = \qwp{S}{\identityOp} - \qwp{S}{\identityOp} + \identityOp - g(1) = \identityOp - g(1) \sqsubseteq \identityOp.
    \end{align*}

    For showing the second part of the statement, assume there is a function $h:\predicate \to \predicate$ that satisfies $tr(h(P)\rho) = tr(P \semantics{S}(\rho)) - tr(\semantics{S}(\rho)) + tr(\rho) - \semanticsErr{S}(\rho,0)$ for each $\rho \in \density, P\in \predicate$, which would imply $tr(h(P)\rho)=tr(\qwlp{S}{P}\rho)$. We can rewrite $\rho=\psi \psi^\dagger$ and get $tr(h(P)\psi \psi^\dagger)=tr(\psi^\dagger h(P)\psi) = \psi^\dagger h(P)\psi$ and analogous $tr(\qwp{S}{P}\psi \psi^\dagger) = \psi^\dagger \qwlp{S}{P}\psi$ which then means that $h(P)= \qwlp{S}{P}$ by \cite[Chapter II, Proposition 2.15]{conway1994} for each $P$, thus they are equal.
\end{proof}

Proof of Proposition~\ref{prop:healthWP}:
\begin{proof}Let $P,Q \in \predicate$.
    \begin{itemize}
        \item Bounded linear: Follows directly from Lemma~\ref{lem:schroedinger}
        \item Subunital: Follows directly from Lemma~\ref{lem:schroedinger}.
        \item Monotonic:
        \begin{align*}
            &\phantom{\Rightarrow.}P\sqsubseteq Q \\
            &\Rightarrow \forall \rho \in \density: tr(P \semanticsRho{S}(\rho,0)) \leq tr(Q \semanticsRho{S}(\rho,0))\\
            &\Rightarrow  \forall \rho \in \density: tr(\qwp{S}{P} \rho) \leq tr(\qwp{S}{Q}\rho)\\
            &\Rightarrow  \qwp{S}{P} \sqsubseteq \qwp{S}{Q}
        \end{align*}
        \item Order-continuous:
        \begin{align*}
            \forall \rho \in \density: &\phantom{=}tr(\qwp{S}{\bigvee_{i=0}^\infty P_i} \rho)= tr((\bigvee_{i=0}^\infty P_i) \semanticsRho{S}(\rho,0))\\
            &= \bigvee_{i=0}^\infty tr(P_i \semanticsRho{S}(\rho,0))= \bigvee_{i=0}^\infty tr(\qwp{S}{P_i}\rho)=  tr(\bigvee_{i=0}^\infty \qwp{S}{P_i}\rho)
        \end{align*}
        that implies $\qwp{S}{\bigvee_{i=0}^\infty P_i} = \bigvee_{i=0}^\infty \qwp{S}{P_i}$.
    \end{itemize}
\end{proof}

Proof of Proposition~\ref{prop:healthWLP}:
\begin{proof}
    Let $P,Q \in \predicate$.
    \begin{itemize}
        \item Affine:
        We have to show that $f:\predicate \to \predicate$ with $f(P)=\qwlp{S}{P} - \qwlp{S}{\zeroOp}$ is linear. To do so, we use $\qwlp{S}{P} = \qwp{S}{P} - \qwp{S}{\identityOp} + \identityOp - g(1)$ as in the proof of Lemma~\ref{lem:wlpexistence}, then $f(P) = \qwp{S}{P}-\qwp{S}{\zeroOp}$. By the linearity of $qwp\llbracket S \rrbracket$ is $\qwp{S}{\zeroOp} = \zeroOp$, thus $f(P) = \qwp{S}{P}$ and $f$ is linear.
        \item Subunital: Follows directly from Lemma~\ref{lem:wlpexistence}.
        \item Monotonic:
        \begin{align*}
            & \phantom{\Rightarrow.}P\sqsubseteq Q \\
            &\Rightarrow \forall \rho \in \density: tr(P \semanticsRho{S}(\rho,0)) \leq tr(Q \semanticsRho{S}(\rho,0))\\
            &\Rightarrow \forall \rho \in \density: tr(P \semanticsRho{S}(\rho,0)) + tr(\rho) - tr(\semanticsRho{S}(\rho,0)) - \semanticsErr{S}(\rho,0)\\
            &\phantom{\Rightarrow \forall \rho \in \density:}\leq tr(Q \semanticsRho{S}(\rho,0)) + tr(\rho) - tr(\semanticsRho{S}(\rho,0)) - \semanticsErr{S}(\rho,0)\\
            &\Rightarrow  \forall \rho \in \density: tr(\qwlp{S}{P} \rho) \leq tr(\qwlp{S}{Q}\rho)\\
            &\Rightarrow  \qwlp{S}{P} \sqsubseteq \qwlp{S}{Q}
        \end{align*}
        \item Order-continuous:
        \begin{align*}
            \forall \rho \in \density: & \phantom{=}tr(\qwlp{S}{\bigvee_{i=0}^\infty P_i} \rho)\\
            &= tr((\bigvee_{i=0}^\infty P_i) \semanticsRho{S}(\rho,0)) + tr(\rho) - tr(\semanticsRho{S}(\rho,0)) - \semanticsErr{S}(\rho,0)\\
            &= \bigvee_{i=0}^\infty tr(P_i \semanticsRho{S}(\rho,0)) + tr(\rho) - tr(\semanticsRho{S}(\rho,0)) - \semanticsErr{S}(\rho,0)\\
            &= \bigvee_{i=0}^\infty tr(\qwlp{S}{P_i}\rho)=  tr(\bigvee_{i=0}^\infty \qwlp{S}{P_i}\rho)
        \end{align*}
        that implies $\qwlp{S}{\bigvee_{i=0}^\infty P_i} = \bigvee_{i=0}^\infty \qwlp{S}{P_i}$.
    \end{itemize}
\end{proof}

Proof of equivalence between weakest preconditions and explicit representation, Proposition~\ref{prop:wpDef}:
\begin{proof}
    We prove that $tr(\qwp{S}{P}\rho) = tr(P\semanticsRho{S}(\rho,0))$ holds for each $\rho \in \density$. 

    We only show the observe-case the other cases are done to \cite{floydHoareLogic}. Note that, \cite{floydHoareLogic} does not talk about convergence of sums. In our case we mean convergence with respect to the SOT and it follows from \cite[Lem. 29]{heisenbergdualityUnruh} that the supremum and limit coincides in the SOT and if the norm of each element is upper bounded, then the limit exists. Note that we only show convergence of the measurement case, the assignment is just a special case of it.
    That means, we have to show that $\norm{\sum_{m \in M'} M_m^\dagger \qwp{S'_m}{P} M_m}$ is bounded for every finite set $M'\subseteq M$. It is $\sum_{m \in M'} M_m^\dagger \qwp{S'_m}{P} M_m \sqsubseteq \sum_{m \in M'} M_m^\dagger M_m \sqsubseteq \identityOp$ because $\zeroOp \sqsubseteq \qwp{S'_m}{P}\sqsubseteq \identityOp$. Then
    $\norm{\sum_{m \in M'} M_m^\dagger \qwp{S'_m}{P} M_m} \leq \norm{\identityOp} = 1$
    and \begin{align*}
        \bigvee_{\text{finite }M' \subseteq M} \sum_{m \in M'} M_m^\dagger \qwp{S'_m}{P} M_m = \sum_{m \in M} M_m^\dagger \qwp{S'_m}{P} M_m.
    \end{align*}
    For the observe case is $
            tr(\qwp{\observe}{P}\rho) = tr(O^\dagger P O \rho) = tr(P O \rho O^\dagger ) $\\
            $= tr(P \semanticsRho{\observe}(\rho,0))$
        for all $\rho \in \density$.

Note that it would not be necessary to do an induction over the structure. Showing the trace equality for each statement separately would also work due to Lemma~\ref{lem:schroedinger}.
\end{proof}

Quite similar to the total correctness case, we also prove that the explicit representation of weakest liberal preconditions given in Proposition~\ref{prop:wlpDef} is correct:
\begin{proof}
    We prove $tr(\qwlp{S}{P}\rho) = tr(P\semanticsRho{S}(\rho,0)) + tr(\rho)-tr(\semanticsRho{S}(\rho,0))- \semanticsErr{S}(\rho,0)$ for each $\rho \in \density$ for each program $S$. 
    Convergence of the sums follows by the same arguments as in the proof of Proposition~\ref{prop:wpDef}.

    Again, it is not needed to do an induction due to Lemma~\ref{lem:wlpexistence}, we can show it for each statement separately:
    \begin{itemize}
        \item $S\equiv \skipbf \mid \qzero \mid \Uq \mid \observe$: It is $\qwlp{S}{P} = \qwp{S}{P}$ and thus
        \begin{align*}
            tr(\qwlp{S}{P}\rho) &= tr(P\semanticsRho{S}(\rho,0))\\
            &= tr(P\semanticsRho{S}(\rho,0)) + tr(\rho)-tr(\semanticsRho{S}(\rho,0))- \semanticsErr{S}(\rho,0)
        \end{align*}
        for all $\rho \in \density$ because because $tr(\rho)=tr(\semanticsRho{S}(\rho,0))$ is proven in Proposition~\ref{pro:allClaims} and $\semanticsErr{S}(\rho,0) = 0$ for all cases except for the observe statement. For $S \equiv \observe$ is $tr(\rho)= tr(\semanticsRho{S}(\rho,0))+ \semanticsErr{S}(\rho,0)$ because $\semanticsErr{S}(\rho,0)=0+tr(\rho) - tr(O\rho O^\dagger)$ and $\semanticsRho{S}(\rho,0) = O\rho O^\dagger$.

        \item $S\equiv \concat$: For all $\rho \in \density$
        \begin{align*}
            &\phantom{=.} tr(\qwlp{\concat}{P}\rho) \\
            &= tr(\qwlp{S_1}{\qwlp{S_2}{P}} \rho)\\
            &= tr(\qwlp{S_2}{P} \semanticsRho{S_1}(\rho,0)) + tr(\rho)-tr(\semanticsRho{S_1}(\rho,0))- \semanticsErr{S_1}(\rho,0)\\
            &= tr(P \semanticsRho{S_2}(\semanticsRho{S_1}(\rho,0),0)) + tr(\semanticsRho{S_1}(\rho,0))-tr(\semanticsRho{S_2}(\semanticsRho{S_1}(\rho,0),0))\\
            &\phantom{=} - \semanticsErr{S_2}(\semanticsRho{S_1}(\rho,0),0) + tr(\rho)-tr(\semanticsRho{S_1}(\rho,0))- \semanticsErr{S_1}(\rho,0)\\
            &= tr(P \semanticsRho{\concat}(\rho,0)) - tr(\semanticsRho{\concat}(\rho,0))- \\
            & \phantom{=}[\semanticsErr{S_2}(\semanticsRho{S_1}(\rho,0),\semanticsErr{S_1}(\rho,0)) - \semanticsErr{S_1}(\rho,0)] + tr(\rho) - \semanticsErr{S_1}(\rho,0)\\
            &= tr(P \semanticsRho{\concat}(\rho,0)) - tr(\semanticsRho{\concat}(\rho,0))\\
            &\phantom{=}-\semanticsErr{S_2}(\semantics{S_1}(\rho,0)) + \semanticsErr{S_1}(\rho,0) + tr(\rho) - \semanticsErr{S_1}(\rho,0)\\
            &= tr(P \semanticsRho{\concat}(\rho,0)) + tr(\rho)-tr(\semanticsRho{\concat}(\rho,0))- \semanticsErr{\concat}(\rho,0)
        \end{align*}

        \item $S\equiv \measurePrime $:
        For all $\rho \in \density$
        \begin{align*}
            &\phantom{=.} tr(\qwlp{\measurePrime}{P}\rho) \\
            &= tr(\sum_m M_m^\dagger (\qwlp{S_m}{P}) M_m \rho) \\
            &= \sum_m tr(M_m^\dagger (\qwlp{S_m}{P}) M_m \rho) \\
            &= \sum_m tr((\qwlp{S_m}{P}) M_m \rho M_m^\dagger) \\
            &= \sum_m tr(P \semanticsRho{S_m}(M_m \rho M_m^\dagger,0)) + tr(M_m \rho M_m^\dagger)\\
            &\phantom{=}-tr(\semanticsRho{S_m}(M_m \rho M_m^\dagger,0))- \semanticsErr{S_m}(M_m \rho M_m^\dagger,0) \\
            &= \sum_m tr(P \semanticsRho{S_m}(M_m \rho M_m^\dagger,0)) + tr(\rho M_m^\dagger M_m )\\
            &\phantom{=}-tr(\semanticsRho{S_m}(M_m \rho M_m^\dagger,0))- \semanticsErr{S_m}(M_m \rho M_m^\dagger,0) \\
            &= tr(P \sum_m \semanticsRho{S_m}(M_m \rho M_m^\dagger,0)) + tr(\rho \sum_m M_m^\dagger M_m )\\
            & \phantom{=}- tr(\sum_m \semanticsRho{S_m}(M_m \rho M_m^\dagger,0))- \sum_m \semanticsErr{S_m}(M_m \rho M_m^\dagger,0) - 0\\
            &= tr(P \semanticsRho{\measurePrime}(\rho,0)) + tr(\rho)\\
            &\phantom{=}-tr(\semanticsRho{\measurePrime}(\rho,0))- \semanticsErr{\measurePrime}(\rho,0)
        \end{align*}

        \item $S\equiv \while' $:
        The goal is to show
        \begin{align*}
            &\phantom{=.}tr(\qwlp{\while'}{P}\rho)\\
            &= tr(P \semanticsRho{\while'}(\rho,0)) + tr(\rho)\\
            &\phantom{=}-tr(\semanticsRho{\while'}(\rho,0))- \semanticsErr{\while'}(\rho,0)
        \end{align*}
        for all $\rho \in \density$. First, we show
        \begin{align*}
            &tr(P_n \rho) = tr(P \semanticsRho{(\while')^n}(\rho,0))+ tr(\rho)\\
            &\phantom{tr(P_n \rho) = }-tr(\semanticsRho{(\while')^n}(\rho,0))\\
            & \phantom{tr(P_n \rho) = }- \semanticsErr{(\while')^n}(\rho,0)
        \end{align*}
        for all $\rho \in \density$ by induction:
        \begin{itemize}
            \item $n=0$ then
            \begin{align*}
                &\phantom{=.} tr(P_0 \rho) = tr(\identityOp \rho) = tr(\rho) \\
                &= 0 + tr(\rho) - 0 - 0= tr(P \Omega(\rho,0)) + tr(\rho) - tr(\Omega(\rho,0)) - \semanticsErr{\Omega}(\rho,0) \\
                &= tr(P \semanticsRho{(\while')^0}(\rho,0))+ tr(\rho) \\
                &\phantom{=} - tr(\semanticsRho{(\while')^0}(\rho,0))- \semanticsErr{(\while')^0}(\rho,0)
            \end{align*}
            \item $n=n+1$ then
            \begin{align*}
                &\phantom{=.}tr(P_{n+1} \rho)\\
                &=  tr (([M_0^\dagger P  M_0] + [M_1^\dagger(\qwlp{S'}{P_n}) M_1]) \rho)\\
                &= tr([M_0^\dagger P  M_0]\rho  + [M_1^\dagger(\qwlp{S'}{P_n}) M_1] \rho)\\
                &=  tr(M_0^\dagger P  M_0 \rho) + tr(M_1^\dagger(\qwlp{S'}{P_n}) M_1 \rho)\\
                &=  tr( P  M_0 \rho M_0^\dagger) + tr((\qwlp{S'}{P_n}) M_1 \rho M_1^\dagger)\\
                &= tr( P  M_0 \rho M_0^\dagger) + tr(P_n \semanticsRho{S'}(M_1 \rho M_1^\dagger,0)) + tr(M_1 \rho M_1^\dagger)-tr(\semanticsRho{S'}(M_1 \rho M_1^\dagger,0))\\
                &\phantom{=}- \semanticsErr{S'}(M_1 \rho M_1^\dagger,0)\\
                &= tr( P  M_0 \rho M_0^\dagger) + tr(P \semanticsRho{(\while')^n}(\semanticsRho{S'}(M_1 \rho M_1^\dagger,0),0))\\
                &\phantom{=}+ tr(\semanticsRho{S'}(M_1 \rho M_1^\dagger,0)) -tr(\semanticsRho{(\while')^n}(\semanticsRho{S'}(M_1 \rho M_1^\dagger,0),0))\\
                &\phantom{=}- \semanticsErr{(\while')^n}(\semanticsRho{S'}(M_1 \rho M_1^\dagger,0),0)+ tr(M_1 \rho M_1^\dagger)\\
                &\phantom{=}-tr(\semanticsRho{S'}(M_1 \rho M_1^\dagger,0))- \semanticsErr{S'}(M_1 \rho M_1^\dagger,0)\\
                &= tr(P M_0 \rho M_0^\dagger) + tr(P \semanticsRho{S';(\while')^n}(M_1 \rho M_1^\dagger,0))\\
                &\phantom{=}+tr(M_1 \rho M_1^\dagger)-tr(\semanticsRho{(\while')^n}(\semantics{S'}(M_1 \rho M_1^\dagger,0)))\\
                &\phantom{=}-\semanticsErr{(\while')^{n}}(\semanticsRho{S'}(M_1\rho M_1^\dagger,0),0) - \semanticsErr{S'}(M_1\rho M_1^\dagger,0)
                \\
                &= tr(P M_0 \rho M_0^\dagger) + tr(P \semanticsRho{S';(\while')^n}(M_1 \rho M_1^\dagger,0))+ tr(\rho)\\
                &\phantom{=}+tr(M_1 \rho M_1^\dagger)-tr(\rho) - tr(\semanticsRho{(\while')^n}(\semantics{S'}(M_1 \rho M_1^\dagger,0)))\\
                &\phantom{=}-\semanticsErr{(\while')^{n}}(\semanticsRho{S'}(M_1\rho M_1^\dagger,0),0) - \semanticsErr{S'}(M_1\rho M_1^\dagger,0)
                \\
                &= tr(P M_0 \rho M_0^\dagger) + tr(P \semanticsRho{S';(\while')^n}(M_1 \rho M_1^\dagger,0))+ tr(\rho)\\
                &\phantom{=}-tr(M_0 \rho M_0^\dagger)- tr(\semanticsRho{(\while')^n}(\semantics{S'}(M_1 \rho M_1^\dagger,0)))\\
                &\phantom{=}-\semanticsErr{(\while')^{n}}(\semanticsRho{S'}(M_1\rho M_1^\dagger,0),0) - \semanticsErr{S'}(M_1\rho M_1^\dagger,0)\\
                &= tr(P M_0 \rho M_0^\dagger) + tr(P \semanticsRho{S';(\while')^n}(M_1 \rho M_1^\dagger,0))+ tr(\rho)\\
                &\phantom{=}-tr(M_0 \rho M_0^\dagger) -tr(\semanticsRho{S';(\while')^n}(M_1 \rho M_1^\dagger,0))\\
                &\phantom{=}- \semanticsErr{(\while')^{n}}(\semantics{S'}(M_1\rho M_1^\dagger,0))\\
                &= tr(P[M_0 \rho M_0^\dagger + \semanticsRho{S';(\while')^n}(M_1 \rho M_1^\dagger,0)])+ tr(\rho)\\
                &\phantom{=}-tr(M_0 \rho M_0^\dagger + \semanticsRho{S';(\while')^n}(M_1 \rho M_1^\dagger,0))\\
                & \phantom{=}- 0 -\semanticsErr{S';(\while')^{n}}(M_1\rho M_1^\dagger,0)\\
                &= tr(P \semanticsRho{(\while')^{n+1}}(\rho,0))+ tr(\rho)\\
                &\phantom{=}-tr(\semanticsRho{(\while')^{n+1}}(\rho,0))\\
                &\phantom{=}- \semanticsErr{(\while')^{n+1}}(\rho,0)
            \end{align*}
        \end{itemize}
        Thus
        \begin{align*}
            &\phantom{=.}tr(P_n \rho)\\
            &= tr(P \semanticsRho{(\while')^n}(\rho,0))+ tr(\rho)\\
            &\phantom{=}-tr(\semanticsRho{(\while')^n}(\rho,0))- \semanticsErr{(\while')^n}(\rho,0)\\
            &= tr(\rho)- tr((\identityOp-P)\semanticsRho{(\while')^n}(\rho,0)) \\
            &\phantom{=}- \semanticsErr{(\while')^n}(\rho,0)
        \end{align*}
        It is $P \sqsubseteq \identityOp$, thus $\identityOp-P$ is a positive operator.
        We know that $\{\semanticsRho{(\while)^{n}}(\rho,0)\}_{n=0}^\infty$ is an increasing sequence, $P$ is positive, so
        \begin{align*}
            &\phantom{=.}tr(P_n \rho) \\
            &= tr(\rho)- tr((\identityOp-P)\semanticsRho{(\while')^n}(\rho,0)) \\
            &\phantom{=}- \semanticsErr{(\while')^n}(\rho,0)\\
            &\geq  tr(\rho)- tr((\identityOp-P)\semanticsRho{(\while')^{n+1}}(\rho,0)) \\
            &\phantom{=}- \semanticsErr{(\while')^{n+1}}(\rho,0)\\
            &= tr(P_{n+1}\rho).
        \end{align*}
        Since $\rho$ is arbitrary, we have $P_n \sqsupseteq P_{n+1}$ and $\{\identityOp-P_n\}_{n=0}^\infty$ is an increasing sequence and $\sqsubseteq$ an $\omega$-cpo, so $\bigvee_{n=0}^\infty \identityOp-P_n$ exits. Then
        \begin{equation*}
            \bigwedge_{n=0}^\infty P_n = \identityOp - \bigvee_{n=0}^\infty (\identityOp-P_n)
        \end{equation*}
        exists too.

        By continuity of the trace operator we have
        \begin{align*}
            &\phantom{=.}tr(\qwlp{\while'}{P}\rho)= tr((\bigwedge_{n=0}^\infty P_n) \rho) = tr((\identityOp - \bigvee_{n=0}^\infty (\identityOp-P_n)) \rho)\\
            &= tr(\rho - \bigvee_{n=0}^\infty (\identityOp-P_n)\rho)=tr(\rho)- \bigvee_{n=0}^\infty tr((\identityOp-P_n)\rho)\\
            &= \bigwedge_{n=0}^\infty [tr(\rho)-tr((\identityOp-P_n)\rho)]= \bigwedge_{n=0}^\infty tr(P_n\rho)\\
            &= \bigwedge_{n=0}^\infty [tr(\rho)- tr((\identityOp-P)\semanticsRho{(\while')^n}(\rho,0)) \\
            &\phantom{=}- \semanticsErr{(\while')^n}(\rho,0)]\\
            &= tr(\rho)- \bigvee_{n=0}^\infty [tr((\identityOp-P)\semanticsRho{(\while')^n}(\rho,0)) \\
            &\phantom{=}+ \semanticsErr{(\while')^n}(\rho,0)]\\
            &= tr(\rho)- tr((\identityOp-P)\bigvee_{n=0}^\infty \semanticsRho{(\while')^n}(\rho,0)) \\
            &\phantom{=}- \bigvee_{n=0}^\infty\semanticsErr{(\while')^n}(\rho,0)\\
            &= tr(\rho)- tr((\identityOp-P)\semanticsRho{\while'}(\rho,0)) \\
            &\phantom{=}- \semanticsErr{\while'}(\rho,0)\\
            &= tr(P \semanticsRho{\while'}(\rho,0)) + tr(\rho)-tr(\semanticsRho{\while'}(\rho,0))\\
            &\phantom{=}- \semanticsErr{\while'}(\rho,0)
        \end{align*}
    \end{itemize}
\end{proof}

Here we give the explicit representation of $qcwp$:
\begin{lemma}
    \label{lem:explicitCWP}
    An explicit representation of the \textit{quantum conditional weakest precondition} transformer $qcwp\llbracket P \rrbracket: \predicate^2\to \predicate^2$ is given by:
\begin{itemize}
    \item $\qcwp{\skipbf}{P,Q}=(P,Q)$
    \item $\qcwp{\qzero}{P,Q}=
    \begin{cases}
        \Bigl( \ket{0}_q\bra{0}P\ket{0}_q \bra{0} + \ket{1}_q\bra{0} P \ket{0}_q \bra{1}, &\\
        \ket{0}_q\bra{0}Q\ket{0}_q \bra{0} + \ket{1}_q\bra{0} Q \ket{0}_q \bra{1} \Bigr)
         &,\text{ if }type(q)=\boolType \\
         \Bigl( \sum_{n\in \mathbb{Z}} \ket{n}_q \bra{0} P \ket{0}_q\bra{n},&\\
          \sum_{n\in \mathbb{Z}} \ket{n}_q \bra{0} Q \ket{0}_q \bra{n}\Bigr)
        & ,\text{ if }type(q)=\intType
    \end{cases}$
    \item $\qcwp{\Uq}{P,Q}= U^\dagger \multdot (P,Q) \multdot U$
    \item $\qcwp{ \observe}{P,Q} = O^\dagger \multdot (P,Q) \multdot O$
    \item $\qcwp{\concat}{P,Q}=\qcwp{S_1}{\qcwp{S_2}{P,Q} }  $
    \item $\qcwp{\measure}{P,Q} = \sum_m M_m^\dagger\multdot (\qcwp{S_m}{P,Q})\multdot M_m$
    \item $\qcwp{\while'}{P,Q}= \bigvee_{n=0}^\infty P_n$ with
    \begin{align*}
        P_0  &= (\zeroOp,\identityOp)\\
        P_{n+1} &= [M_0^\dagger \multdot (P,Q) \multdot M_0] \plusdot [M_1^\dagger \multdot\qcwp{S'}{P_n} \multdot M_1]
    \end{align*}
    and $\bigvee_{n=0}^\infty$ denoting the least upper bound according to $\unlhd$.
\end{itemize}
\end{lemma}

Proof of Proposition~\ref{prop:healthcwp}:
\begin{proof}
    Let $P,P',Q,Q' \in \predicate$ be predicates.
    \begin{itemize}
        \item Linear Interpretation:
        Let $a,b \in \mathbb{R}_{\geq 0}$ and $aP+bP' \in \predicate$.

         Let $\rho \in \density$ be arbitrary but fixed. Assume $tr(\qwlp{S}{\identityOp}\rho) \neq 0$ (otherwise both sides are undefined). Then
        \begin{align*}
            &\phantom{=.}\fracTr(\qcwp{S}{aP+bP',\identityOp} \multdot\rho) = \frac{tr(\qwp{S}{aP+bP'} \rho)}{tr(\qwlp{S}{\identityOp}\rho)} \\
            &= \frac{tr((a \cdot \qwp{S}{P} + b\cdot \qwp{S}{P'}) \rho)}{tr(\qwlp{S}{\identityOp}\rho)} \\
            &= \frac{tr(a \cdot \qwp{S}{P} \rho) + tr(b\cdot \qwp{S}{P'} \rho)}{tr(\qwlp{S}{\identityOp}\rho)} \\
            &=\frac{a \cdot tr( \qwp{S}{P} \rho)}{tr(\qwlp{S}{\identityOp}\rho)} + \frac{ b\cdot tr(\qwp{S}{P'} \rho)}{tr(\qwlp{S}{\identityOp}\rho)} \\
            &= a \cdot \fracTr(\qcwp{S}{P,\identityOp}\multdot \rho) + b\cdot\fracTr(\qcwp{S}{P',\identityOp}\multdot \rho)
        \end{align*}

        \item Affine: We have to show that $f:\predicate^2 \to \predicate^2$ with $f(P,Q)=\qcwp{S}{P,Q} - \qcwp{S}{\zeroOp,\zeroOp}$ is linear.
        It is \begin{align*}
            f(P,Q) &= \qcwp{S}{P,Q} - \qcwp{S}{\zeroOp,\zeroOp} \\
            &= (\qwp{S}{P},\qwlp{S}{Q})-(\qwp{S}{\zeroOp},{\qwlp{S}{\zeroOp}})\\
            &= (\qwp{S}{P}-\qwp{S}{\zeroOp},\qwlp{S}{Q}-\qwlp{S}{\zeroOp})\\
            &= (\qwp{S}{P},\qwlp{S}{Q}-\qwlp{S}{\zeroOp})
        \end{align*}
        We know that $qwp\llbracket S \rrbracket$ is linear and $qwlp\llbracket S \rrbracket$ affine, thus $f':\predicate \to \predicate$ with $f'(Q)=\qwlp{S}{Q}-\qwlp{S}{\zeroOp}$ is linear. So we can conclude that
        $f(P,Q)=(\qwp{S}{P}, f'(Q))$ is linear as well and thus $qcwp\llbracket S \rrbracket $ affine.
        \item Monotonic: Let $(P,P') \unlhd (Q,Q')$. Then $P \sqsubseteq Q$ and $P' \sqsupseteq Q'$. As $qwp$ and $qwlp$ are monotonic,
        \begin{align*}
            \qwp{S}{P} &\sqsubseteq \qwp{S}{Q},\\
            \qwlp{S}{P'} &\sqsupseteq \qwlp{S}{Q'}
        \end{align*}
        for every program $S$. That implies $
            \qcwp{S}{P,P'} \unlhd \qcwp{S}{Q,Q'}$.
        \item Order-continuous:
        First of all, $\unlhd$ is also an $\omega$-cpo on $\predicate^2$, thus  $\bigvee_{i=0}^\infty (P_i,Q_i)$ exists for any increasing chain $\{(P_i,Q_i)\}_{i \in \mathbb{N}}$.
    As $qwp$ and $qwlp$ are monotonic, the existence of $\bigvee_{i=0}^\infty \qcwp{S}{P_i,Q_i}$ follows directly.

    It remains to show that
        \begin{align*}
            \qcwp{S}{\bigvee_{i=0}^\infty (P_i,Q_i)} = \bigvee_{i=0}^\infty \qcwp{S}{P_i,Q_i}.
        \end{align*}

    First we have to show that $\bigwedge_{i=0}^\infty \qwlp{S}{Q_i} = \qwlp{S}{\bigwedge_{i=0}^\infty Q_i}$ holds. This follows from the continuity of the trace:

    \begin{align*}
        tr(\bigwedge_{i=0}^\infty \qwlp{S}{Q_i}) &=\bigwedge_{i=0}^\infty  tr(\qwlp{S}{Q_i})\\
        &= \bigwedge_{i=0}^\infty (tr(Q_i \semanticsRho{S}(\rho,0)) + tr(\rho) - tr(\semanticsRho{S}(\rho,0)) - \semanticsErr{S}(\rho,0))\\
        &=  tr(\bigwedge_{i=0}^\infty Q_i \semanticsRho{S}(\rho,0)) + tr(\rho) - tr(\semanticsRho{S}(\rho,0)) - \semanticsErr{S}(\rho,0)\\
        &= tr(\qwlp{S}{\bigwedge_{i=0}^\infty Q_i})
    \end{align*}

    Then
    \begin{align*}
        \qcwp{S}{\bigvee_{i=0}^\infty (P_i,Q_i)} &= \qcwp{S}{\bigvee_{i=0}^\infty P_i ,\bigwedge_{i=0}^\infty Q_i}= (\qwp{S}{\bigvee_{i=0}^\infty P_i }, \qwlp{S}{\bigwedge_{i=0}^\infty Q_i})\\
        &= (\bigvee_{i=0}^\infty \qwp{S}{ P_i }, \bigwedge_{i=0}^\infty \qwlp{S}{Q_i})= \bigvee_{i=0}^\infty \qcwp{S}{(P_i,Q_i)}.
    \end{align*}
    \end{itemize}
\end{proof}

Proof of Proposition~\ref{prop:healthcwlp}:
\begin{proof}
    Let $P,P',Q,Q' \in \predicate$ be predicates.
    \begin{itemize}
        \item Affine: We have to show that $f:\predicate^2 \to \predicate^2$ with $f(P,Q)=\qcwlp{S}{P,Q} - \qcwlp{S}{\zeroOp,\identityOp}$ is linear.
        It is \begin{align*}
            f(P,Q) &= \qcwlp{S}{P,Q} - \qcwlp{S}{\zeroOp,\zeroOp} \\
            &= (\qwlp{S}{P},\qwlp{S}{Q})-(\qwlp{S}{\zeroOp},\qwlp{S}{\zeroOp})\\
            &= (\qwlp{S}{P}-\qwlp{S}{\zeroOp},\qwlp{S}{Q}-\qwlp{S}{\zeroOp})
        \end{align*}
        We know that $qwlp\llbracket S \rrbracket$ is affine, thus $f':\predicate \to \predicate$ with $f'(Q)=\qwlp{S}{Q}-\qwlp{S}{\zeroOp}$ linear. So we can conclude that
        $f(P,Q)=(f'(P), f'(Q))$ is linear as well and thus $qcwlp\llbracket S \rrbracket $ affine.
        \item Monotonic: Let $(P,Q) \phantom{.}\dot \unlhd \phantom{.} (P',Q')$. Then $P \sqsubseteq P'$ and $Q \sqsubseteq Q'$. As $qwlp$ is monotonic,
        \begin{align*}
            \qwlp{S}{P} &\sqsubseteq \qwp{S}{P'},\\
            \qwlp{S}{Q} &\sqsubseteq \qwlp{S}{Q'}
        \end{align*}
        for every program $S$. That implies $
            \qcwlp{S}{P,Q} \phantom{.}\dot \unlhd \phantom{.} \qcwlp{S}{P',Q'}$.

        \item Order-continuous: $\dot \unlhd$ is an $\omega$-cpo on $\predicate^2$, that means $\bigvee_{i=0}^\infty (P_i,Q_i)$ exists for increasing sequences $\{P_i\}_{i \in \mathbb{N}}$ and $\{Q_i\}_{i\in \mathbb{N}}$ of predicates.
        Then
        \begin{align*}
            \qcwlp{S}{\bigvee_{i=0}^\infty (P_i,Q_i)} &= (\qwlp{S}{\bigvee_{i=0}^\infty P_i }, \qwlp{S}{\bigvee_{i=0}^\infty  Q_i}) \\
            &= (\bigvee_{i=0}^\infty \qwlp{S}{ P_i },\bigvee_{i=0}^\infty  \qwlp{S}{ Q_i}) \\
            &= \bigvee_{i=0}^\infty \qcwlp{S}{ P_i,Q_i}.
        \end{align*}
    \end{itemize}
\end{proof}

Proof of Lemma~\ref{lem:ErWpEqual}:
\begin{proof}
For all $\rho \in \densityFull$ is
    \begin{align*}
        tr(\qwlp{S}{\identityOp} \rho)\overset{\text{Lem. }\ref{lem:wlpexistence}}{=}& tr(\identityOp\semanticsRho{S}(\rho,0)) + tr(\rho)-tr(\semanticsRho{S}(\rho,0))- \semanticsErr{S}(\rho,0)\\
        \overset{\phantom{Lem. 71}}{=}& 1 - \semanticsErr{S}(\rho,0)\\
        \overset{\text{Lem. }\ref{lem:operationalVsDenotational}}{=}& 1-  \prob{\operationalMC{\rho}{S}}{ \eventually \errConfig}\\
        \overset{\phantom{Lem. 71}}{=}& \prob{\rewardMC{\rho}{\identityOp}{S}}{\neg \eventually \errConfig}.
    \end{align*}
    Also
    \begin{align*}
        ER^{\rewardMC{\rho}{P}{S}} (\eventually \termConfig) \overset{\phantom{Lem. 711}}{=}& \sum_{\rho'} \prob{\rewardMC{\rho}{P}{S}}{\eventually \config{\emptyProgram}{\rho'}}\cdot tr(P\rho') \\
        \overset{\phantom{Lem. 711}}{=}& tr(P \sum_{\rho'} \prob{\rewardMC{\rho}{P}{S}}{\eventually \config{\emptyProgram}{\rho'}}\cdot \rho')\\
        \overset{\text{Lem. }\ref{lem:operationalVsDenotational}\phantom{a}}{=} & tr(P \semanticsRho{S}(\rho,0))\\
        \overset{\text{Lem. }\ref{lem:schroedinger}}{=}& tr(\qwp{S}{P} \rho)
    \end{align*}
    and then
    \begin{align*}
        LER^{\rewardMC{\rho}{P}{S}} (\eventually \termConfig) \overset{\phantom{Lem. 71}}{=}& ER^{\rewardMC{\rho}{P}{S}} (\eventually \termConfig) + \prob{\rewardMC{\rho}{P}{S}}{\neg \eventually \termConfig}\\
        \overset{\phantom{Lem. 71}}{=}&  tr(\qwp{S}{P} \rho) + \bigl( 1- \prob{\rewardMC{\rho}{P}{S}}{ \eventually \termConfig}\bigr)\\
        \overset{\text{Lem. }\ref{lem:operationalVsDenotational}}{=}& tr(\qwp{S}{P} \rho) + tr(\rho) - tr(\semanticsRho{S}(\rho,0)) - \semanticsErr{S}(\rho,0)\\
        \overset{\text{Lem. }\ref{lem:wlpexistence}}{=}& tr(\qwlp{S}{P}\rho).
    \end{align*}
\end{proof}
\subsection{Details of the Example}
\label{app:ex}

\begin{figure}
    \begin{tikzpicture}
      \node (0) at (0,0){$\config{Hq;\dots}{\rho}$};
      \node (1) at (0,-1){$\config{Hp;\dots}{\rho_1}$};
      \node (2) at (0,-2){$\config{observe(q\otimes p, \identityOp_4 - \ket{11}\bra{11});\dots}{\rho_2}$};
      \node (3) at (2,-3.5){$\config{Hr;\dots}{\rho_3}$};
      \node (4) at (2,-4.5){$\config{\emptyProgram}{\rho_4}$};
      \node (5) at (-2,-4.5){$\errConfig$};

      \node (sink) at (0, -6){$\termConfig$};

      \path[->] (0) edge node [left,color=gray]{$1$} (1);
      \path[->] (1) edge node [left,color=gray]{$1$} (2);
      \path[->] (2) edge node [right,color=gray]{$tr(\rho_3')$} (3);
      \path[->] (3) edge node [left,color=gray]{$1$} (4);
      \path[->] (2) edge node [left,color=gray]{$1-tr(\rho_3')$} (5);

      \path[->] (4) edge node [above,color=gray] {$1$} (sink);
      \path[->] (5) edge node [above,color=gray] {$1$} (sink);
      \path[->,loop below] (sink) edge node [right,color=gray] {$1$} (sink);
    \end{tikzpicture}
    \caption{Operational semantics of the Quantum Fast-Dice-Roller with \begin{itemize}
      \item $\rho_1 =(H \otimes \identityOp_2 \otimes \identityOp_2) \rho (H \otimes \identityOp_2 \otimes \identityOp_2)^\dagger$
      \item $\rho_2 = (H \otimes H \otimes \identityOp_2) \rho (H \otimes H \otimes \identityOp_2)^\dagger$
      \item $\rho_3 = \frac{1}{tr(\rho_3')}\rho_3'$ with $\rho_3' = ((\identityOp_4 - \ket{11}\bra{11}) \otimes \identityOp_2)(H \otimes H \otimes \identityOp_2) \rho (H \otimes H \otimes \identityOp_2)^\dagger ((\identityOp_4 - \ket{11}\bra{11}) \otimes \identityOp_2)^\dagger$
      \item $\rho_4 = \frac{1}{tr(\rho_3')}((\identityOp_4 - \ket{11}\bra{11}) \otimes \identityOp_2)(H \otimes H \otimes H) \rho (H \otimes H \otimes H)^\dagger ((\identityOp_4 - \ket{11}\bra{11}) \otimes \identityOp_2)^\dagger$
    \end{itemize}
    }
    \label{fig:opSem}
  \end{figure}

  The operational semantics of the Quantum Fast-Dice-Roller can be found in Figure~\ref{fig:opSem}.

  To compute $\qwp{S}{P}$ and $\qwlp{S}{\identityOp}$, we use the rules from Proposition~\ref{prop:wpDef} and~\ref{prop:wlpDef} and obtain
\begin{align*}
  \qwp{S}{P} &= (\identityOpVar{4} \otimes H)^\dagger ((\identityOpVar{4}-\ketbra{11}{11})\otimes \identityOp)^\dagger (\identityOp \otimes H \otimes \identityOp)^\dagger (H \otimes \identityOpVar{4})^\dagger P \\
  & \phantom{=}(H \otimes \identityOpVar{4}) (\identityOp \otimes H \otimes \identityOp) ((\identityOpVar{4}-\ketbra{11}{11})\otimes \identityOp)(\identityOpVar{4} \otimes H)\\
  &= \left(\begin{smallmatrix}
    3/4 & \makebox[0.2cm]{ }0\makebox[0.2cm]{ } & 1/4 & \makebox[0.2cm]{ }0\makebox[0.2cm]{ } & 1/4 & \makebox[0.2cm]{ }0\makebox[0.2cm]{ } & -1/4 & \makebox[0.2cm]{ }0\makebox[0.2cm]{ }\\
    0 & 0 & 0 & 0 & 0 & 0 & 0 & 0\\
    1/4 & 0 & 1/12 & 0 & 1/12& 0 & -1/12 & 0 \\
    0 & 0 & 0 & 0 & 0 & 0 & 0 & 0\\
    1/4 & 0 & 1/12 & 0 & 1/12& 0 & -1/12 & 0 \\
    0 & 0 & 0 & 0 & 0 & 0 & 0 & 0\\
    -1/4 & 0 & -1/12 & 0 & -1/12& 0 & 1/12 & 0 \\
    0 & 0 & 0 & 0 & 0 & 0 & 0 & 0\\
  \end{smallmatrix}\right)\\
  \qwlp{S}{\identityOp} &=(\identityOpVar{4} \otimes H)^\dagger ((\identityOpVar{4}-\ketbra{11}{11})\otimes \identityOp)^\dagger (\identityOp \otimes H \otimes \identityOp)^\dagger (H \otimes \identityOpVar{4})^\dagger \identityOp \\
  & \phantom{=}(H \otimes \identityOpVar{4}) (\identityOp \otimes H \otimes \identityOp) ((\identityOpVar{4}-\ketbra{11}{11})\otimes \identityOp)(\identityOpVar{4} \otimes H)\\
  &= \left(\begin{smallmatrix}
    3/4 & 0 & 1/4 & 0 & 1/4 & 0 & -1/4 & 0 \\
    0 & 3/4 & 0 & 1/4 & 0 & 1/4 & 0 & -1/4\\
    1/4 & 0 & 3/4 & 0 & - 1/4 & 0 & 1/4 & 0 \\
    0 & 1/4 & 0 & 3/4 & 0 & -1/4 & 0 & 1/4\\
    1/4 & 0 & -1/4 & 0 & 3/4& 0 & 1/4 & 0 \\
    0 & 1/4 & 0 & -1/4 & 0 & 3/4 & 0 & 1/4\\
    -1/4 & 0 & 1/4 & 0 & 1/4& 0 & 3/4 & 0 \\
    0 & -1/4 & 0 & 1/4 & 0 & 1/4 & 0 & 3/4\\
  \end{smallmatrix}\right)
\end{align*}

\end{document}